\documentclass[12pt,a4paper]{article}
\usepackage[DIV=13]{typearea}
\usepackage[onehalfspacing]{setspace}
\usepackage{layout}

\usepackage[english]{babel}
\usepackage[ansinew]{inputenc}
\usepackage{amssymb,amsmath,bm,array,dsfont,graphicx,natbib}
\usepackage[colorlinks,citecolor=black,linkcolor=black,urlcolor=black]{hyperref}
\usepackage[hang,tight]{subfigure}
\usepackage{authblk}
\usepackage{framed}
\usepackage{xcolor}
\usepackage{bbm}
\usepackage{longtable}
\usepackage{rotating}
\usepackage{amsthm}

\colorlet{shadecolor}{gray!25}

\newcommand{\newoperator}[3]{\newcommand*{#1}{\mathop{#2}#3}}
\newcommand{\renewoperator}[3]{\renewcommand*{#1}{\mathop{#2}#3}}
\newcommand{\mA}{A}

\newcommand{\mB}{B}

\newcommand{\mC}{C}

\newcommand{\mD}{D}

\newcommand{\mE}{E}
\newcommand{\ve}{e}
\newcommand{\mF}{F}

\newcommand{\mG}{G}

\newcommand{\mH}{H}

\newcommand{\mI}{I}

\newcommand{\mJ}{J}

\newcommand{\mK}{K}

\newcommand{\mM}{M}

\newcommand{\mN}{N}

\newcommand{\mP}{P}

\newcommand{\mQ}{Q}

\newcommand{\mR}{R}

\newcommand{\mS}{S}
\newcommand{\vs}{s}
\newcommand{\mT}{T}

\newcommand{\mU}{U}
\newcommand{\vu}{u}
\newcommand{\mV}{V}
\newcommand{\vv}{v}
\newcommand{\mW}{W}
\newcommand{\vw}{w}
\newcommand{\mX}{X}
\newcommand{\vx}{x}
\newcommand{\mY}{Y}
\newcommand{\vy}{y}
\newcommand{\mZ}{Z}

\newcommand{\valpha}{\alpha}
\newcommand{\vbeta}{\beta}

\newcommand{\vdelta}{\delta}

\newcommand{\vzeta}{\zeta}
\newcommand{\veta}{\eta}
\newcommand{\vtheta}{\theta}

\newcommand{\vmu}{\mu}
\newcommand{\vnu}{\nu}

\newcommand{\mGamma}{\varGamma}

\newcommand{\mTheta}{\varTheta}

\newcommand{\mSigma}{\varSigma}

\renewoperator{\Re}{\mathrm{Re}}{\nolimits}
\renewoperator{\Im}{\mathrm{Im}}{\nolimits}
\makeatletter
\newcommand{\rd}{\@ifnextchar^{\DIfF}{\DIfF^{}}}
\def\DIfF^#1{%
   \mathop{\mathrm{\mathstrut d}}%
   \nolimits^{#1}\gobblespace}
\def\gobblespace{\futurelet\diffarg\opspace}
\def\opspace{%
   \let\DiffSpace\!%
   \ifx\diffarg(%
   \let\DiffSpace\relax
   \else
   \ifx\diffarg[%
   \let\DiffSpace\relax
   \else
   \ifx\diffarg\{%
   \let\DiffSpace\relax
   \fi\fi\fi\DiffSpace}

\newcommand{\E}{\operatorname{E}}

\newcommand{\Var}{\operatorname{Var}}

\newcommand{\Cov}{\operatorname{Cov}}

\newoperator{\ip}{\mathrm{int}}{\nolimits}

\newcommand{\plim}{\operatorname{plim}}

\newcommand{\tr}{\operatorname{tr}}
\renewcommand{\vec}{\operatorname{vec}}
\newcommand{\vech}{\operatorname{vech}}

\newcommand{\pto}{\stackrel{p}{\longrightarrow}}
\newcommand{\dto}{\stackrel{d}{\longrightarrow}}

\newcommand{\indic}[0]{\mathbbm{1}(b \geq 1)}
\newcommand{\indico}[0]{\mathbbm{1}(b_0 \geq 1)}

\newtheorem{theorem}{Theorem}[section]

\newtheorem{lemma}[theorem]{Lemma}
\newtheorem{corollary}[theorem]{Corollary}

\newcommand{\beq}{\begin{equation}}
\newcommand{\eeq}{\end{equation}}
\newcommand{\bal}{\begin{align*}}
\newcommand{\eal}{\end{align*}}

\newcommand{\bvec}{\begin{pmatrix}}
\newcommand{\evec}{\end{pmatrix}}
\newcommand{\bmat}{\begin{bmatrix}}
\newcommand{\emat}{\end{bmatrix}}

\newcommand{\bsmat}{\begin{smallmatrix}}
\newcommand{\esmat}{\end{smallmatrix}}

\title{Fractional trends in unobserved components models}
\author[1,2]{Tobias Hartl\footnote{Corresponding author. E-Mail: tobias1.hartl@ur.de\\
		The authors thank Uwe Hassler, Morten {\O}. Nielsen, Christoph Rust, the participants of the econometric seminar in Nuremberg, the department seminar at the Christian Albrechts University Kiel, the DAGStat conference 2019 in Munich, the workshop on high-dimensional time series in economics and finance 2019 in Vienna, the Annual Meeting of the German Statistical Society 2019 in Trier, the Annual Meeting of the German Economic Society 2019 in Leipzig, the Seminar on International Economic Policy at the University of Zurich, the International Conference on Computational and Financial Econometrics 2019 in London, the Symposium in Honor of Michael Hauser at WU Vienna, and the Standing Field Committee in Econometrics of the German Economic Society for many valuable comments. Support through the projects TS283/1-1 and WE4847/4-1 financed by the German Research Foundation (DFG) is gratefully acknowledged.}}
\author[1]{Rolf Tschernig}
\author[1,2]{Enzo Weber}

\affil[1]{University of Regensburg, 93053 Regensburg, Germany}
\affil[2]{Institute for Employment Research (IAB), 90478 Nuremberg, Germany}
\date{May 2020}

\begin{document}
\maketitle

\thispagestyle{empty}
\setcounter{page}{0}

\paragraph{\bf Abstract.}
We develop a generalization of unobserved components models that allows for a wide range of long-run dynamics by modelling the permanent component as a fractionally integrated process. 
The model allows for cointegration, does not require stationarity, and can be cast in state space form. 
We derive the Kalman filter estimator for the common fractionally integrated component and establish consistency and asymptotic (mixed) normality of the maximum likelihood estimator.  
We apply the model to extract a common long-run component of three US inflation measures, where we show that the $I(1)$ assumption is likely to be violated for the common trend. 


\paragraph{\bf JEL-Classification.}

C32, C51, E31

\newpage

\section{Introduction}
Unobserved components (UC) models are widely used to decompose time series into latent components of different persistence. Applications in economics include, among others, trend-cycle decompositions, the analysis of long-run equilibrium relations, testing for mean reversion e.g.\ in asset returns, and forecasting \citep[see][for an overview]{KimNel1999, KooShe2015}.

\noindent
Despite their wide spread, current UC models exhibit two major limitations. 
First, they require a priori assumptions about the integration order of a series and, therefore, an endogenous treatment of the long-run dynamic characteristics is infeasible. 
And second, they restrict the long-run component to be $I(0)$, $I(1)$, or $I(2)$. Statistical inference about the degree of persistence of a long-run component is then limited to prior unit root testing, ignoring the non-standard behavior of economic series that exhibit long memory and hindering the estimation of the integration order on a continuous support jointly with the other parameters of the model. Furthermore, model selection uncertainty from prior unit root testing is not taken into account. Finally, misspecification of the integration order may pollute the estimates of permanent and transitory components and bias the variance estimates for the permanent and transitory shocks. 

\noindent
While for the Beveridge-Nelson decomposition a generalization to ARFIMA processes was derived by \cite{AriMar2004} and \cite{Pro2016}, and low-frequency transformations that allow for fractional integration have been proposed by \cite{MueWat2018}, UC models lack a generalization to fractionally integrated processes. 
Deriving such a generalization is particularly challenging: 
It requires to study the convergence properties of the Kalman filter through which the unobserved components are estimated when fractional integration is allowed. 
In addition, to enable feasible estimation for time series of length $n$ with $n$ large,
 a modification of the Kalman filter is necessary, as the state vector of fractionally integrated processes is of dimension $n+1$, thus making the standard Kalman filter inapplicable from a computational perspective. 
Moreover, the asymptotic theory of the maximum likelihood estimator, that is utilized to estimate the model parameters, has to be derived.
 So far, asymptotic results are only available for the $I(1)$ case considered in \cite{ChaMilPa2009}, where in contrast to our model the integration order is assumed to be known.
Providing the theoretical analysis required for fractionally integrated UC models together with a computationally feasible estimator for the latent components is the core of this paper. 

\noindent
We contribute to the literature by deriving a fractionally integrated unobserved components model that allows for a flexible treatment of the long-run dynamic characteristics of  multivariate stochastic processes by letting the common integration order to take values on a set of positive real numbers including zero. Since we model a $p$-dimensional vector of observable random variables $\{\vy_t\}_{t=1}^n$ as a linear function of a scalar latent variable $x_t$ that is fractionally integrated of order $b$, our model exhibits $p-1$ fractional cointegration relations. Furthermore, our model can be used to decompose a set of variables into long- and short-run components, where the latter components are $I(0)$. 

\noindent
The model is cast in state space form and  allows for 
asymptotically stationary and nonstationary data. 
Although an exact state space representation of our model exists, estimating a latent fractionally integrated component via the Kalman filter is computationally infeasible for time series with sample size $n$ large. Therefore, we derive a modified version of the Kalman filter that is based on a truncated state space representation of our fractionally integrated unobserved components model while correcting the observable variables for the approximation error that results from the truncation. Our modified Kalman filter yields the same prediction error and likelihood function as the standard Kalman filter that is based on the full state space representation of a fractionally integrated process but greatly reduces the computing time by keeping the state dimension manageable. E.g.\ for our application in section \ref{Ch:4}, the modified Kalman filter is found to be about $150$ times faster than the standard Kalman filter.

\noindent
The second main technical contribution of our paper is to establish the asymptotic theory for the maximum likelihood estimator of our fractionally integrated unobserved components model. 
Since the asymptotic properties of the objective function depend on the fractional integration order $b_0$ of the data-generating process and differ for $b_0 < 1/2$ and $b_0 > 1/2$, we consider the asymptotically stationary case and the nonstationary case separately, where in each case the objective function of the maximum likelihood estimator uniformly converges. 
While a central limit theorem for martingale difference sequences holds for $b_0 < 1/2$ and yields  asymptotic normality of the maximum likelihood estimator, the nonstationary case is more involved.
Here, we first show that the prediction error variance of the Kalman filter converges. Next, we derive a  functional central limit theorem for the relevant partial sums of the score function that include fractional processes. From the functional central limit theorem the convergence rates of the estimates follow directly. Finally, we prove that the maximum likelihood estimator is asymptotically normally distributed, while a rotation of the parameter estimators that corresponds to the cointegrating matrix converges at rate $n^{b_0}$ to a mixed normal distribution, thus reflecting the behavior of cointegration models. From these results, it follows for the model parameters that standard inference results remain valid when a fractionally integrated component is introduced.

\noindent
As an empirical application, we consider the estimation of unobserved long-run inflation by extracting a common fractional component from a set of price measures for the US. For inflation, there exists substantial evidence suggesting that the series are fractionally integrated \citep[cf.\ eg.][]{HasWol1995, TscWebWe2013}.
We confirm such findings and estimate the integration order of unobserved long-run inflation to be $0.476$. We also show that misspecifying the integration order to be one yields estimated fundamental shocks that are antipersistent, which violates one important assumption of unobserved components models.

\noindent
The structure of the paper is as follows. Section \ref{Ch:2} details the fractionally integrated unobserved components model and discusses the estimation of the conditional expected value of the scalar latent variable that is allowed to be fractionally integrated. 
Section \ref{Ch:3} considers the maximum likelihood estimator for our model. By generalizing the proofs of \cite{ChaMilPa2009} for a common $I(1)$ component to the fractional case, we are able to show consistency, to derive the convergence rates for different parameters and to establish a central limit theorem for the maximum likelihood estimator. 
In section \ref{Ch:4} the model is applied to extract a common long-run component from different US inflation measures. Section \ref{Ch:5} concludes.
All proofs are collected in the appendix.

\section{A setup for common fractional components} \label{Ch:2}
In this section we first derive the fractionally integrated unobserved components model and state the necessary assumptions for identification. Next, we cast the model in state space form, from which we derive the Kalman filter estimator for the latent common long-run component, thereby generalizing the permanent-transitory decomposition of \citet{ChaMilPa2009}. Furthermore, since the Kalman filter estimator based on the exact state space representation is computationally infeasible for long time series, we propose a modified Kalman filter estimator that is based on a finite ARMA approximation of the fractionally integrated process but directly corrects for the resulting approximation error. In corollary \ref{th:3} we show that the modified estimator yields the same prediction error as the estimator that is based on the exact state space representation and, therefore, has the same likelihood but keeps the state dimension manageable.

\noindent
To begin with, consider the unobserved components model
\begin{align}
	\vy_t &= \vbeta x_t + \vu_t,  \label{dgp:1} \qquad
	\Delta^b_+ x_t = \eta_t, \qquad t=1,...,n,
\end{align}
where $\vy_t$ is a $p$-dimensional observable time series, $x_t$ is a scalar latent variable that is fractionally integrated of order $b$, $x_t \sim I(b)$, $b \in D$, $D = \{d \in \mathbb{R}\ | \ 0 \leq d < 3/2,\ d \neq 1/2\}$, $\vbeta$ is a $p\times1$ vector of factor loadings that are unobserved, $\vu_t \sim \mathrm{NID}(0, \Sigma)$ and $\eta_t\sim \mathrm{NID}(0, 1)$ are iid errors of dimension $p$ and $1$ that are independent and $\mSigma$ is diagonal and has full rank. We collect the parameters in $\vtheta = (\beta', (\vech \Sigma)', b)' \in \mTheta$.  
The model may be interpreted as a system where $p$ observable variables $\vy_t$ are driven by one common, fractionally integrated  stochastic trend $x_t$, such that the whole system is $I(b)$ and $p-1$ cointegration relations exist. The true parameters of the data-generating process are denoted as $\vbeta_0$, $\Sigma_0$, and $b_0$. They are collected in $\vtheta_0=(\vbeta_0', (\vech \mSigma_0)', b_0)' \in \mTheta$. We exclude the singular point $b_0=1/2$ since inference is different for $b_0 < 1/2$, where the maximum likelihood estimator is asymptotically Gaussian, and $b_0>1/2$, where a rotation of the parameter estimator for $\vbeta$ is asymptotically mixed normal, as will be shown in section \ref{Ch:3}. The same restriction applies to other cointegrated models \citep[cf.\ e.g.][]{JohNie2012}. Since we impose $\Var(\eta_t)=1$, $\mSigma$ diagonal and of full rank, the model is identified up to a sign for $\vbeta$. Therefore we restrict the first entry to be positive for unique identification.

\noindent
The fractional difference operator $\Delta^b$ is defined as 
\begin{align*}
	\Delta^{b} &= (1-L)^{b} = \sum_{j = 0}^{\infty}\pi_{j}(b)L^{j},  \qquad
	\pi_{j}(b) = 
	\begin{cases}
		\frac{j-b-1}{j}\pi_{j-1}(b) &  j = 1, 2, ..., \\ 
		1										&   j = 0,
	\end{cases} 
\end{align*}
and a $+\,$--subscript amounts to a truncation of an operator at $t \leq 0$, i.e.\ for an arbitrary process $z_t$, $\Delta^b_+ z_t = \sum_{j=0}^{t-1}\pi_j(b) L^j z_t$ \citep[see e.g.][]{Joh2008}. 
For $b \in \mathbb{N}_0$ the fractional long-run component nests the standard integer integrated specifications, whereas $b \in D$ adds flexibility to the weighting of past shocks. Throughout the paper, we adopt the type II definition of fractional integration \citep{MarRob1999} that assumes zero starting values for all fractional processes, and, as a consequence, allows for a smooth treatment of the asymptotically stationary ($b < 1/2$) and the nonstationary ($b > 1/2$) case. Due to the type II definition the inverse fractional difference  $\Delta_+^{-b}$ exists and is given by $\Delta^{-b}_+z_t = (1-L)^{-b}_+z_t = \sum_{j = 0}^{t-1}\varphi_{j}(b)z_{t-j}$, where $\varphi_j(b)=\pi_j(-b)$ for all $j$.
Finally, we make use of the fractional lag operator introduced in \cite{Joh2008} that is defined as $L_b = 1-\Delta_+^b$ and nests the standard lag operator $L_1 = L$ for $b = 1$. {Note that $L_bz_t$ preserves the integration order of a random variable $z_t$ since $b \in D$ is restricted to be non-negative.

\noindent
Let $\indic$ be an indicator function that becomes one if $b \geq 1$ and zero otherwise 
and let $d = b - \indic$ denote the mean-reverting fraction of a long memory process. Define $\Delta^{-d}_+ =\sum_{j=0}^{t-1} \varphi_j(d) L^j$ and $\Delta^{d}_+ = \sum_{j=0}^{t-1} \pi_j(d) L^j$ as a function of $d$, such that $\Delta^{-b}_+=(1-L)_+^{-\indic}\sum_{j=0}^{t-1} \varphi_j(d)L^j$ distinguishes between an integer integration order and the fractionally integrated polynomial with $d \in [0, 1)$. For notational convenience we omit $d$ in the binomial expansion of the fractional difference operators $\Delta_+^d$, $\Delta_+^{-d}$ and denote $\pi_j$, $\varphi_j$ as the $j$-th coefficient of $\Delta_+^d$, $\Delta_+^{-d}$ if not stated different explicitly. 
Then $x_t$ in \eqref{dgp:1} is represented as 
\begin{align} \label{dgp:1b}
x_t   = \indic {x}_{t-1} + \sum_{j=0}^{t-1}\varphi_j \eta_{t-j}.
\end{align}
\noindent 
Given the parameters $b$, $\vbeta$, and $\mSigma$, the exact state space representation of our model \eqref{dgp:1} is given by 
\begin{align*}
{\valpha}_{t+1} &= {\mT} {\valpha}_t +{\mR} \veta_{t+1}, &&
{\vy}_t = {\mZ} {\valpha}_t + \vu_t,
\end{align*}
where
\begin{align*}
{\mT}=\bmat \indic & 1 & 0 & \cdots & 0 \\ 
0 & 0 & 1 & \cdots & 0 \\
\vdots & \vdots & \vdots & \ddots & \vdots \\
0 & 0 & 0 & \cdots & 1 \\
0 & 0 & 0 & \cdots & 0
\emat, &&
{\mR} = \bvec 1 \\ \varphi_1 \\ \vdots \\ \varphi_{n-1} \\ \varphi_{n} \evec ,  &&&
{\valpha}_t = \bvec {x}_t \\ \varphi_1 \eta_t + \cdots + \varphi_n \eta_{t-n+1} \\
\vdots \\ \varphi_{n-1} \eta_t + \varphi_n \eta_{t-1} \\ \varphi_{n} \eta_t \evec,
\end{align*}
${\mZ}= \bmat \vbeta & 0 &\cdots & 0 \emat$ and where $\eta_t=0$ for all $t\leq 0$ due to \eqref{dgp:1b}.

\noindent
Let $\mathcal{F}_t$ be the $\sigma$-field generated by the observable variables $y_1$, ..., $y_t$. Furthermore, let $z_{t|s}=\mathrm{E}_\theta(z_t |\mathcal{F}_s)$ for $z=x, \alpha$, and ${\mP}_{t|s} = \mathrm{Var}_\theta({\valpha}_t|\mathcal{F}_s)$ with ${\omega}^{(i,j)}_{t}$ as its $(i, j)$-th entry for $s=t-1$. The $\vtheta$-subscript denotes that expectations are taken given a parameter vector $\vtheta$, and $\E_{\vtheta_0}(\vy_t | \mathcal{F}_{t-1}) = \E(\vy_t | \mathcal{F}_{t-1})$. 
Additionally, let ${\alpha}_{t|t-1}^{(j)}$ denote the $j$-th entry of ${\valpha}_{t|t-1}$.  
The prediction and updating steps of the Kalman filter for model \eqref{dgp:1} given the observable data and the parameter vector $\vtheta$ are 
\begin{align}
{\vv}_t(\theta) &= \vy_t - \mathrm{E}_\theta({\vy}_t | \mathcal{F}_{t-1}) = \vy_t - \vbeta \mathrm{E}_\theta({x}_t | \mathcal{F}_{t-1}) = \vy_t - \vbeta {x}_{t|t-1}, \label{tilde_v_t} \\
{\mF}_t &= \mathrm{Var}_\theta ({\vv}_t(\vtheta ) | \mathcal{F}_{t-1}) = \vbeta \mathrm{Var}_\theta(x_t|\mathcal{F}_{t-1})\vbeta' + \mSigma = \vbeta w_t^{(1,1)}  \vbeta' + \mSigma \label{eq:F_t}, \\
{\valpha}_{t+1|t}&={\mT} {\valpha}_{t|t-1} + {\mT} {\mP}_{t|t-1} {\mZ}'\mF_t^{-1} {\vv}_t(\theta), \label{ap:up:x}\\
{\mP}_{t+1|t} &= {\mT} {\mP}_{t|t-1}{\mT}' - {\mT} {\mP}_{t|t-1}{\mZ}' {\mF}_t^{-1}{\mZ} {\mP}_{t|t-1} {\mT}' + {R}{R}'. \label{tilde_P_t+1_t}
\end{align}
The following theorem states the conditional expectation of the latent variable $x_t$ given $\mathcal{F}_{t-1}$ 
and generalizes the results of \cite{ChaMilPa2009} for $I(1)$ stochastic trends to the fractional domain.
\begin{theorem}\label{th:1}
	For the exact state space representation of the unobserved components model \eqref{dgp:1} the conditional expectation of the latent variable $x_{t+1}$ is 
	 given by
	\begin{align*}
		{x}_{t+1|t} {=} \frac{\vbeta'\mSigma^{-1}}{\vbeta' \mSigma^{-1}\vbeta}y_{t+1} - z_{t+1}(\vtheta),
	\end{align*}
	where
		\begin{align}
			z_{t+1}(\vtheta) &=  \frac{\vbeta' \mSigma^{-1}  }{\vbeta' \mSigma^{-1}\vbeta }\left( \Delta_+^b \vy_{t+1}  -  \E_\vtheta(\Delta_+^b \vy_{t+1} | \mathcal{F}_t)  \right), \nonumber \\
			\vv_{t+1}(\vtheta) &=\left( I -  \frac{\vbeta \vbeta' \mSigma^{-1}  }{\vbeta' \mSigma^{-1}\vbeta } \right)\vy_{t+1} + \vbeta z_{t+1}(\vtheta). \label{kf:can2}
	\end{align}
\end{theorem}
\noindent
The proof of theorem \ref{th:1} is contained in appendix \ref{AS2}. There, and in the proofs that follow, we denote $w_t$ as any $I(0)$ process that is a function of the underlying NID distributed shocks $u_{1},...,u_t$, and $\eta_1,...,\eta_t$. Since $\E_{\vtheta_0}(\vy_t |\mathcal{F}_{t-1}) = \E(\vy_t |\mathcal{F}_{t-1})$, and thus $\vv_t(\vtheta_0) = \vy_t - \E(\vy_t | \mathcal{F}_{t-1})$, it follows that $(\vv_t(\vtheta_0), \mathcal{F}_t)$ is a martingale difference sequence (MDS).

\noindent
Theorem \ref{th:1} illustrates that the Kalman filter estimator $x_{t+1|t}$ can be decomposed into a linear combination of $\vy_{t+1}$ that is $I(b_0)$ and an additive component $z_{t+1}(\vtheta)$ where the latter is the prediction error for the fractionally differenced univariate  process $\Delta_+^b \frac{\vbeta' \mSigma^{-1}  }{\vbeta' \mSigma^{-1}\vbeta } y_{t+1} $ given the filtration $\mathcal{F}_t$. The  integration order of this prediction error is given by the following lemma.
\begin{lemma} \label{lemma:z_t_order}
The univariate prediction error $z_{t}(\vtheta)$ is $I(b_0 - b)$ for all $t=1,...,n$.
\end{lemma}
\noindent The proof is included in appendix \ref{AS2}.
Thus, the Kalman filter estimator $x_{t+1|t}$ is always $I(b_0)$. 
The prediction error $\vv_{t+1}(\vtheta)$ combines  errors from $\vbeta \neq \vbeta_0$ and errors from $b \neq b_0$. It is
$I(b_0)$ for $\beta \neq \beta_0$, since $\left(\mI - \frac{\beta \beta' \mSigma^{-1}}{\vbeta' \mSigma^{-1}\vbeta} \right)\vbeta_0 x_{t+1} \neq 0$, whereas $\vbeta = \vbeta_0$ yields $\vv_{t+1}(\vtheta)=\left( \mI - \frac{\vbeta_0 \vbeta_0' \mSigma^{-1}}{\vbeta_0'\mSigma^{-1}\vbeta_0} \right)\vu_{t+1} + 
\vbeta_0 z_{t+1}(\vtheta)\sim I(b_0 - b)$ by lemma \ref{lemma:z_t_order}. Finally, $\vv_{t+1}(\vtheta_0) \sim I(0)$.


\noindent
Although a finite-order state space representation of the system in \eqref{dgp:1} exists since a fractionally integrated process of type II exhibits a finite-order autoregressive representation of length $n-1$, estimating such a system is only computationally feasible when $n$ is small. To estimate $\valpha_t$ the Kalman filter computes the inverse of the $(n+1)\times(n+1)$ covariance matrix $\mP_{t|t-1}$ for $t=1,...,n$ sequentially, which makes the filter inapplicable for large $n$. As a solution, \cite{ChaPal1998} suggest to truncate the Wold representation of a fractionally integrated process after $m$ lags before the model is cast in state space form, and provide consistency results for $b_0 < 1/2$. \cite{HarWei2018} find that a purely fractionally integrated trend is well approximated by finite ARMA processes in several simulation studies. For optimization purposes their approach is particularly convenient since it maps from the fractional integration order $b$ to its related ARMA coefficients and, therefore, optimization is conducted over $b$. 

\noindent
Nonetheless, the literature lacks consistency results for finite approximations of fractionally integrated processes in state space form when $b_0 > 1/2$, and we expect any estimator that truncates the fractionally integrated process at lag $m$, $m<n$, to become inconsistent as soon as $b_0 > 1/2$, $b_0 \neq 1$, since the variance of the truncated sum $(1-L)^{-\indico} \sum_{j=m+1}^{n-1}\varphi_j(d_0) \eta_{n-j}$ diverges as $n \to \infty$.

\noindent
As a solution, we include a correction for the resulting approximation error that allows us to contribute to the literature on fractionally integrated processes in state space form by deriving consistency results for the maximum likelihood estimator when $b_0 \in D$. To obtain a computationally feasible representation, we approximate the fractionally integrated process by a finite-order ARMA process, but directly correct for the resulting approximation error. We base our theoretical analysis on ARMA($1, m$) approximations of $x_t$, where the moving average polynomial truncates the stable part of the Wold representation of a fractionally integrated process, whereas the AR polynomial controls for integration orders greater or equal to one.
As will be shown in this section, the modified Kalman filter yields the same likelihood function as the one that is based on the exact state space representation of a fractionally integrated process.

\noindent
Let $\tilde{\vy}_t$ denote an approximate version of \eqref{dgp:1} and \eqref{dgp:1b} that is obtained by truncating the fractional polynomial $ \sum_{i=0}^{t-1}\varphi_i \eta_{t-i}$ after lag $m$,
\begin{align}\label{tru:1}
\tilde{\vy}_t = \vbeta \tilde{x}_t + \vu_t, && \tilde{x}_t = \indic \tilde{x}_{t-1} + \sum_{i=0}^{m}\varphi_i \eta_{t-i},
\end{align}
such that $(1-L)^{\indic}(\tilde{x}_t - x_t) = -\sum_{i=m+1}^{t-1}\varphi_i \eta_{t-i}$. \\
The system matrices and variables of the approximate state space representation are denoted with tilde, i.e.\ $\tilde{\mT}$, $\tilde{\mZ}$, $\tilde{\mR}$, $\tilde{\valpha}_{t}$,  $\tilde{v}_t(\vtheta)$, $\tilde{\mP}_{t|s}$, and $\tilde{\omega}^{(i,j)}_{t}$. Hence, $\tilde{\mT} = \mT^{({1:(m+1), 1:(m+1)})}$ consists of the upper $m+1$ columns and rows of $\mT$, $\tilde{\mZ}= \mZ^{(\cdot, 1:(m+1))}$ holds the first $m+1$ columns of $\mZ$, $\tilde{R} = \mR^{(1:(m+1), \cdot)}$ consists of the first $m+1$ rows of $\mR$ and the $(m+1)$ vector $\tilde{\alpha}_t$ is given by
$\tilde{\valpha}_t = \bvec \tilde{x}_t & \varphi_1 \eta_t + ... + \varphi_m \eta_{t+1-m} &
\cdots & 
\varphi_{m} \eta_t \evec'$.
$\tilde{\mP}_{t|s}$, $\tilde{\vv}_t(\vtheta)$ are defined accordingly. The Kalman filter equations \eqref{tilde_v_t} to \eqref{tilde_P_t+1_t} hold equivalently if denoted with tilde. 

\noindent
In the following theorem we state the conditional expectation $\tilde{x}_{t+1|t}$ of the truncated model as a function of $x_{t+1|t}$ and an approximation error.
\begin{theorem}\label{th:2}
	Let $\ve_i$ be a $(1 \times t)$ unit vector with a one at column $i$ and zeros elsewhere. Define $\mY_t = (y_1', ..., y_t')'$ and $\veta_{1:t} = (\eta_1, ..., \eta_t)'$. For the truncated model \eqref{tru:1} the conditional expectation of the latent variable can be written as
	\begin{align*}
		\tilde{x}_{t+1|t} &= x_{t+1|t} - \epsilon_{t+1}(\vtheta), \\
		\epsilon_{t+1}(\vtheta) &=
		\begin{cases}
			\sum_{i=m+1}^{t} \varphi_i \ve_{t+1-i} \mSigma_{\eta_{1:t}\mY_t} \mSigma_{\mY_t}^{-1} \mY_t & \text{if } b < 1,\\
		\sum_{s=m+1}^{t}\sum_{i=m+1}^{s} \varphi_i \ve_{s+1-i} \mSigma_{\eta_{1:t}\mY_t} \mSigma_{\mY_t}^{-1} \mY_t & \text{if } b \geq 1,
		\end{cases}
	\end{align*}
where $\epsilon_{t+1}(\vtheta)$ denotes the approximation error, and $\mSigma_{\veta_{1:t}\mY_t} = \Cov_\vtheta ( \veta_{1:t}, \mY_t)$, $\mSigma_{\mY_t}=\Var_\vtheta(\mY_t)$. Furthermore $\mE_\vtheta(\epsilon_{t+1}(\vtheta)) = 0$. Details on these matrices are presented in the proof, which is contained in appendix \ref{AS2}.
\end{theorem}

\noindent
The prediction error $\vv_{t+1}(\vtheta)$ can be decomposed into the prediction error of the truncated model plus the approximation error
\begin{align*}
	\vv_{t+1}(\vtheta) = \vy_{t+1} - \mE_\vtheta(y_{t+1} | \mathcal{F}_t) &= \vy_{t+1} -  \E_\vtheta(\tilde{y}_{t+1} | \mathcal{F}_t) - \E_\vtheta(\vy_{t+1} - \tilde{y}_{t+1} | \mathcal{F}_{t}) = \\
	&= \vy_{t+1} - \vbeta \tilde{x}_{t+1|t} - \vbeta \epsilon_{t+1}(\vtheta) = \tilde\vv_{t+1}(\vtheta) - \vbeta \epsilon_{t+1}(\vtheta).
\end{align*}
Note that the approximation error $\epsilon_t(\vtheta)$ is the Kalman filter estimate for $x_{t} - \tilde{x}_{t} = (1-L)^{-\indic} \sum_{i=m+1}^{t-1}\varphi_i \eta_{t-i}$ given $\mathcal{F}_{t-1}$ and, therefore, it is $\mathcal{F}_{t-1}$-measurable and can be calculated given the formula in theorem \ref{th:2}. Consequently, the results from theorem \ref{th:1} for the exact representation carry over if $\vy_{t}$ is corrected for the approximation error, as the following corollary states.
\begin{corollary}\label{th:3}
	Define  $\ddot{\vy}_t = \vy_t - \vbeta \epsilon_t(\vtheta)$.
	Using the results in theorem \ref{th:1} and \ref{th:2} yields 
	\begin{align*}
		\E_\vtheta(\ddot{\vy}_{t+1} | \mathcal{F}_t)&=\E_\vtheta(\tilde{\vy}_{t+1} | \mathcal{F}_t) = \vbeta \tilde{x}_{t+1|t}, \\
		\tilde{x}_{t+1|t} &=  \frac{\vbeta' \mSigma^{-1}}{\vbeta' \mSigma^{-1} \vbeta} \vy_{t+1} - \epsilon_{t+1}(\vtheta) - z_{t+1}(\vtheta),
	\end{align*}
and $\vv_{t+1}(\vtheta)=\ddot{\vy}_{t+1}  - \E_\vtheta(\tilde{\vy}_{t+1} | \mathcal{F}_t) = \ddot{\vy}_{t+1}  - \E_\vtheta(\ddot{\vy}_{t+1} | \mathcal{F}_t)$. 
\end{corollary}

\noindent
From corollary \ref{th:3} it follows that the prediction errors of the exact representation \eqref{dgp:1} using $\{\vy_t\}_{t=1}^n$ and the truncated model \eqref{tru:1} together with the 
approximation-corrected $\{\ddot{\vy}_t\}_{t=1}^n$ are identical and have the same conditional likelihood given $\vtheta$. Hence, maximizing the likelihood of the approximation-corrected truncated model solves the same optimization problem as for the exact state space representation but requires a smaller number of state estimates from the Kalman filter if $m < n$. The modified Kalman filter outperforms the standard Kalman filter from a computational perspective whenever $p \ll n$, as it requires to invert the $np\times np$ matrix $\mSigma_{Y_n}$ once, whereas the Kalman filter based on the full representation of \eqref{dgp:1} sequentially inverts the $(n+1) \times (n+1)$ matrix $\mP_{t|t-1}$ for each $t=1,...,n$. 
E.g.\ for our application in section \ref{Ch:4}, the modified Kalman filter is about $150$ times faster than the standard Kalman filter.

\noindent
Although we base our theoretical analysis on ARMA($1, m$) approximations of $x_t$, including further lags of the autoregressive polynomial may improve the approximation quality in finite samples, as \cite{HarWei2018} show, and, therefore, speed up the parameter optimization. Nonetheless, the asymptotic results remain unaffected by an extended AR polynomial since correcting for the approximation error yields an exact representation of a fractionally integrated process anyway. For notational convenience we therefore stick to the simplest ARMA({$1, m$}) approximation in section \ref{Ch:2}, 
whereas in our empirical application in section \ref{Ch:4} we use ARMA($4$, $4$) approximations for a faster convergence of the estimator. 

\noindent
Having shown that an exact representation \eqref{dgp:1} together with $\{\vy_t\}_{t=1}^n$ yields the same conditional likelihood of the prediction error as a truncated, approximation-corrected model $\eqref{tru:1}$ together with $\{\ddot{\vy}_t\}_{t=1}^n$ for a given $\vtheta$, we turn to the estimation of the unknown parameters $\vtheta$ in the subsequent section, where we focus on the exact state space representation of \eqref{dgp:1}. For the asymptotic results to carry over to the truncated, approximation-corrected model it is required that $\epsilon_t(\vtheta) < \infty$, and therefore the truncation parameter is required to depend on the sample size $n$, $m = m(n)$, whenever $b > 1/2$.  

\section{Maximum likelihood estimation}\label{Ch:3}
In this section we derive the maximum likelihood (ML) estimator for the unknown parameters $\vtheta$ in the unobserved components model \eqref{dgp:1} with a common fractional trend and determine the asymptotic properties of the ML estimator. With respect to the latter, two major difficulties have to be tackled. 
First, as it already becomes clear from theorem \ref{th:1} and lemma \ref{lemma:z_t_order}, $z_t(\vtheta)$ depends on $b_0 - b$ and is nonstationary for $b_0 - b \geq 1/2$. We tackle this issue by first establishing consistency of the ML estimator for $b$, where we show that the estimator is nested in the ARFIMA optimization problem considered in \cite{Nie2015}. There, consistency of the estimator for $b$ is shown by splitting $D$ into different intervals and showing that the relevant parameter space reduces to $D_3(\kappa_3)=D \cap \{b: b-b_0 \geq -1/2 + \kappa_3\}$, $0<\kappa_3<1/2$, where the objective function of the estimator converges uniformly. Consequently, $z_t(\vtheta)$ and the partial derivative of $\vv_t(\vtheta)$ w.r.t.\ $b$ converge to stationary processes. 

\noindent
The second difficulty arises from the partial derivative of $\vv_t(\vtheta)$ w.r.t.\ $\vbeta$ that is $I(b_0)$, which implies that the convergence rate of the ML estimator for $\vbeta$ depends on $b_0$ for $b_0 \in (1/2, 3/2)$. Consequently, we consider the asymptotically stationary case $b_0 \in [0, 1/2)$ and the nonstationary case $b_0 \in (1/2, 3/2)$ separately. For both cases we show that the ML estimator of $\vtheta$ converges to a normal distribution, whereas in the latter case a certain rotation of the parameters is asymptotically mixed normally distributed.

\noindent
The section is organized as follows. We first state the log likelihood of the state space model \eqref{dgp:1} together with its first and second derivative and comment on the convergence of the prediction error variance $F_t$ in \eqref{eq:F_t}. Next, we show consistency of the ML estimator for $b$. Finally, we derive the asymptotic distribution for the ML estimator of $\vtheta$ for the asymptotically stationary case $b_0 \in [0, 1/2)$ and the nonstationary case $b_0 \in (1/2, 3/2)$ separately,
including a discussion  on the cointegration properties implied by the model.

\noindent
The log likelihood of our state space system is given by
\begin{align}\label{ll}
	l_n(\vtheta) = -\frac{n}{2}\log \det \mF^{[n]} - \frac{1}{2} \tr \mF^{[n]^{-1}} \sum_{t=1}^{n} \vv_t(\vtheta) \vv_t(\vtheta)',
\end{align}
where $\mF^{[n]}=\lim_{t \to \infty} \Var_\vtheta(v_t(\vtheta) | \mathcal{F}_{t-1})$ is the steady state variance of the prediction error that depends on the fixed system dimension $n$ due to the type II definition of long memory. 
The existence of a steady state $\mF^{[n]}$ is shown in lemma \ref{L:1b} in appendix \ref{AS3}. 
The derivation of the asymptotic properties of the ML estimator requires convergence of the steady state variance $\mF^{[n]}$ as $n \to \infty$. This is shown in lemma \ref{L:1c} in appendix \ref{AS3}, where special care is taken w.r.t.\ the state dimension increasing with $n$. 

\noindent
An analytical solution for the score and Hessian matrix was derived in \cite{ChaMilPa2009} and is given by
\begin{align}
	\vs_n(\vtheta)  &= -\frac{n	}{2}\frac{\partial (\vec \mF^{[n]})'}{\partial \theta} \vec \mF^{[n]^{-1}} + \frac{1}{2} \frac{\partial (\vec \mF^{[n]})'}{\partial \theta}\vec \left(\mF^{[n]^{-1}}\sum_{t=1}^{n}\vv_t(\vtheta)\vv_t(\vtheta)' \mF^{[n]^{-1}}\right) \nonumber \\
	&- \sum_{t=1}^{n}\frac{\partial \vv_t(\vtheta)'}{\partial \theta} \mF^{[n]^{-1}}\vv_t(\vtheta), \label{eq:Gradient}
\end{align}
\begingroup
\allowdisplaybreaks
and
\begin{align} \label{eq:Hessian}
	&\mH_n(\vtheta) = \sum_{h=1}^8 \mH_{n,h}(\vtheta),  \\
	&\mH_{n,1}(\vtheta) = -\frac{n}{2} \left[\mI \otimes (\vec  \mF^{[n]^{-1}})'\right] \left(\frac{\partial^2 }{\partial \theta\partial \theta'} \otimes \vec \mF^{[n]}\right),  \nonumber \\
	&\mH_{n,2}(\vtheta)  = \frac{1}{2}\left\{\mI \otimes \left\{\vec { \left[\mF^{[n]^{-1}} \left(\sum_{t=1}^{n}\vv_t(\vtheta) \vv_t(\vtheta)'\right)\mF^{[n]^{-1}} \right] } \right\}' \right\}\left(\frac{\partial^2}{\partial \theta\partial \theta'} \otimes \vec \mF^{[n]}\right), \nonumber \\
	&\mH_{n,3}(\vtheta) = \frac{n}{2}\frac{\partial (\vec  \mF^{[n]})'}{\partial \theta} \left(\mF^{[n]^{-1}} \otimes \mF^{[n]^{-1}}\right)\frac{\partial (\vec \mF^{[n]})}{\partial \theta'}, \nonumber   \\
	&\mH_{n,4}(\vtheta)  = -\frac{1}{2}\frac{\partial (\vec  \mF^{[n]})'}{\partial \theta} \Big[\mF^{[n]^{-1}} \otimes \mF^{[n]^{-1}}\left(\sum_{t=1}^{n}\vv_t(\vtheta)\vv_t(\vtheta)'\right)\mF^{[n]^{-1}} \nonumber \\
	&\qquad\qquad + \mF^{[n]^{-1}}\left(\sum_{t=1}^{n}\vv_t(\vtheta)\vv_t(\vtheta)'\right)\mF^{[n]^{-1}} \otimes \mF^{[n]^{-1}}\Big] \frac{\partial \vec \mF^{[n]}}{\partial \vtheta'}, \nonumber  \\
	&\mH_{n,5}(\vtheta)  = - \sum_{t=1}^{n}\frac{\partial \vv_t(\vtheta)'}{\partial \vtheta} \mF^{[n]^{-1}} \frac{\partial \vv_t(\vtheta)}{\partial \vtheta'},\quad \nonumber \\
	&\mH_{n,6}(\vtheta) = -\sum_{t=1}^{n} \left( \mI \otimes \vv_t(\vtheta)' \mF^{[n]^{-1}} \right)\left(\frac{\partial^2}{\partial \vtheta \partial \vtheta'} \otimes \vv_t(\vtheta) \right), \nonumber  \\
	&\mH_{n,7}(\vtheta) = \frac{\partial (\vec \mF^{[n]})'}{\partial \vtheta} \left(\mF^{[n]^{-1}} \otimes \mF^{[n]^{-1}}\right)\sum_{t=1}^{n}\left(\frac{\partial\vv_t(\vtheta)}{\partial \vtheta'}\otimes \vv_t(\vtheta)\right), \quad \mH_{n,8}(\vtheta) = \mH_{n,7}(\vtheta)'.  \nonumber 
\end{align}

\subsection{Consistency of the ML estimator for $b$}
Having stated the log likelihood together with its derivatives, we turn to the estimation of $b$. 
By theorem \ref{th:1} the prediction error has the decomposition  $\vv_{t+1}(\vtheta)= (I - \frac{\vbeta \vbeta' \mSigma^{-1}}{\vbeta'\mSigma^{-1}\vbeta}) \vy_{t+1} + \vbeta z_{t+1}(\vtheta)$. Since  the second term of $z_{t+1}(\vtheta)$  is $I(b_0-b)$ by  lemma \ref{lemma:z_t_order}, the prediction error  is $I(b_0)$ whenever $\vbeta\neq\vbeta_0$ and $I(b_0-b)$ in case of $\vbeta=\vbeta_0$. However, since the first term in $\vv_{t+1}(\vtheta)$ is invariant with respect to $b$, only the second term $\vbeta z_{t+1}(\vtheta)$  matters w.r.t.\ estimating $b$. The latter term
 is asymptotically stationary if $b_0 - b < 1/2$, such that a law of large numbers can be applied to obtain uniform convergence of the objective function for $b$. For $b_0 - b \geq 1/2$, $z_{t+1}(\vtheta)$ is nonstationary, and the rate of convergence of the objective function \eqref{ll} depends on $b_0-b$. Thus, the objective function of the ML estimator for $b$ does not converge uniformly on $D$. 
For ARFIMA models \cite{Nie2015} shows consistency of the conditional sum-of-squares (CSS) estimator for $b$, and the CSS estimator has the same limit distribution as the maximum likelihood estimator under Gaussianity \citep{HuaRob2011}. Thus, by showing that our objective function of the ML estimator for $b$ is asymptotically nested in the ARFIMA objective function considered in \cite{Nie2015} and that our setup satisfies assumptions A to D in \cite{Nie2015}, we prove that consistency for the ML estimator of $b$ carries over from the CSS estimator. The following theorem summarizes the results. 
	\begin{theorem}\label{th:consistency}
		The ML estimator for $b$ in model \eqref{dgp:1} is consistent, i.e.\ $\hat{b} \pto b_0$
		as $n \to \infty$. 
	\end{theorem}
\noindent
The proof is contained in appendix \ref{AS3}. 

\noindent
Theorem \ref{th:consistency} implies that the relevant parameter space for $b$ asymptotically reduces to the neighborhood of $b_0$, implying that $z_{t+1}(\hat\vtheta)$ is asymptotically stationary and the objective function for the ML estimator of $b$ converges uniformly.

\subsection{Asymptotic distribution of the maximum likelihood estimator}
\noindent
Next we turn to the asymptotic analysis of the maximum likelihood estimator for $\vtheta$.
To derive its asymptotic properties, we follow the well-established approach used for stationary models and apply a first order Taylor expansion to the score vector, which yields
\begin{align}
	{\vs}_n(\hat{\theta}_n) = \vs_n(\theta_0) + \mH_n(\theta_n)(\hat{\theta}_n - {\theta}_0),
\end{align}
where $\hat{\vtheta}_n$ is the maximum likelihood estimator for $\vtheta_0$, and $\mH_n(\vtheta_n)$ denotes the Hessian with rows evaluated at mean values between 
$\hat{\vtheta}_n$ and $\vtheta_0$. Given that $\vs_n(\hat{\vtheta}_n)=0$ if $\hat{\vtheta}_n$ is an interior solution, we write
\begin{align}\label{eq3:1}
		\vnu_n' \mA^{-1} (\hat{\vtheta}_n - \vtheta_0) = -\left[\vnu_n^{-1}\mA' \mH_n(\vtheta_n) \mA \vnu_n^{-1'}\right]^{-1}\left[\vnu_n^{-1} \mA' \vs_n(\vtheta_0)\right],
\end{align}
where $\vnu_n$ is a scaling matrix and $\mA$ is a rotation matrix that will be defined in \eqref{eq:A_N_A_S_A_D_nu_n} below.
Again following \cite{ChaMilPa2009}, the score vector \eqref{eq:Gradient} evaluated at the true parameter value $\vtheta_0$ is given by
\begin{align} \label{eq:sn0}
\vs_n(\vtheta_0) &= 
\frac{1}{2}\frac{\partial( \vec \mF^{[n]})'}{\partial \vtheta}\Bigg \rvert_{\vtheta=\vtheta_0} \left( \mF_0^{{[n]}^{-1}} \otimes \mF_0^{{[n]}^{-1}}\right) \vec  \sum_{t=1}^{n} \left(\vv_t(\vtheta_0) \vv_t(\vtheta_0)^{'}-\mF_0^{[n]}\right)\nonumber \\&- \sum_{t=1}^{n} \frac{\partial \vv_t(\vtheta)'}{\partial \vtheta}\Bigg\rvert_{\vtheta = \vtheta_0}\mF_0^{{[n]}^{-1}}\vv_t(\vtheta_0),
\end{align}
where $\mF^{{[n]}}_0$ is $\mF^{{[n]}}$ evaluated at $\vtheta = \vtheta_0$. 

\noindent
It is easy to see that the only stochastic component in $s_n(\vtheta_0)$ is $\vv_t(\vtheta_0)$ and its derivative evaluated at $\vtheta_0$. From the decomposition of $\vv_t(\vtheta_0)$ derived in theorem \ref{th:1}, one can obtain its derivatives stated in the following lemma.
\begin{lemma}\label{L3a}
	The first partial derivatives of $\vv_t(\vtheta)$, evaluated at $\vtheta_0$, are given by
	\begin{align*}
	 \frac{\partial \vv_t(\vtheta)'}{\partial \vbeta}\Bigg\rvert_{\vtheta = \vtheta_0} &=- \left( \mI - \frac{\mSigma_0^{-1} \vbeta_0 \vbeta_0'}{\vbeta_0' \mSigma_0^{-1}\vbeta_0} \right) x_t + a_\vbeta^0(\vu_t, \eta_t), \\
	 \frac{\partial \vv_t(\vtheta)'}{\partial \vec \mSigma} \Bigg\rvert_{\vtheta = \vtheta_0} &=a^0_\Sigma(\vu_t, \eta_t) , \qquad
	\frac{\partial \vv_t(\vtheta)'}{\partial b}\Bigg \rvert_{\vtheta = \vtheta_0} = a^0_b(\vu_t, \eta_t),
	\end{align*} 
where $a^0_i(\vu_t, \eta_t)$ are $I(0)$ processes that depend on $\eta_1,...\eta_t$, $\vu_1,..., \vu_t$, $i=\vbeta, \mSigma, b$.
\end{lemma} 
\noindent
The proof of lemma \ref{L3a} is contained in appendix \ref{AS3}. As the lemma shows, $\partial \vv_t(\vtheta)' / \partial \vbeta$ at $\vtheta_0$ is the only source of fractional integration in the gradient $s_n(\vtheta_0)$, whereas $\partial \vv_t(\vtheta)' / \partial \vec \mSigma$, $\partial \vv_t(\vtheta)' / \partial b$ at $\vtheta_0$ are $I(0)$. 
Similar to the $I(1)$ case studied in \cite{ChaMilPa2009} the partial derivative $\partial \vv_t(\vtheta)' / \partial \vbeta$ at $\vtheta_0$ is a process of dimension $(p\times p)$ that is driven by one common fractionally integrated trend $x_t$, such that $\partial \vv_t(\vtheta)' / \partial \vbeta$ at $\vtheta_0$ is cointegrated.   Defining the $(p \times p)$-dimensional projection matrix 
\begin{align} \label{eq:Px}
\mP_x =  \mI - \frac{\mSigma_0^{-1}\vbeta_0\vbeta_0'}{\vbeta_0' \mSigma_0^{-1} \vbeta_0},
\end{align}
as  \cite{ChaPal1998} do for the $I(1)$-case, allows to write  $\partial \vv_t(\vtheta)' / \partial \vbeta\rvert_{\vtheta=\vtheta_0} = - \mP_x x_t + a_\vbeta^0(u_t,\eta_t)$. While for each column in $\partial \vv_t(\vtheta)' / \partial \vbeta\big\rvert_{\vtheta=\vtheta_0}$ the dimension of the cointegration space is $p-1$,  $c\vbeta_0$ is the only common cointegrating vector for all $p$ columns, where $c$ is any nonzero constant, eliminating the single common trend from all $p^2$ derivatives,  $\vbeta_0'\partial \vv_t(\vtheta)' / \partial \vbeta\big\rvert_{\vtheta=\vtheta_0} \sim I(0)$. 
Thus,  the projection matrix satisfies $\vbeta_0' \mP_x = 0$. Furthermore, $ \mP_x \mSigma^{-1}_0 \vbeta_0 = 0$ holds. 
From the latter equation it follows that  $\mP_x \mSigma^{-1}_0$ is relevant for determining the cointegration space for $\vy_t=\vbeta_0x_t+\vu_t$. To deal with the singularity in $\mP_x$, we follow the approach of \cite{ChaMilPa2009} and define $\mGamma_0$ as a $p \times (p-1)$ matrix for which
\begin{align}\label{gamma:prop}
\mGamma_0' \mSigma_0^{-1} \vbeta_0 = 0, && \mGamma_0'\mSigma_0^{-1}\mGamma_0 = \mI.
\end{align}
Note that $\mP_x = \mSigma_0^{-1} \mGamma_0 \mGamma_0'$. Thus, $\mGamma_0' \mSigma_0^{-1}$ determines the $p-1$-dimensional cointegration space for $\vy_t$. From the left equation in  \eqref{gamma:prop} it follows that the cointegration vectors for $\vy_t$ and for the partial derivatives $\partial \vv_t(\vtheta)' / \partial \vbeta\rvert_{\vtheta=\vtheta_0}$ are orthogonal.  For a broad discussion of the cointegrating properties we refer to \citet[ch.\ 4]{ChaMilPa2009}.
In addition, note that the derivatives $\frac{\partial \vv_t(\vtheta)'}{\partial \vtheta} = -\frac{\partial x_{t|t-1}\vbeta'}{\partial \vtheta}$ are $\mathcal{F}_{t-1}$-measurable since $x_{t|t-1}$ is $\mathcal{F}_{t-1}$-measurable.


\noindent
Next, we study the asymptotic properties of $\vv_t(\vtheta)$ and $\frac{\partial \vv_t(\vtheta)'}{\partial \vtheta}$ at $\vtheta = \vtheta_0$.
From the Kalman recursions, in particular \eqref{tilde_v_t} and \eqref{ap:up:x} which contain random components, it follows that $\vv_t(\vtheta)$ is normally distributed since the recursions are linear and the errors $\eta_t$ and $\vu_t$ are assumed to be NID. 
Furthermore, $(\vv_t(\vtheta_0), \mathcal{F}_t)$ is a martingale difference sequence (MDS) by construction.  Moreover, the MDS is asymptotically stationary since its conditional variance $\Var( \vv_t(\vtheta_0) | \mathcal{F}_{t-1})$ converges asymptotically, $\lim\limits_{n \to \infty}\lim\limits_{t \to \infty} \Var( \vv_t(\vtheta_0) | \mathcal{F}_{t-1})= \mF_0$, as shown in lemma \ref{L:1c} in the appendix, so that $\mF_0$ is the asymptotic variance for the MDS $\vv_t(\vtheta_0)$. Since $\vv_t(\vtheta_0)$ adapted to $\mathcal{F}_t$ is uncorrelated,  normally distributed due to the NID errors as argued above, and has a finite asymptotic variance,  we have  $\vv_t(\vtheta_0) \dto \mathrm{NID}(0, \mF_0^{[n]})$ 
 as $t \to \infty$ for given $n$ and given  the adaption to $\mathcal{F}_t$.
It follows  from the results  of \citet[pp. 85--91]{Mui1982}
on the asymptotic properties of the Wishart distribution that
\begin{align}\label{eq:wis}
\frac{1}{\sqrt{n}} \sum_{t=1}^{n} \vec \left( \vv_t(\vtheta_0) \vv_t(\vtheta_0)' - \mF_0^{[n]} \right) \dto N\left(0, (\mI + \mK) (\mF_0 \otimes \mF_0)\right),
\end{align} 
as $n \to \infty$ where $K$ is the commutation matrix.

\noindent
As shown in lemma \ref{L3:b} in appendix \ref{AS3}, $\left(\frac{\partial \vv_t(\vtheta)'}{\partial \vtheta}\Big\rvert_{\vtheta = \vtheta_0}  \mF_0^{{[n]}^{-1}} \vv_t(\vtheta_0), \mathcal{F}_t \right)$ is a MDS since the partial derivative is $\mathcal{F}_{t-1}$-measurable. Moreover, both terms in the score vector \eqref{eq:sn0}, $\sum_{t=1}^n \left(\vv_t(\vtheta_0) \vv_t(\vtheta_0)' -  \mF_0^{[n]}\right)$ and $\sum_{t=1}^n \frac{\partial \vv_t(\vtheta)'}{\partial \vtheta}\Big\rvert_{\vtheta = \vtheta_0}  \mF_0^{{[n]}^{-1}} \vv_t(\vtheta_0)$, become independent asymptotically. Thus, we obtain comparable results as  \citet[p.~234]{ChaMilPa2009}.

\subsection{Asymptotic distribution of the ML estimator for $b_0 < 1/2$}
\noindent
For $b_0 < 1/2$ the asymptotic properties of the ML estimator 
for ARFIMA processes in the time domain have already been established \citep[cf.\ e.g.\ ][]{Ber1995, Rob2006}. In the asymptotically stationary case, we can show that their results carry over to unobserved components models. 

\noindent 
To derive the asymptotic distribution of the ML estimator for $b_0 < 1/2$, we use a central limit theorem (CLT) for MDS that applies to $\partial \vv_t(\vtheta)'/\partial \vtheta \big\rvert_{\vtheta=\vtheta_0}\mF_0^{[n]^{-1}}\vv_t(\vtheta_0)$ since the partial derivatives at $\vtheta_0$ are asymptotically stationary. Furthermore we show convergence in distribution for the first term in \eqref{eq:sn0}. Lemma \ref{L3:c} in appendix \ref{AS3} summarizes the results for both terms. A martingale CLT for the gradient  \eqref{eq:sn0} then yields $\frac{1}{\sqrt{n}} s_n(\vtheta_0) \dto N(0, \mathcal{J}_0)$, where $\mathcal{J}_0$ is the limiting information matrix \citep[][eq. 11.3.11]{Dav2000}. Asymptotic independence of both stochastic terms in \eqref{eq:sn0} facilitates the computation of $\mathcal{J}_0$. Finally, from \citet[][eq. 11.3.15]{Dav2000} a CLT for the ML estimator $\hat{\vtheta}_n$ follows as shown in theorem \ref{th:4}. 

\begin{theorem}\label{th:4}
	For $b_0 \in [0, 1/2)$ the maximum likelihood estimator is consistent and asymptotically normally distributed 
	$
	\sqrt{n} ( \hat{\vtheta}_n- \vtheta_0) \dto \mathrm{N}(0, \mathcal{J}_0^{-1})$  $\text{as } n \to \infty,
	$
	with
	\begin{align*}
	\mathcal{J}_0&=  \plim_{n \to \infty}\frac{1}{n} \sum_{t=1}^{n} \frac{\partial \vv_t(\vtheta)'}{\partial \vtheta} \Bigg \rvert_{\vtheta = \vtheta_0} \mF_0^{-1} \frac{\partial \vv_t(\vtheta)}{\partial \vtheta'}\Bigg\rvert_{\vtheta = \vtheta_0} \\
	&+ \frac{1}{2}\left[\frac{\partial (\vec \mF)'}{\partial \theta}\Bigg \rvert_{\vtheta=\vtheta_0}(\mF_0^{-1} \otimes \mF_0^{-1})  \frac{\partial \vec \mF}{\partial \theta'}\Bigg \rvert_{\vtheta=\vtheta_0}\right].
	\end{align*}
\end{theorem}
\noindent
The proof is contained in appendix \ref{AS3}. 
\noindent

\subsection{Asymptotic distribution of the ML estimator for $b_0>1/2$}


Having shown that the ML estimator is asymptotically normal for $b_0 < 1/2$, we turn to the nonstationary case $b_0 \in (1/2, 3/2)$. Then the usual MDS CLT does not apply since by lemma \ref{L3a} the derivative of $\vv_t(\vtheta)$ at $\vtheta_0$ is a nonstationary process. 
Inference for a broad class of (potentially) nonstationary models is considered in \citet[][sections 8 and 11]{Woo1994}, where sufficient conditions for consistency and asymptotic (mixed) normality of the ML estimator are derived. \cite{ChaMilPa2009} extend this setup by including a rotation matrix $\mA$. Their setup also nests our fractional trend model and allowed \cite{ParPhi2001} to study  the asymptotic behavior of the NLS estimator for nonlinear cointegration models. It requires to consider the following three sufficient conditions:
\begin{itemize}
	\item[ML1:] $\vnu_n^{-1} \mA' \vs_n(\vtheta_0)  \dto \mN$ as $n \to \infty$,
	\item[ML2:] $-\vnu_n^{-1} \mA' \mH_n(\vtheta_0) \mA \vnu_n^{-1'}  \dto \mM$ a.s.\ as $n \to \infty$ with $\mM$ positive definite with probability one  and
	\item[ML3:] there exists a sequence of invertible normalization matrices $\vmu_n$ such that $\vmu_n \vnu_n^{-1} \to 0$ a.s.\ and \begin{align*}
		\sup_{\vtheta \in \mTheta_n} \lvert \lvert \vmu_n^{-1} \mA'\left( \mH_n(\vtheta) - \mH_n(\vtheta_0) \right)\mA \vmu_n^{-1'}\rvert \rvert \pto 0,
	\end{align*}
	where $\mTheta_n=\left\{\vtheta \big| ||\vmu_n' \mA^{-1} (\vtheta - \vtheta_0) || \leq 1\right\}$ is a sequence of shrinking neighborhoods of $\vtheta_0$.
\end{itemize}
The random matrices $\mM$ and $\mN$ and the nonstochastic matrices $\mA$ and $\vnu_n$ will be defined below in \eqref{eq:A_N_A_S_A_D_nu_n}, and $\vmu_n$ in the proof of lemma \ref{ML:3}. As in the $I(1)$ case considered in \cite{ChaMilPa2009}, under conditions ML1 to ML3, equation \eqref{eq3:1} converges as $n \to \infty$
\begin{align}\label{normality}
	\vnu_n' \mA^{-1} (\hat{\vtheta}_n - \vtheta_0) &= - \left[\vnu_n^{-1} \mA'  \mH_n(\vtheta_0) \mA \vnu_n^{-1'}\right]^{-1} \left[ \vnu_n^{-1} \mA' \vs_n(\vtheta_0) \right] + o_p(1) \dto \mM^{-1} \mN.
\end{align}
Showing that ML1 to ML3 hold, such that \eqref{normality} follows, is the subject of the remaining section, where we proceed as follows. To distinguish between $I(b_0)$ and $I(0)$ processes we first derive an expression for the rotation matrix $\mA$. Lemma \ref{L3:d} contains a functional central limit theorem (FCLT) for the different components in $\mA' \vs_n(\vtheta_0)$, which directly yields the entries of the scaling matrix $\vnu_n$. Finally,  in lemmas \ref{ML:1} to \ref{ML:3} we prove that ML1 to ML3 hold and thus \eqref{normality}. Theorem \ref{th:cons} summarizes the results and defines $\mM$, $\mN$.

\noindent
As lemma \ref{L3a} shows, the partial derivative w.r.t.\ $\vbeta$ at $\vtheta_0$ is the only source of fractional integration in the partial derivatives of $\vv_t(\vtheta)$ at $\vtheta_0$, whereas the partial derivatives w.r.t.\ $\vec \mSigma$ and $b$ are $I(0)$. 
\noindent
Again following \cite{ChaMilPa2009}, to distinguish between $I(0)$ and $I(b_0)$ components, let the rotation matrix be defined as $\mA = \bmat \mA_N & \mA_S & \mA_D\emat$, where $\mA_N$ is $k \times (p-1)$, $\mA_S$ is $k \times (1+p(p+1)/2)$ and $\mA_D$ is $k \times 1$, where $k=p + p(p+1)/2 +1$ is the dimension of $\vtheta$. The scaling matrix $\nu_n$ adjusts for different convergence rates
\begin{align} \label{eq:A_N_A_S_A_D_nu_n}
\mA_N = \bmat \mGamma_0 \\ 0 \\ 0\emat, \ \ \mA_S = \bmat \frac{\vbeta_0}{(\vbeta_0'\mSigma_0^{-1} \vbeta_0)^{1/2}} & 0 \\ 0 & \mI \\ 0 & 0\emat, \ \ \mA_D = \bmat 0\\0  \\ 1 \emat, \ \ \vnu_n^{} = \bmat  n^{b_0}  \mI_{p-1} & 0 \\ 0 & n^{1/2}\mI_{k-p+1} \emat.
\end{align}
From lemma \ref{L3a} and the properties of $\mGamma_0$ in \eqref{gamma:prop} is easy to see that
\begin{align}\label{gammax}
\mA_N' \frac{\partial \vv_t(\vtheta)'}{\partial\vtheta}\Bigg\rvert_{\vtheta = \vtheta_0} = -\mGamma_0' x_t &+ \mGamma_0'a_\vbeta^0(\vu_t, \eta_t) \sim I(b_0)
\end{align}
whereas $\mA_S'\frac{\partial \vv_t(\vtheta)'}{\partial\vtheta}\big\rvert_{\vtheta = \vtheta_0}$, $\mA_D'\frac{\partial \vv_t(\vtheta)'}{\partial\vtheta}\big\rvert_{\vtheta = \vtheta_0}$ are $I(0)$.

\noindent
To derive the distribution properties of $\mM$, $\mN$ in \eqref{normality} we define the partial sums
\begin{align*} 
&\mU_n(r) = \frac{1}{\sqrt{n}} \sum_{t=1}^{\lfloor nr \rfloor} \mF_0^{{[n]}^{-1}}\vv_t(\vtheta_0),  
&&\mW_n(r) = \frac{1}{\sqrt{n}} \sum_{t=1}^{\lfloor nr \rfloor} \mA_S' \frac{\partial \vv_t(\vtheta)'}{\partial\vtheta}\Bigg\rvert_{\vtheta=\vtheta_0} \mF_0^{{[n]}^{-1}} \vv_t(\vtheta_0) , 	\\ 
&Y_n(r) = \frac{1}{\sqrt{n}} \sum_{t=1}^{\lfloor nr \rfloor} \mA_D' \frac{\partial \vv_t(\vtheta)'}{\partial\vtheta}\Bigg\rvert_{\vtheta=\vtheta_0} \mF_0^{{[n]}^{-1}} \vv_t(\vtheta_0), 
&&\mX_n(r) = \frac{1}{n^{b_0-1/2}} \sum_{t=1}^{\lfloor nr \rfloor} \mA_N' \Delta\frac{\partial \vv_t(\vtheta)'}{\partial\vtheta}\Bigg\rvert_{\vtheta=\vtheta_0}, 
\end{align*}
and 
\begin{align*}
\mV_n = \frac{1}{n^{b_0}} \sum_{t=1}^n \mA_N' \frac{\partial \vv_t(\vtheta)'}{\partial\vtheta}\Bigg\rvert_{\vtheta=\vtheta_0} \mF_0^{{[n]}^{-1}} \vv_t(\vtheta_0). 
\end{align*}
Since multiplication with $\mA_S'$ and $\mA_D'$ eliminates the nonstationary part of $\frac{\partial \vv_t(\vtheta)'}{\partial \vtheta}\Big\rvert_{\partial \vtheta = \vtheta_0}$ and  $\frac{\partial \vv_t(\vtheta)'}{\partial \vtheta}\Big\rvert_{\vtheta = \vtheta_0} \mF_0^{{[n]}^{-1}}  \vv_t(\vtheta_0)$ is a MDS, the FCLT of \citet[Lemma 3.3]{ChaMilPa2009} carries over directly for $\mU_n(r)$, $\mW_n(r)$ and $\mY_n(r)$. For $\mX_n(r)$  that contains nonstationary fractionally integrated common components we extend their FCLT in the following lemma where $\Rightarrow$ denotes weak convergence.
\begin{lemma}\label{L3:d}
	For $b_0 \in (1/2, 3/2)$ the following FCLT holds for the partial sums 
	\begin{align*}
		(\mU_n(r), \mW_n(r), \mY_n(r), \mX_n(r)) \Rightarrow (\mU(r), \mW(r), \mY(r), \mX(r)) 
	\end{align*}
	as $n \to \infty$ where $\mU(\cdot)$, $\mW(\cdot)$, $Y(\cdot)$ are multivariate Brownian motions, whereas $\mX(\cdot)$ is fractional Brownian motion of type II that is independent from $\mU(r)$. Furthermore, $\mV_n \dto \mV = \int_{0}^{1} \mX(r) \rd \mU(r)$ and has full rank a.s.
\end{lemma}
\noindent
The proof is contained in appendix \ref{AS3}.
\noindent
Denoting in the sequel $W(1)$ by $W$ and $Y(1)$ by $Y$, one has 
\begin{align}
\Var(\mW) &= \plim_{n \to \infty} \mA_S'\left(\frac{1}{n} \sum_{t=1}^{n} \frac{\partial \vv_t(\vtheta)'}{\partial \vtheta}\Bigg\rvert_{\vtheta=\vtheta_0}\mF_0^{-1}\frac{\partial \vv_t(\vtheta)}{\partial \vtheta'}\Bigg\rvert_{\vtheta=\vtheta_0}\right) \mA_S, \label{var:W}\\
\Var(Y) &= \plim_{n\to \infty} \mA_D'\left(\frac{1}{n} \sum_{t=1}^{n} \frac{\partial \vv_t(\vtheta)'}{\partial \vtheta}\Bigg\rvert_{\vtheta=\vtheta_0}\mF_0^{-1}\frac{\partial \vv_t(\vtheta)}{\partial \vtheta'}\Bigg\rvert_{\vtheta=\vtheta_0}\right) \mA_D.\label{var:Y}
\end{align}

\noindent
With the FCLT of lemma \ref{L3:d} at hand, lemmas \ref{ML:1} to \ref{ML:3} prove that the conditions ML1 to ML3 hold. They are contained in appendix \ref{AS3}. The following theorem summarizes the results by stating the asymptotic properties of the maximum likelihood estimator.
\begin{theorem}\label{th:cons}
	The ML estimator for model \eqref{dgp:1} satisfies  for $b_0 \in (1/2, 3/2)$,
	\begin{align*}
		\vnu_n' \mA^{-1} \left(\hat{\vtheta}_n - \vtheta_0\right) \dto \mM^{-1} \mN,
	\end{align*}
	where $\mN = \bvec - (\int_{0}^{1} \mX(r)\rd \mU(r))' & \mZ'-\mW' &	Q - Y\evec'$,
	\begin{align}\label{eq:M}
		\mM &= \bmat \int_{0}^{1} \mX(r) \mF_0^{-1} \mX'(r) \rd r & 0 & 0 \\
		0 & \Var(\mZ) + \Var(\mW)& 0 \\
		0 & 0 & \Var(Q) + \Var(Y) \emat,
	\end{align}
\begin{align}
	\mZ_n &= \frac{1}{2} \mA_S' \left[\frac{\partial (\vec \mF)'}{\partial \theta}\Bigg \rvert_{\vtheta=\vtheta_0}(\mF_0^{-1} \otimes \mF_0^{-1})\vec \left(\frac{1}{\sqrt{n}}\sum_{t=1}^{n} (\vv_t(\vtheta_0) \vv_t(\vtheta_0)' -\mF_0)\right)\right] \dto Z, \label{eq:Z} \\
	Q_n &= \frac{1}{2} \mA_D' \left[\frac{\partial (\vec \mF)'}{\partial \theta}\Bigg \rvert_{\vtheta=\vtheta_0}(\mF_0^{-1} \otimes \mF_0^{-1})\vec \left(\frac{1}{\sqrt{n}}\sum_{t=1}^{n}(\vv_t(\vtheta_0) \vv_t(\vtheta_0)' -\mF_0)\right)\right] \dto Q, \label{eq:Q}
\end{align}
as $n \to \infty$ with $Z \sim N(0, \Var(Z))$, $Q \sim N(0, \Var(Q))$, 
\begin{align}
	\Var(\mZ)&=\frac{1}{2} \mA_S' \left[\frac{\partial (\vec \mF)'}{\partial \theta}\Bigg \rvert_{\vtheta=\vtheta_0}(\mF_0^{-1} \otimes \mF_0^{-1})  \frac{\partial (\vec \mF)}{\partial \theta'}\Bigg \rvert_{\vtheta=\vtheta_0}\right] \mA_S, \label{VarZ}\\
	\Var(Q)&=\frac{1}{2} \mA_D'  \left[\frac{\partial (\vec \mF)'}{\partial \theta}\Bigg \rvert_{\vtheta=\vtheta_0}(\mF_0^{-1} \otimes \mF_0^{-1})  \frac{\partial (\vec \mF)}{\partial \theta'}\Bigg \rvert_{\vtheta=\vtheta_0}\right] \mA_D. \label{VarQ}
\end{align}
$\Var(W)$, $\Var(Y)$ are given in \eqref{var:W} and \eqref{var:Y}. 
\end{theorem}
\noindent
Define
	$
		\bvec \mR ' & \mS' \evec' = \left[\Var (\mZ) + \Var (\mW) \right]^{-1}(\mZ - \mW).
$
	Then it follows from theorem \ref{th:cons} that
	\begin{align}
		\frac{\vbeta_0' \mSigma_0^{-1}}{(\vbeta_0' \mSigma_0^{-1} \vbeta_0)^{1/2}}	\left[\sqrt{n}(\hat{\vbeta} - \vbeta_0)\right] &\dto \mR	, \label{conv:1}\\
		\mGamma_0'\mSigma_0^{-1} \left[ n^{b_0} (\hat{\vbeta} - \vbeta_0) \right] & \dto -\left[  \int_{0}^{1} \mX(r) \mF_0^{-1} \mX'(r) \rd r \right]^{-1} \int_{0}^{1} \mX(r)\rd \mU(r),\label{conv:2}	\\
				\sqrt{n} \left(\vech \hat{\mSigma} - \vech \mSigma_0 \right) & \dto \mS,  \label{eq:hatSigma_asympDist} \\
		\sqrt{n} (\hat{b} - b_0) &\dto \left[\Var(Q) + \Var(Y)\right]^{-1}(Q - Y).\label{eq:hat_b_asympDist} 
	\end{align}

\noindent
\citet[p.~236]{ChaMilPa2009} conclude from their counterpart of theorem \ref{th:cons} that
\begin{align} \label{eq:hat_beta_asympDist}
\sqrt{n} \left(\hat{\vbeta} - \vbeta_0\right) \dto \frac{\vbeta_0}{(\vbeta_0' \mSigma_0^{-1} \vbeta_0)^{1/2}} \mR.
\end{align}

\noindent
To show this, multiply \eqref{conv:1} by $\vbeta_0/(\vbeta_0' \mSigma_0^{-1} \vbeta_0)^{1/2}$ and then insert 
\eqref{eq:Px} to obtain
\begin{align*}
		\frac{\vbeta_0\vbeta_0' \mSigma_0^{-1}}{\vbeta_0' \mSigma_0^{-1} \vbeta_0}	\left(\sqrt{n}(\hat{\vbeta} - \vbeta_0)\right) 
 = \left(\sqrt{n}(\hat{\vbeta} - \vbeta_0)\right) - \mP_x' \left(\sqrt{n}(\hat{\vbeta} - \vbeta_0)\right) & \dto \frac{\vbeta_0}{(\vbeta_0' \mSigma_0^{-1} \vbeta_0)^{1/2}} \mR.
\end{align*}
Using $\mP_x' = \mGamma_0 \mGamma_0' \mSigma_0^{-1}$, the second term converges to zero in probability for $n\to\infty$ and $b_0>1/2$ since from \eqref{conv:2} one has $ \mGamma_0' \mSigma_0^{-1} \left(n^{b_0}(\hat{\vbeta} - \vbeta_0)\right) = O_p(1)$. 

\noindent
From theorem \ref{th:cons} it follows directly that the maximum likelihood estimator for $\vtheta$ is consistent and asymptotically normal. As in the $I(1)$ model of \cite{ChaMilPa2009}, the estimator for $\vtheta$ converges at rate $\sqrt{n}$ with one particular exception. $\mGamma_0'\mSigma_0^{-1}\hat{\vbeta}$ converges at rate $n^{b_0}$ and is mixed normally distributed. Recall that the rotation $\mGamma_0'\mSigma_0^{-1}$ is the cointegrating matrix as it projects out the common fractional trend  $\mGamma_0'\mSigma_0^{-1}\vbeta_0x_t = 0$. Therefore, the faster convergence rate for the cointegrating matrix in error-correction models carries over to the fractionally integrated unobserved components model.
Additionally, theorem \ref{th:cons} shows that the standard inference results, which were shown to be valid for nonstationary $I(1)$ trends in state space models by \cite{ChaMilPa2009}, remain valid when the persistence of the common component is generalized to the nonstationary fractional domain. Due to  \eqref{eq:hatSigma_asympDist}, \eqref{eq:hat_b_asympDist}, and \eqref{eq:hat_beta_asympDist} the information matrix equality holds asymptotically. Thus, an estimate for the parameter covariance matrix can be obtained from the negative inverse of the  Hessian matrix computed in the numerical optimization.

\noindent
In a nutshell, the ML estimator is consistent for $b \in D$. It converges to the normal distribution as $n \to \infty$ whenever $b_0 < 1/2$, as shown in theorem \ref{th:4}. For $b_0 \in (1/2, 3/2)$ theorem \ref{th:cons} states that the ML estimator is asymptotically normally distributed where a particular rotation of the parameter vector exhibits an asymptotically  mixed normal distribution. Thus,  $t$-ratios for parameter significance and asymptotic tests such as the likelihood ratio test, the Wald test, and the LM test, remain valid in the fractionally integrated UC model within the two distinct intervals in $D$. Therefore, our results for the nonstationary region generalize the statement of \cite{ChaMilPa2009} for the $I(1)$ case. Based on simulation results,  \cite{HarWei2018} report good finite sample performance of the ML estimator for fractionally integrated UC models.

\section{Fractional trends in US inflation}\label{Ch:4}
We apply our fractional UC model to extract a common long-run component from three inflation measures for the US, the consumer price index (CPI), the personal consumption expenditures index (PCI), and the producer price index (PPI). The literature has so far only considered an $I(1)$ common component in US inflation \citep[cf.\ e.g.][]{DomGom2006, StoWat2016} that was interpreted as long-run or core inflation. 
We contribute to the literature by investigating whether the $I(1)$ assumption for the long-run component holds.
Furthermore, we show how estimates for the long-run component $x_t$ together with its fundamental shocks $\eta_t$ are affected if fractional integration is allowed for. If the $I(1)$ assumption for the long-run component is violated in the $I(1)$ UC model, then the asymptotic results of \cite{ChaMilPa2009} are not applicable. In that case the fractional UC model provides valid inferential results, as it covers integration orders $b \in D$.

\noindent
The data was downloaded from the Federal Reserve Bank of St.\ Louis (mnemonics: CPIAUCSL, PCEPI, WPSFD49207), is in monthly frequency and spans from 1961:1 to 2018:12. The three series were generated via log differences
\begin{align*}
	\pi_{i, t} = 100 \times \Delta \log price_{i, t},
\end{align*}
where $i \in \{CPI, PCI, PPI\}$ indexes the inflation measures. 
Since all three series intend to measure price growth for the US, we model them as a function of one common scalar long-run component $x_t$, which in our case is a fractionally integrated trend, and three uncorrelated idiosyncratic components $\vu_t$
\begin{align}
	\bvec \pi_{CPI,t} \\ \pi_{PCI,t} \\ \pi_{PPI, t} \evec = \bvec 1\\ \beta^{PCI}\\ \beta^{PPI}\evec x_t + \bvec u^{CPI}_t \\ u^{PCI}_t \\ u^{PPI}_t \evec. 
\end{align}
This implies a cointegration rank $r = p - 1 = 2$ among the inflation measures, which is confirmed by the sequential likelihood ratio test for fractional time series of \cite{JohNie2012} that clearly rejects the null hypothesis for $r=1$ (p-value $0.001$) but fails to reject for $r=2$ (p-value $0.102$). Furthermore, we allow for $\Var(\eta_t) = \sigma_\eta^2 \neq 1$ and restrict $\vbeta^{CPI}$ to one  for unique identification of $x_t$. Since the standard errors of the three inflation measures differ considerably, we allow for $\vbeta_{PCI} \neq 1$ and $\vbeta_{PPI} \neq 1$.

\noindent
We enrich our ARMA approximation of the fractionally integrated process $x_t$ by additional AR coefficients, which does not affect the asymptotic properties of the ML estimator but reduces the approximation error. Since choosing the same lag order for the AR and the MA polynomial is computationally efficient, as any AR polynomial of length less or equal to $m$ does not affect the dimension of the state vector, we use ARMA($m$, $m$) approximations in the following. As \cite{HarWei2018} demonstrate in a simulation study, setting $m \geq 3$ yields an approximation error that is hardly visible. Therefore, we consider ARMA($4, 4$) approximations in the following.
Since the Wold representation of an ARMA process $a(L)\tilde{x}_t = b(L)\eta_t$ is given by $\tilde{x}_t = a(L)^{-1}b(L) \eta_t = \psi(L)\eta_t$ the approximation error becomes
\begin{align*}
	\tilde{\epsilon}_{t+1}(\vtheta) &=
	\begin{cases}
	\sum_{i=1}^{t} (\varphi_i-\psi_i) \ve_{t+1-i} \mSigma_{\eta_{1:t}\mY_t} \mSigma_{\mY_t}^{-1} \mY_t & \text{if } b < 1,\\
	\sum_{s=1}^{t}\sum_{i=1}^{s} (\varphi_i-\psi_i) \ve_{s+1-i} \mSigma_{\eta_{1:t}\mY_t} \mSigma_{\mY_t}^{-1} \mY_t & \text{if } b \geq 1,
	\end{cases}
\end{align*}
and is again $\mathcal{F}_{t}$-measurable. 

\noindent
Technically, for a fixed $b$, the ARMA coefficients in $a(L)$, $b(L)$, and thus $\psi(L)$, are obtained beforehand by minimizing the mean squared error between the Wold representations of $\tilde{x}_t$ and $x_t$. A continuous function that maps from the integration order $b$ to the ARMA coefficients is
then constructed by first optimizing over a grid of $b$ and second smoothing the ARMA coefficients over $b$ using splines. Hence, optimization of the likelihood for the fractionally integrated UC model is conducted over the scalar fractional integration order $b$ and does not involve the estimation of any parameters in $a(L)$, $b(L)$. This procedure keeps the dimension of the parameter vector $\vtheta$ small during the optimization. Further details together with simulation results are contained in \cite{HarWei2018}.

\noindent
Starting values for the 
ML estimator of $\vtheta$ are obtained by drawing $1000$ combinations of initial values for $b$, $\vbeta$, and $\mSigma$ from uniform distributions with appropriate support and maximizing the likelihood while ignoring the approximation error. As \cite{HarWei2018} show, this procedure already yields quite precise estimates for the unknown parameters and is computationally fast. The optimized parameters corresponding to the largest likelihood are then taken as starting values for the approximation-corrected 
ML estimator. For an unconstrained optimization, we use a matrix logarithm parametrization for the covariance matrices. Standard errors are denoted in parentheses.

%
\noindent
For the loadings we estimate $\hat{\vbeta} = \bvec 1  & \underset{(0.015)}{0.816} & \underset{(0.058) }{1.245} \evec'$, which reflects the heterogeneous volatility of the three inflation measures. The integration order estimate $\hat{b}=0.476$ $(0.030)$ is in line with the literature, where e.g.\  \cite{HasWol1995} estimate an integration order of $0.41$ for US CPI inflation from 1969:1 to 1992:12, while \cite{Bai1996} estimates $\hat{b}=0.47$ for US CPI inflation from 1948:1 to 1990:7. Hence, there is substantial evidence for long-run inflation being mean-reverting and integrated of order around $1/2$. Our estimated integration order of $0.476$ implies that a unit shock still has more than $14\%$ of its initial impact on inflation after one year, and more than $4\%$ of its initial impact after ten years. The variance estimates for the fundamental shocks $\eta_t$, $\vu_t$ are $\log \hat{\sigma}_{\eta}^2 = -3.275$ $(0.065)$,   
$\log \hat{\sigma}_{\vu_{CPI}}^2 = -4.374$ $(0.098)$,
$\log \hat{\sigma}_{\vu_{PCI}}^2 = -5.839$ $(0.232)$, and
$\log \hat{\sigma}_{\vu_{PPI}}^2 = -1.782$  $(0.058)$, implying 
$\hat{\sigma}_{\eta}^2=0.0378$, $ \hat{\sigma}_{\vu_{CPI}}^2 = 0.013$, $\hat{\sigma}_{\vu_{PCI}}^2=0.003$, and $\hat{\sigma}_{\vu_{PPI}}^2 =0.168$. These estimates reflect the relatively high idiosyncratic volatility of the producer price index series, compared to CPI and PCI. 
The log likelihood is $311.278$.
Our results furthermore indicate that the $I(1)$ assumption for the long-run component is likely to be violated.

\noindent
As a benchmark we also report results based on the fractionally cointegrated VAR (FCVAR) model of \cite{JohNie2012}. Note that the two models are not nested, since they specify the fundamental shocks differently. For the FCVAR model, we estimate an integration order $\hat{b}_{FCVAR}=0.394$ $(0.025)$ that is somewhat smaller than the one obtained from our fractionally integrated unobserved components model but provides additional evidence against the $I(1)$ assumption for inflation. 
The smaller estimated integration order for the FCVAR model may be explained by the findings of \citet{SunPhi2004} who show that an additive $I(0)$ term can downward-bias the estimated integration order when the $I(0)$ term is not properly included in the model. Furthermore, we can calculate an estimate for $\vbeta$ from the orthogonal complement of the cointegrating vector of the FCVAR model and obtain $\hat{\vbeta}_{FCVAR} = \bvec 1 & 0.904 & 1.052 \evec' $. Again, the results obtained from the FCVAR model slightly differ from the fractionally integrated unobserved components model but point to a similar direction. 

\begin{figure}[h!]
	\begin{center}\vspace{-0.5 cm}
		\includegraphics[width=\textwidth, trim={0cm 0cm 0cm 0cm}]{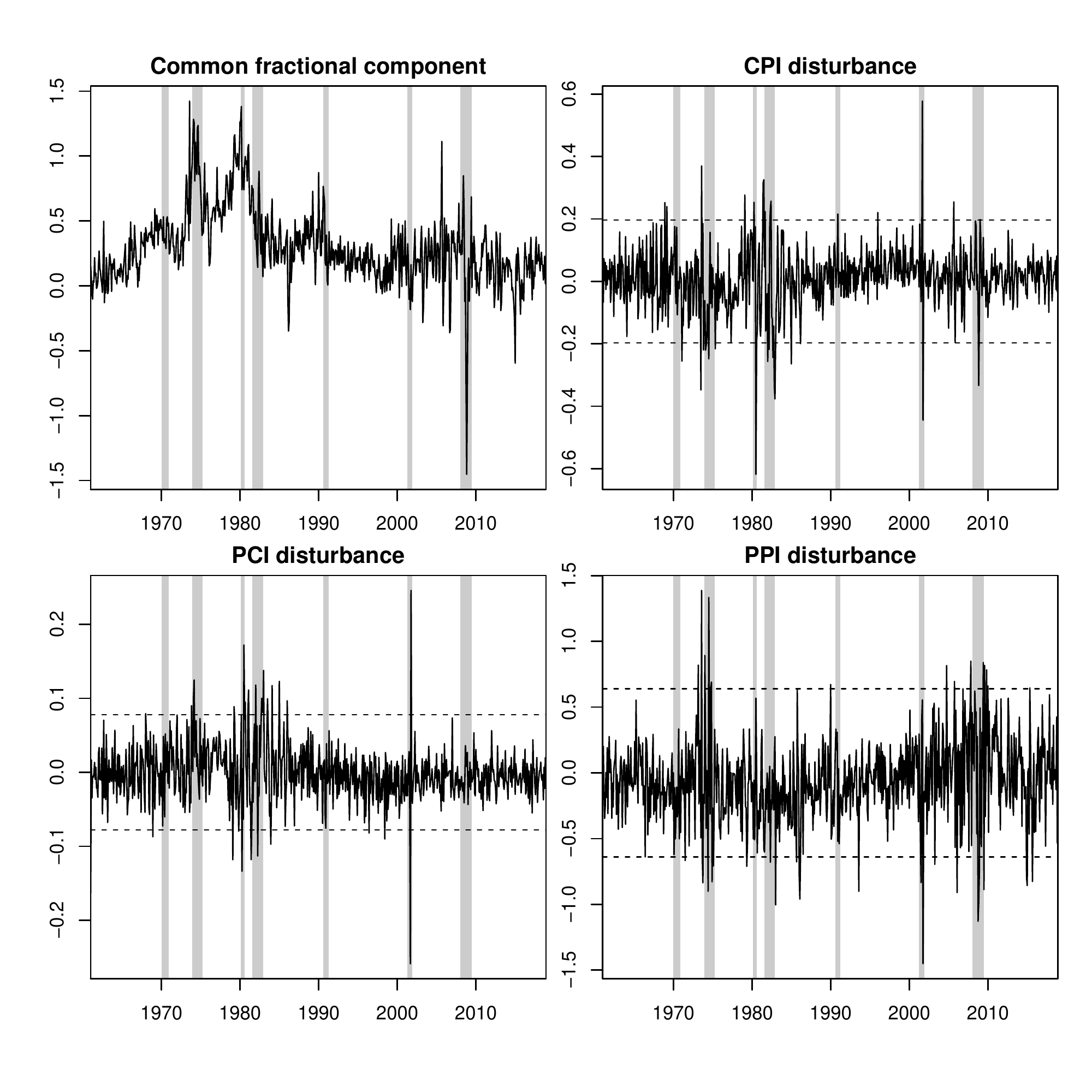}
	\end{center}\vspace{-1cm} \caption{Common fractional component and $I(0)$ idiosyncratic disturbances of US consumer price index, personal consumption expenditures: chain index, and producer price index. Shaded areas correspond to NBER recession periods.} \label{fig:irarma}
\end{figure}

\noindent
Figure \ref{fig:irarma} sketches the dynamics of the estimated common fractionally integrated component $\hat{x}_t$ and the idiosyncratic disturbances $\hat{\vu}_t$ together with two standard deviations (dashed). 
As one can see, the common component captures the dynamics of the three inflation measures well. Due to the long memory property, mean-reversion can take quite a long time, as the 1970s and the second half of the 1980s show. The disturbance terms seem to be $I(0)$, such that the long-run dynamics of the three inflation measures are well described by one common fractionally integrated trend component and, therefore, two fractional cointegration relations exist. 
As the figure shows, $\vu_t$ may be heteroskedastic and even autocorrelated. These features could be included into the model and
 we leave this challenge open for future research.

\begin{figure}[h!]
	\begin{center}\vspace{-0.5 cm}
		\includegraphics[width=\textwidth, trim={0cm 0cm 0cm 0cm}]{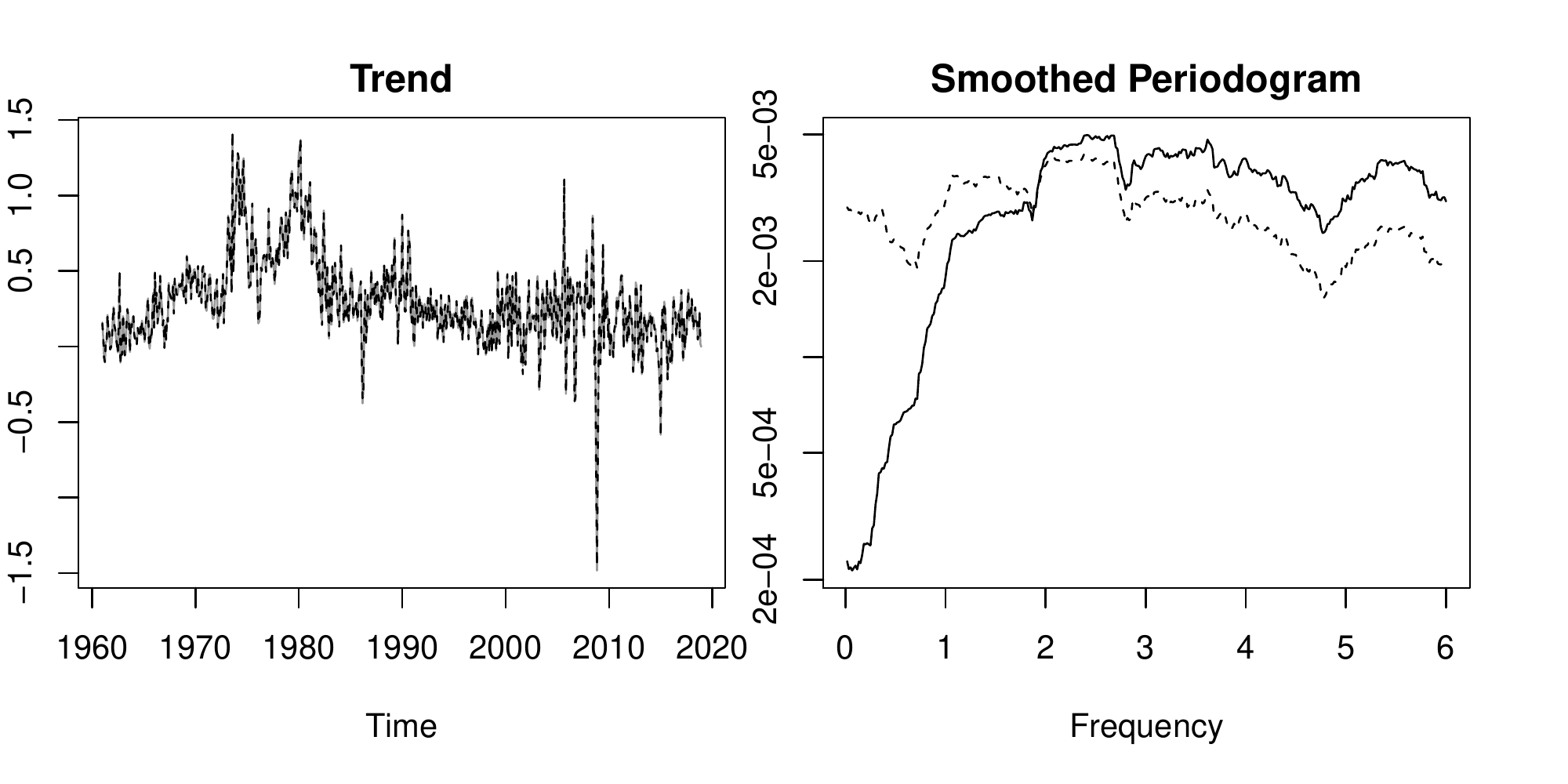}
	\end{center}\vspace{-1cm} \caption{Common trend and smoothed periodogram of the fundamental shock series for the $I(1)$ common trend model (solid) and the fractional trend model (dashed).} \label{fig:i1}
\end{figure}

\noindent
We compare our results with the $I(1)$ UC model that was studied in \cite{ChaMilPa2009} by estimating the latter as a benchmark. 
While we obtain similar estimates for the loadings in $\vbeta$, the log likelihood of the $I(1)$ UC model is $234.322$ and hence clearly smaller than in the fractionally integrated setup. Figure \ref{fig:i1} plots the long-run component estimate from the $I(1)$ UC model for US inflation together with the fractional trend estimate on the left-hand side. The other graph shows the periodogram for the two fundamental shock series that drive the long-run components and are assumed to follow Gaussian white noise processes in both models.

\noindent
As the graphs show, the two trend estimates are very similar, although the solid line was generated by an $I(1)$ filter, that is an unweighted sum of past shocks, whereas the dashed line was generated by a fractional filter with $b=0.476$ that assigns decreasing weights to $\hat{\eta}_{t-h}$ as $h$ increases. The similarity of the two processes results from a violation of the white noise assumption for the fundamental shocks of the $I(1)$ UC model: As the periodogram shows, these shocks exhibit a zero at the origin, which indicates anti-persistence, whereas the periodogram of the fundamental shocks for the fractional unobserved components model does not show such violations of the white noise assumption. In addition, the exact local Whittle estimator (with $m=n^{0.65}$ as in \cite{ShiPhi2005}) suggests an integration order of $-0.486$ for the fundamental shocks of the $I(1)$ trend (and $0.00$ for those of the $I(d)$ trend). Applying an $I(1)$ filter to an anti-persistent shock series with integration order $-0.486$ produces a series that is integrated of order $0.514$, instead of an $I(1)$ trend.

\noindent
Estimating a misspecified $I(1)$ common trend model for US inflation therefore pollutes the fundamental shock estimates and leads to wrong conclusions about their persistence. Since inflation shocks are misleadingly assumed to exhibit a permanent impact, the $I(1)$ model produces incorrect impulse responses, whereas the $I(d)$ model captures the mean-reverting nature of inflation via the impulse response function correctly. 

\noindent
Since the Gaussian white noise assumption for the fundamental shocks is crucial for consistency and asymptotic normality of the ML estimator of \cite{ChaMilPa2009}, a violation may yield inconsistent parameter estimates and incorrect inference. Thus, for US inflation we find that a fractional common component should be considered instead of an $I(1)$ trend component. In general, the fundamental shocks of the permanent component should be checked for (anti-)persistence. 

\noindent
We expect further consequences in the general multivariate $I(d)$ case that carry over from $I(1)$ UC models: If additional unobserved components are added to the model that correlate with the fundamental shocks, as e.g.\ in the correlated $I(1)$ UC model of \cite{MorNelZi2003} or the simultaneous UC model of \cite{Web2011}, a violation of the $I(1)$ assumption may produce spurious cycles and bias the estimates for the latent components.

\section{Conclusion}\label{Ch:5}

We propose a multivariate fractionally integrated unobserved components model and derive a computationally efficient modification of the Kalman filter to estimate a single, fractionally integrated common component. Furthermore, we show consistency and assess the asymptotic distribution of the maximum likelihood estimator  for integration orders $b \in D = \{d \in \mathbb{R}\ | \ 0 \leq d < 3/2,\ d \neq 1/2\}$, thereby generalizing the asymptotic results of \cite{ChaMilPa2009} for a common $I(1)$ component. As we show, the maximum likelihood estimator is asymptotically normally distributed whenever $b_0 < 1/2$. For $b_0 \in (1/2, 3/2)$  the maximum likelihood estimator is also asymptotically normal, however a particular rotation of the parameter vector, corresponding to the cointegrating matrix, exhibits an asymptotically mixed normal distribution with rate $n^{b_0}$. 
We apply our fractionally integrated unobserved components model to extract a long-run component from three US inflation series and obtain an estimated integration order of $0.476$ for the long-run component. Due to a violation of the $I(1)$ assumption  the widely applied $I(1)$ unobserved components model yields anti-persistent long-run shocks, while those from our fractionally integrated 
model appear to be in line with the model assumptions. 

\noindent
Future research could generalize our results to multiple common long-run components, potentially exhibiting different integration orders. Furthermore, a trend-cycle decomposition that allows for autocorrelated idiosyncratic shocks
may yield new insights with regard to common trends and cycles for macroeconomic time series. Finally, settings with dependent shocks, such as the correlated unobserved components model of \cite{MorNelZi2003} and the simultaneous unobserved components model of \cite{Web2011}, could be considered.

\section*{Acknowledgments}
The authors thank Uwe Hassler, Ulrich M\"uller, Morten {\O}.\ Nielsen, Christoph Rust, the participants of the econometric seminar in Nuremberg, the department seminar at the Christian Albrechts University Kiel, the DAGStat conference 2019 in Munich, the workshop on high-dimensional time series in economics and finance 2019 in Vienna, the Annual Meeting of the German Statistical Society 2019 in Trier, the Annual Meeting of the German Economic Society 2019 in Leipzig, the Seminar on International Economic Policy at the University of Zurich, the International Conference on Computational and Financial Econometrics 2019 in London, the Symposium in Honor of Michael Hauser at WU Vienna, and the Standing Field Committee in Econometrics of the German Economic Society for many valuable comments. 
This work was supported by the German Research Foundation (DFG) via the projects TS283/1-1 and WE4847/4-1.
\clearpage
\appendix
\clearpage
\section{Mathematical appendix}
\subsection{Proofs for section 2}\label{AS2}
\noindent
The following lemma is required for theorem \ref{th:1}.
\begin{lemma}\label{L:1}
	For the prediction error variance $\mF_t$ in \eqref{eq:F_t} it holds that
	\begin{align*}
		\vbeta' \mF_t^{-1}\vbeta &= \frac{\vbeta' \mSigma^{-1} \vbeta}{1+\vbeta'\mSigma^{-1}\vbeta \omega^{(1, 1)}_{t}}.
	\end{align*}
\end{lemma}
\begin{proof}[Proof of Lemma \ref{L:1}]
	From the inverse of the prediction error variance
	$\mF_t^{-1} = \mSigma^{-1} - \mSigma^{-1} \vbeta \vbeta' \mSigma^{-1} \omega^{(1,1)}_{t} \left( 1 + \vbeta'\mSigma^{-1}\vbeta \omega^{(1,1)}_{t} \right)^{-1}$,
	it follows that
$
	 \mF_t^{-1}\vbeta = \frac{\mSigma^{-1} \vbeta}{1+\vbeta'\mSigma^{-1}\vbeta \omega_{t}^{(1,1)}}.
	$
\end{proof}
\noindent
\begin{proof}[Proof of Theorem \ref{th:1}]
	Using \eqref{ap:up:x} of the exact state space representation
	$
		\valpha_{t+1|t} = T \valpha_{t|t-1} + T \mP_{t|t-1} \mZ' \mF_t^{-1} v_t(\vtheta),
	$
	and using $\vbeta'\mF^{-1}_t = \vbeta'\mSigma^{-1} \left(1+\vbeta'\mSigma^{-1}\vbeta \omega^{(1,1)}_{t}\right)^{-1}$ analogously to the result of lemma \ref{L:1}, the conditional expectation of ${x}_{t+1}$ is given by
	\begin{align} \label{eq:x_t+1_t}
		{x}_{t+1|t} &= \indic {x}_{t|t-1} + \alpha_{t|t-1}^{(2)} +\left(\indic\omega^{(1,1)}_{t} +  \omega^{(1,2)}_{t}\right)   \frac{\vbeta'\mSigma^{-1}}{1+\vbeta'\mSigma^{-1}\vbeta \omega^{(1,1)}_{t}}\vv_t(\vtheta). 
	\end{align}
	Next, we iterate $\alpha_{t|t-1}^{(2)}$ using \eqref{ap:up:x} and define $N_t = 1+\vbeta'\mSigma^{-1}\vbeta \omega^{(1,1)}_{t}$
	to obtain
	\begin{align}
		\alpha_{t|t-1}^{(2)} 
		& = \alpha_{1|0}^{(t+1)} + \sum_{j=0}^{t-2} \frac{\vbeta'\mSigma^{-1}\omega^{(1,3+j)}_{t-1-j}}{N_{t-1-j}}  (\vy_{t-1-j} - \vbeta x_{t-1-j|t-2-j} ). \label{eq:alpha_t_t-1}
	\end{align}
	After inserting \eqref{eq:alpha_t_t-1} into \eqref{eq:x_t+1_t} one has
	\begin{align}\label{eq:xnew}
		x_{t+1|t} 	&=\indic \left( {x}_{t|t-1} + \frac{\vbeta'\mSigma^{-1}\omega^{(1,1)}_{t}}{N_{t}}
		(\vy_t - \vbeta x_{t|t-1})\right) \notag \\
		&+ \alpha_{1|0}^{(t+1)} + \sum_{j=0}^{t-1} 
		\frac{\vbeta'\mSigma^{-1}\omega^{(1,2+j)}_{t-j}}{N_{t-j}}
		(\vy_{t-j} - \vbeta {x}_{t-j|t-j-1}),
	\end{align}
	where $\alpha_{1|0}^{(t+1)}=0$. 
	To unify the denominators we add and subtract $\indico [ x_{t|t-1} +   \frac{ \vbeta'\mSigma^{-1}}{N_n}  (\vy_t - \vbeta x_{t|t-1})] $ together with $\frac{\varphi_{j+1}(d_0)}{N_n} $ inside the sum of \eqref{eq:xnew}
		\begin{align*}
			x_{t+1|t} &= 
			\indico \left[ x_{t|t-1} +   \frac{ \vbeta'\mSigma^{-1}}{N_n}(\vy_t - \vbeta x_{t|t-1})\right]\\ &+\sum_{j=0}^{t-1} \frac{\varphi_{j+1}(d_0)\vbeta' \mSigma^{-1} }{N_n} (\vy_{t-j} - \vbeta x_{t-j|t-j-1}) + z_{1, t}(\vtheta),
		\end{align*}
		where
		\begin{align*}
			z_{1, t}(\vtheta) &=  \indic \left[x_{t|t-1}+ \frac{\omega_t^{(1,1)} \vbeta'\mSigma^{-1}}{N_t}(\vy_t - \vbeta x_{t|t-1})\right]  - \indico \Big[ x_{t|t-1} \\&+   \frac{ \vbeta'\mSigma^{-1}}{N_n}  (\vy_t - \vbeta x_{t|t-1})\Big]
			+\sum_{j=0}^{t-1} \left[ \frac{\omega_{t-j}^{(1, 2+j)}}{N_{t-j}} -  \frac{\varphi_{j+1}(d_0)}{N_n} \right] \vbeta' \mSigma^{-1} (\vy_{t-j} - \vbeta x_{t-j|t-j-1}).
		\end{align*}
	
	\noindent
	Subtracting $\indico x_{t|t-1}$ and using the fractional difference operator $\sum_{j=0}^{t-1}\varphi_{j+1}(d_0) \vy_{t-j} = (\Delta_+^{-d_0}-1) \vy_{t+1}$ gives
		$
			\Delta^{\indico}{x}_{t+1|t}
			= { \indico\vbeta' \mSigma^{-1}}{N_{n}}^{-1} (\vy_t - \vbeta x_{t|t-1}) +{\vbeta' \mSigma^{-1}}{N_{n}}^{-1}(\Delta^{-{d_0}}_+-1)(\vy_{t+1} - \vbeta \vx_{t+1|t}) + z_{1,t}(\vtheta).
		$
		By taking fractional differences $\Delta_+^{d_0}$ one has
		\begin{align*}
			\Delta^{b_0}_+{x}_{t+1|t}  &= \frac{ \indico\vbeta' \mSigma^{-1}}{N_{n}}\Delta_+^{d_0} (\vy_t - \vbeta x_{t|t-1}) +\frac{\vbeta' \mSigma^{-1}}{N_{n}}(1-\Delta^{d_0}_+)(\vy_{t+1} - \vbeta \vx_{t+1|t})\\
			&+\Delta^{d_0}_+ z_{1, t}(\vtheta)= \frac{\vbeta' \mSigma^{-1}}{N_{n}}(1-\Delta^{b_0}_+)(\vy_{t+1} - \vbeta \vx_{t+1|t}) +\Delta^{d_0}_+ z_{1, t}(\vtheta),
		\end{align*}
		where the last step follows from $\Delta_+^{d_0} L + (1-\Delta_+^{d_0}) = (1-L)^{d_0}_+L + 1 - (1-L)^{d_0}_+=1 - (1-L)^{d_0}_+(1-L) = 1 - (1-L)^{d_0 + 1}_+ = L_{d_0+1}$.
		Bringing all $x_{t+1|t}$ to the left-hand side 
		and solving for ${x}_{t+1|t}$ yields
		\begin{align*}
			{x}_{t+1|t} &= \left( 1 - L_{b_0} \frac{1 + \vbeta' \mSigma^{-1}\vbeta (\omega^{(1,1)}_{n}-1)}{N_{n}} \right)^{-1} \left\{  \frac{\vbeta' \mSigma^{-1}}{N_{n}}L_{b_0}\vy_{t+1} +\Delta^{d_0}_+ z_{1, t}(\vtheta) 	\right\} = \\
			&= \sum_{j=0}^{} \left( \frac{1 + \vbeta' \mSigma^{-1}\vbeta(\omega^{(1,1)}_{n}-1) }{N_{n}}L_{b_0}\right)^j \left\{  \frac{\vbeta' \mSigma^{-1}}{N_{n}}L_{b_0}\vy_{t+1} +\Delta^{d_0}_+ z_{1, t}(\vtheta) \right\} = \\
			&=\sum_{j=0}^{} \left( \frac{1 + \vbeta' \mSigma^{-1}\vbeta(\omega^{(1,1)}_{n}-1)}{N_{n}}\right)^j \left\{  \frac{\vbeta' \mSigma^{-1}}{N_{n}}\vy_{t+1}  +\Delta^{d_0}_+ z_{1, t}(\vtheta) \right\} + w_{t+1},
		\end{align*}
		where $w_{t+1}$ is an $I(0)$ process that accounts for the impact of the fractional differences in $L_{b_0}$. Finally, using a geometric series and plugging in $N_n$ gives
		\begin{align}
			{x}_{t+1|t} &= \left( 1 -\frac{1 + \vbeta' \mSigma^{-1}\vbeta(\omega^{(1,1)}_{n}-1)}{N_{n}} \right)^{-1} \left\{  \frac{\vbeta' \mSigma^{-1}}{N_{n}}\vy_{t+1} +\Delta^{d_0}_+ z_{1, t}(\vtheta) \right\} + w_{t+1}= \notag \\
&=
			\frac{\vbeta'\mSigma^{-1}}{\vbeta' \mSigma^{-1}\vbeta}y_{t+1} + \frac{1 + \vbeta' \mSigma^{-1}\vbeta \omega^{(1,1)}_{n}}{\vbeta'\mSigma^{-1}\vbeta} \Delta^{d_0}_+ z_{1, t}(\vtheta) + w_{t+1} =  \frac{\vbeta'\mSigma^{-1}}{\vbeta' \mSigma^{-1}\vbeta}y_{t+1} - z_{t+1}(\vtheta), \label{eq:kfred}
		\end{align}
where $z_{t+1}(\vtheta) =  -  \frac{1 + \vbeta' \mSigma^{-1}\vbeta \omega^{(1,1)}_{n}}{\vbeta'\mSigma^{-1}\vbeta} \Delta^{d_0}_+ z_{1, t}(\vtheta)  - w_{t+1} $ and the minus sign is included to facilitate its interpretation.
	
	\noindent
	By multiplication of \eqref{eq:kfred} with $\vbeta$ one obtains the conditional expectation
	\begin{align}\label{kf:can}
		\E_\vtheta(\vy_{t+1} | \mathcal{F}_t)&= \beta x_{t+1|t} = \frac{\beta \beta' \mSigma^{-1}}{\vbeta' \mSigma^{-1}\vbeta} y_{t+1} - \vbeta z_{t+1}(\vtheta), 
	\end{align}
and the prediction error in \eqref{kf:can2}.

\noindent
To derive an expression for $z_{t+1}(\vtheta)$, we add and subtract $\frac{\vbeta \vbeta' \mSigma^{-1}  }{\vbeta' \mSigma^{-1}\vbeta }\vv_{t+1}(\vtheta)$ to $\vv_{t+1}(\vtheta) = y_{t+1} - \E_\vtheta (\vy_{t+1}| \mathcal{F}_t)$
\begin{align}\label{vvt:ident}
	\vv_{t+1}(\vtheta) 
&= \left( I -  \frac{\vbeta \vbeta' \mSigma^{-1}  }{\vbeta' \mSigma^{-1}\vbeta } \right)\vy_{t+1} + \frac{\vbeta \vbeta' \mSigma^{-1}  }{\vbeta' \mSigma^{-1}\vbeta }\left( \vy_{t+1} - \E_\vtheta(\vy_{t+1} | \mathcal{F}_t) \right),
\end{align}
since $\left( I -  \frac{\vbeta \vbeta' \mSigma^{-1}  }{\vbeta' \mSigma^{-1}\vbeta } \right)\E_\vtheta(\vy_{t+1} | \mathcal{F}_t)= \left( I -  \frac{\vbeta \vbeta' \mSigma^{-1}  }{\vbeta' \mSigma^{-1}\vbeta } \right)\beta\E_\vtheta(\vx_{t+1} | \mathcal{F}_t)=0$. By adding and subtracting $\sum_{i=1}^t \pi_i(b) \vy_{t+1-i}$ inside the last parentheses equation \eqref{vvt:ident} becomes 
\begin{align}\label{vvt:ident2}
	\vv_{t+1}(\vtheta) &=\left( I -  \frac{\vbeta \vbeta' \mSigma^{-1}  }{\vbeta' \mSigma^{-1}\vbeta } \right)\vy_{t+1} + \frac{\vbeta \vbeta' \mSigma^{-1}  }{\vbeta' \mSigma^{-1}\vbeta }\left( \Delta_+^b \vy_{t+1}  -  \E_\vtheta(\Delta_+^b \vy_{t+1} | \mathcal{F}_t)  \right).
\end{align}

\noindent
We can plug \eqref{vvt:ident2} into \eqref{kf:can2} and solve for $\beta z_{t+1}$ which yields
$			\vbeta z_{t+1}(\vtheta) = \frac{\vbeta \vbeta' \mSigma^{-1}  }{\vbeta' \mSigma^{-1}\vbeta }( \Delta_+^b \vy_{t+1}  -  \E_\vtheta(\Delta_+^b \vy_{t+1} | \mathcal{F}_t) ).$
		This completes the proof.
\end{proof}
\begin{proof}[Proof of Lemma \ref{lemma:z_t_order}]
To derive the integration order of  $z_t(\vtheta)$ given in theorem \ref{th:1}, which is the prediction error of the univariate process  $\Delta_+^b \frac{\vbeta' \mSigma^{-1}  }{\vbeta' \mSigma^{-1}\vbeta }y_{t} $,  we show that $z_t(\vtheta)$ is identical to the residuals in  \cite{Nie2015} for which he determined the integration order. First, consider $z_t(\vtheta)$, for which one has from model \eqref{dgp:1}  $\Delta_+^b \frac{\vbeta' \mSigma^{-1}}{\vbeta' \mSigma^{-1}\vbeta} \vy_{t} = \vzeta_{t}(\vtheta) = \eta_{t} + \Delta_+^b \frac{\vbeta' \mSigma^{-1}}{\vbeta' \mSigma^{-1}\vbeta} \vu_{t}$. Since $\eta_t\sim I(0)$ and $\Delta_+^b \vu_t \sim I(-b)$ their sum  $\vzeta_{t}(\vtheta)$ is $I(0)$ due to the aggregation properties of fractional processes. Furtheremore, since $\eta_t$, $\Delta_+^b \vu_t$ are independent, it follows from \citet[][p. 29]{GraNew1986} that 
$
\zeta_t(\vtheta) = A_+(L, \vtheta) g_t = \sum_{i=0}^{t-1}A_i(\vtheta) g_{t-i}
$
follows a moving average process of order $n-1$, where $g_t$ is Gaussian white noise and zero for all $t \leq 0$.
The coefficients $A_i$ are obtained by matching the partial autocovariance functions of $A_+(L, \vtheta) g_t$ and $\eta_{t} + \Delta_+^b \frac{\vbeta' \mSigma^{-1}}{\vbeta' \mSigma^{-1}\vbeta} \vu_{t}$. They are $A_0=1$, $A_i = \pi_i(b) (1+\vbeta'\mSigma^{-1}\vbeta)^{-1/2}$, and $\Var(g_t)=1+(\vbeta'\mSigma^{-1}\vbeta)^{-1}$.
Due to the $I(0)$ property, $A_+(L,\vtheta)$ remains invertible for $n\to \infty$. Additionally, $ \frac{\vbeta' \mSigma^{-1}}{\vbeta' \mSigma^{-1}\vbeta} \vy_{t} $ has an ARMA($n-1$, $n-1$) state space representation \citep[cf.][ch.\ 3.4]{DurKoo2012}.
	Rearranging with $B_+(L, \vtheta)$ as the truncated inverse of $A_+(L, \vtheta)$ gives $\frac{\vbeta' \mSigma^{-1}}{\vbeta' \mSigma^{-1}\vbeta} \vy_{t} = -(B_+(L, \vtheta) \Delta_+^b - 1) \frac{\vbeta' \mSigma^{-1}}{\vbeta' \mSigma^{-1}\vbeta} \vy_t + g_t$, from which it becomes clear that for a given $\vtheta$ the prediction error $z_t(\vtheta)$ and the residuals $g_t(\vtheta)$ as considered in \cite{Nie2015} are identical since using \eqref{vvt:ident}
	\begin{align}\label{vbar}
		z_t(\vtheta) &=\frac{\vbeta' \mSigma^{-1}}{\vbeta' \mSigma^{-1}\vbeta}\left[ \vy_t - \E_\vtheta(\vy_t | \mathcal{F}_{t-1}) \right]
		= \frac{\vbeta' \mSigma^{-1}}{\vbeta' \mSigma^{-1}\vbeta}  \left[ \vy_t +  (B_+(L, \vtheta) \Delta_+^b - 1) \vy_t   \right]= \notag\\
		&=  B_+(L, \vtheta) \Delta_+^b  \frac{\vbeta' \mSigma^{-1}}{\vbeta' \mSigma^{-1}\vbeta}  \vy_t=  g_t(\vtheta).
	\end{align} 
From $\Delta_+^b \vy_t \sim I(b_0-b)$ it follows that $z_t(\vtheta)\sim I(b_0-b)$. 
\end{proof}
\noindent
The following lemmas are required for theorem \ref{th:2}.
\begin{lemma}\label{L2:a}
	The covariance of $\vy_t$, $\veta_{t-j}$ is given by
	\begin{align}\label{eq:Cov_YEta}
		\Cov_\vtheta\left(y_t,\eta_{t-j}\right)  & =\begin{cases}  \beta \sum_{i=0}^j \varphi_i & \text{ if } b\geq 1, \\
			\beta \varphi_j, & \text{ if } b <1,
		\end{cases}
		\quad j=0,\ldots,t-1.  
	\end{align}
\end{lemma}
\begin{proof}[Proof of Lemma \ref{L2:a}]
		Let $b \geq 1$. From $x_t = (1-L)_+^{-\indic}\sum_{i=0}^{t-1}\varphi_i \veta_{t-i}$ it follows that 
	\begin{align*}
		\Cov_\vtheta\left(\vy_t,\eta_{t-j}\right) & =  \Cov_\vtheta ( \vbeta x_t + \vu_t, \eta_{t-j}) =\Cov_\vtheta\left( \vbeta \sum_{s=1}^t\sum_{i=0}^{s-1} \varphi_i \eta_{s-i} + \vu_t , \eta_{t-j} \right) \nonumber \\
		& = \vbeta  \sum_{s=1}^t\sum_{i=0}^{s-1} \varphi_i \Cov_\vtheta(\eta_{s-i}, \eta_{t-j}) =  \vbeta  \sum_{s=t-j}^t \varphi_{s-(t-j)} = \vbeta \sum_{i=0}^j \varphi_i.
	\end{align*}
	For $b<1$ one has
	$
		\Cov_\vtheta\left(\vy_t,\eta_{t-j}\right)  = \Cov_\vtheta(\vbeta \sum_{i=0}^{t-1}\varphi_i \eta_{t-i} + \vu_t, \eta_{t-j} ) = \vbeta\varphi_j,$ $j=0,...,t-1$.
\end{proof}
\begin{lemma}\label{L2:b}
	The autocovariance function of $\vy_t$ satisfies
	\begin{align} \label{eq:Cov_YY}
		\Cov_\vtheta\left( \vy_t, \vy_{t-k} \right)  & = \begin{cases}
			\vbeta \vbeta^\prime \sum_{s=1}^{t-k} \sum_{u=1}^s \sum_{s^\prime=u}^{t-k} \varphi_{s-u} \varphi_{s^\prime -u} \\ + \beta \beta^\prime \sum_{s=t-k+1}^{t} \sum_{u=1}^{t-k} \sum_{s^\prime=u}^{t-k} \varphi_{s-u} \varphi_{s^\prime -u}& \text{ if } b\geq 1, \\ 
			\vbeta \vbeta^\prime \sum_{l=0}^{t-1-k} \varphi_{k+l} \varphi_l,& \text{ if } b<1, \end{cases} k=1,\ldots,t-1.
	\end{align}
\end{lemma}
\begin{proof}[Proof of Lemma \ref{L2:b}]
Let $b\geq 1$.  From $x_t = (1-L)_+^{-\indic}\sum_{i=0}^{t-1}\varphi_i \veta_{t-i}$  one has for $k=1,\ldots,t-1$,
\begin{align*}
	\Cov_\vtheta(\vy_t, \vy_{t-k}) & = \Cov_\vtheta\left( \vbeta \sum_{s=1}^t\sum_{i=0}^{s-1} \varphi_i \eta_{s-i} + \vu_t, 
	\vbeta \sum_{s^\prime=1}^{t-k}\sum_{i^\prime=0}^{s^\prime-1} \varphi_{i^\prime} \eta_{s^\prime-i^\prime} + \vu_{t-k}\right) \nonumber \\
	& = \vbeta \left[ \sum_{s=1}^t \sum_{i=0}^{s-1} \sum_{s^\prime=1}^{t-k} \sum_{i^\prime=0}^{s^\prime-1} \varphi_i \varphi_{i^\prime} \Cov(\eta_{s-i}, \eta_{s^\prime-i^\prime})\right] \vbeta^\prime,  \nonumber
	\intertext{and with defining $u= s-i$ and $u^\prime = s^\prime-i^\prime$ 
one obtains
}
	\Cov_\vtheta(\vy_t, \vy_{t-k}) & = \vbeta\vbeta^\prime \sum_{s=1}^{t-k} \sum_{u=1}^s \sum_{s^\prime=1}^{t-k} \sum_{u^\prime=1}^{s^\prime} \varphi_{s-u} \varphi_{s^\prime -u^\prime} \Cov_\vtheta(\eta_u, \eta_{u^\prime})  \\
	& + \vbeta\vbeta^\prime \sum_{s=t-k+1}^{t} \sum_{u=1}^{t-k} \sum_{s^\prime=1}^{t-k} \sum_{u^\prime=1}^{s^\prime} \varphi_{s-u} \varphi_{s^\prime -u^\prime} \Cov_\vtheta(\eta_u, \eta_{u^\prime})\\
	& = \beta\beta^\prime \sum_{s=1}^{t-k} \sum_{u=1}^s \sum_{s^\prime=u}^{t-k} \varphi_{s-u} \varphi_{s^\prime -u} +  \beta \beta^\prime \sum_{s=t-k+1}^{t} \sum_{u=1}^{t-k} \sum_{s^\prime=u}^{t-k} \varphi_{s-u} \varphi_{s^\prime -u},\quad b\geq 1.
\end{align*} 

\noindent
For $b<1$ one has
$
	\Cov_\vtheta\left(\vy_t, \vy_{t-k} \right)   = \Cov_\vtheta\left( \vbeta \sum_{i=0}^{t-1} \varphi_i \eta_{t-i} + u_t, \beta \sum_{l=0}^{t-k-1} \varphi_l \eta_{t-k-l} + u_{t-k} \right) \nonumber = \vbeta \vbeta^\prime \sum_{l=0}^{t-k-1} \varphi_{k+l} \varphi_l$. 
Here $t-i=t-k-l$ was used to obtain $i=k+l$ and $l\leq t-k-1$. 
\end{proof}
\begin{corollary}\label{L2:c}
Given $\vtheta$, the joint normal distribution of $\veta_{1:t} = (\veta_{1}, ..., \veta_t)'$, $\mY_t = (\vy_1', ... , \vy_t')'$ is given by
\begin{align} \label{eq:joint_normal}
	\bvec \eta_{1:t} \\ \mY_t \evec \sim N\left( 0, \bmat \mI_t & \Sigma_{\eta_{1:t} Y_t} \\ 
	\Sigma_{\eta_{1:t}Y_t}' &  \Sigma_{Y_t} \emat \right),
\end{align}
where the $(tp\times t)$ covariance matrix
$\Sigma_{\eta_{1:t}Y_t}'$ has entries $Cov_\vtheta(y_s,\eta_{s-i})$, $i=0,\ldots,s-1$ for $s=1,\ldots,t$ given by \eqref{eq:Cov_YEta}  and zero matrices for all $s$ with $i>0$ and  $s+i\leq t-1$.  The $(tp\times tp)$ covariance matrix $\Sigma_{Y_t} $ has entries given by \eqref{eq:Cov_YY}.
\end{corollary}
\noindent
\begin{proof}[Proof of Theorem \ref{th:2}]
	Using lemmas \ref{L2:a}, \ref{L2:b}, and corollary \ref{L2:c}, the conditional expectation of the latent state from the truncated model \eqref{tru:1} can be rearranged such that
	\begin{align*}
		\tilde{x}_{t+1|t} &= x_{t+1|t} + \E_\vtheta (\tilde{x}_{t+1} - x_{t+1} | \mathcal{F}_t) = x_{t+1|t} - \Delta_+^{-\indic}\sum_{i=m+1}^{t} \varphi_i \E_\vtheta(\veta_{t+1-i} | \mathcal{F}_t) = \\
		&= x_{t+1|t} - \Delta^{-\indic}_+\sum_{i=m+1}^{t} \varphi_i \ve_{t+1-i} \mSigma_{\veta_{1:t}\mY_t} \mSigma_{\mY_{t}}^{-1} \mY_t,
	\end{align*}
where the last step follows from
$
	\E_\vtheta(\veta_{1:t}|\mathcal{F}_{t}) = 
	\mSigma_{\veta_{1:t}Y_{t}} \mSigma_{Y_t} ^{-1} \left(\mY_t - \E_\vtheta(Y_t)\right)= \mSigma_{\veta_{1:t}Y_t} \Sigma_{\mY_t} ^{-1} \mY_t.
$
\end{proof}

\subsection{Proofs for section 3} \label{AS3}
\begingroup
\allowdisplaybreaks
\begin{lemma}\label{L:1b}
	For a fixed state dimension $n$ the prediction error covariance matrix of the exact model \eqref{dgp:1} has a steady state
	\begin{align*}
		\mP_{t|t-1} = \mP^{[n]} + O(e^{-t}),
	\end{align*}
	and, therefore, $\omega_{t}^{(i, j)} \to \omega^{(i, j)}_{[n]}$ as $t \to \infty$, and $\lim_{t \to \infty } \mF_t = \mF^{[n]}$ where the superscript $[n]$ denotes the dependence of $\lim_{t \to \infty} \mF_t$ on the system dimension $n$ due to the type II definition of fractional integration. 
\end{lemma}
\begin{proof}[Proof of Lemma \ref{L:1b}]
	As shown by \citet[section 4.4]{AndMoo1979}, any stable, time invariant state space model with positive semi-definite initial prediction error covariance matrix $\mP_{1|0}$ has a steady state solution for $\mP_{t+1|t}$. Furthermore, a non-stable system has a steady state solution for $\mP_{t+1|t}$ if it is stabilisable and detectable and if $\mP_{1|0}$ is positive semi-definite. 
	Note that $\mP_{1|0}$ is given by 
	\begin{align*}
		\bmat 1 & \varphi_1   & \cdots & \varphi_n \\ 
		\varphi_1 & \varphi_1^2 & \cdots & \varphi_1\varphi_n \\ 
		\vdots & \vdots & \ddots & \vdots \\
		\varphi_n & \varphi_n\varphi_1 & \cdots & \varphi_n^2 \emat,
	\end{align*} 
	which follows from $\alpha_1=(x_1, \varphi_1 \eta_1,\ldots,\varphi_n \eta_1)'$. Therefore, the matrix $\mP_{1|0}$  is positive semidefinite. Hence, it is sufficient to show that our model is stable for $b < 1$ and stabilisable and detectable for $b \geq 1$. For this, consider the representation
	\begin{align*}
		\vy_t = \mZ^* \valpha_t^* = \bmat \mZ & \mI \emat \bvec \valpha_t \\ \vu_t \evec, && \valpha^*_t = \mT^* \valpha_{t-1}^* + \mG \bvec \eta_t \\ \vu_t \evec = \bmat \mT & 0 \\ 0 & 0 \emat \valpha_{t-1}^* + \bmat \mR & 0 \\ 0 & \mI \emat \bvec \eta_t \\ \vu_t \evec.
	\end{align*}
	The following definitions are taken from \citet[section 3.3]{Har1990}. A system is stable if the characteristic roots of the transition matrix $\mT^*$ have modulus less than one, i.e. $| \lambda_i (\mT^*)| < 1$ $\forall i$. 
	Furthermore, a system is called stabilisable if there exists a matrix $\mS$ such that $|\lambda_i(\mT^* + \mG \mS')| < 1$ $\forall i$. 
	Finally, a system is detectable if there exists a matrix $\mD$ such that $| \lambda_i(\mT^* - \mD \mZ^*)|<1$ $\forall i$. \\  
	Beginning with the stable case, $b < 1$, we note that $\mT^*$ is a strictly upper triangular matrix, such that its eigenvalues $\lambda_i (\mT^*)=0$ $\forall i$. Another way to see this is to rewrite $x_t$ as $x_t = - \sum_{i=1}^{t-1} \pi_i(b) x_{t-i} + \eta_t$, where all roots of $- \sum_{i=1}^{t-1}\pi_i(b) L^i$ lie outside the unit circle for $b < 1$. Hence, for $b < 1$ the system is stable. 
	
	\noindent
	For $b \in [1, 1.5)$ the system is not stable since its largest eigenvalue equals $1$ due to the unit root imposed on $\vx_t$ via $\mT$. 
	Nonetheless, the nonstationary unobserved components model is detectable since a $(n+1+p) \times p$ matrix $\mD$ with $\mD^{(1,1)}=1/\vbeta^{(1)}$ in its upper left entry and all other elements $0$ yields a strictly upper triangular matrix $\mT^* - \mD \mZ^*$ such that all eigenvalues are zero. Furthermore, the model is stabilisable since an $(n+1+p)\times (1+p)$ matrix $\mS$ with $\mS^{(1,1)}=-1$ and all other entries zero yields eigenvalues that are bounded below one in absolute value due to the stationary coefficients $\varphi_i$. Therefore, the nonstationary model is also stabilisable. Consequently, lemma \ref{L:1b} follows.
\end{proof}

\begin{lemma} \label{L:1c}
	
	As $n \to \infty$ the steady state prediction error variance $\mF^{[n]}$ defined in lemma \ref{L:1b} converges 
	\begin{align*}
		\lim_{n \to \infty} \mF^{[n]} = \lim_{n \to \infty}  \lim_{t \to \infty} \mF_{t}^{[n]} = \mF,
	\end{align*}
	where $\mF_t^{[n]} = \Var_\vtheta ( \vv_t(\vtheta) | \mathcal{F}_{t-1}) $ indicates the dependence of $\mF_t$ on the state dimension $n$, and $0 < \mF < \infty$.
\end{lemma}

\begin{proof}[Proof of Lemma \ref{L:1c}]
	To prove lemma \ref{L:1c} we first consider $\mF_t^{[n]}$ and derive the limits for $\mF_t^{[n+1]} - \mF_t^{[n]}$. Note that $\mF_t^{[n+1]}$, $\mF_t^{[n]}$ are identical for $t \leq n$, such that $\lim_{n \to \infty}\mF_t^{[n+1]} - \mF_t^{[n]}=0$ holds for a fixed $t$. Thus, we only consider $t > n$.  Next, we show that the limit of $\mF_t^{[n]}$ is bounded. 
	
	\noindent
	To simplify the notation, we define $P=I - \frac{\mSigma^{-1}\vbeta \vbeta'}{\vbeta' \mSigma^{-1}\vbeta}$ analog to \eqref{eq:Px}. Then from theorem \ref{th:1} $\vv_t(\vtheta)= \mP' \vy_t + (I-\mP)'(\Delta_+^b \vy_t - \E_\vtheta(\Delta_+^b \vy_t | \mathcal{F}_{t-1}))$, such that $\mF_t^{[n]}=\Var_\vtheta(\vv_t(\vtheta) | \mathcal{F}_{t-1}) = \Var_\vtheta(\mP' \vu_t(\vtheta) + (I - \mP)' \Delta_+^b \vy_t | \mathcal{F}_{t-1} ) =- \mP' \mSigma \mP + \mP' \mSigma + \mSigma P+ \Var_\vtheta( (I - \mP)' \Delta_+^b \vy_t | \mathcal{F}_{t-1})$ since $P'\vbeta x_t = 0$ and $\Cov_\vtheta(\mP'\vu_t(\vtheta), (I-P)'\Delta_+^b \vy_t | \mathcal{F}_{t-1}) = \mP'\mSigma (I - \mP)$. Furthermore $- \mP' \mSigma \mP + \mP' \mSigma + \mSigma P = \mSigma - \frac{\vbeta \vbeta'}{\vbeta' \mSigma^{-1}\vbeta}$ which can be seen by plugging in $P$. Again using $P'\vbeta x_t = 0$, the latter term is $\Var_\vtheta( (I - \mP)' \Delta_+^b \vy_t | \mathcal{F}_{t-1}) = \Var_\vtheta(  \Delta_+^b \vy_t - \mP' \Delta_+^b\vu_t(\vtheta)| \mathcal{F}_{t-1}) =  \Var_\vtheta(   \vbeta \eta_t(\vtheta) + (I -  \mP)' \Delta_+^b\vu_t(\vtheta)| \mathcal{F}_{t-1}) = \vbeta \vbeta' +  \Var_\vtheta(  (I - \mP)' \Delta_+^b\vu_t(\vtheta)| \mathcal{F}_{t-1})$. Thus
	\begin{align}\label{eq:Fn}
	 \mF_t^{[n]} = \mSigma - \frac{\vbeta \vbeta'}{\vbeta' \mSigma \vbeta} + \vbeta \vbeta' + \Var_\vtheta(  (I - \mP)' \Delta_+^b\vu_t(\vtheta)| \mathcal{F}_{t-1}).
	\end{align}
	For the latter term we define $A_{n}=		\Var_\vtheta(  (I - \mP)' \sum_{i=0}^{n-1}\pi_i(b)\vu_{t-i}(\vtheta))$, which is independent of $t$ due to $t>n$, and $B_{n, t}= \Var_\vtheta(\E_\vtheta(  (I - \mP)' \sum_{i=0}^{n-1}\pi_i(b)\vu_{t-i}(\vtheta)| \mathcal{F}_{t-1} ))$. It follows from the law of total variance that
	\begin{align}\label{eq:LtV}
		\Var_\vtheta(  (I - \mP)' \sum_{i=0}^{n-1}\pi_i(b)\vu_{t-i}(\vtheta)| \mathcal{F}_{t-1})= A_{n} - B_{n, t}.
	\end{align}
	
	\noindent
	Since all other terms are constant, the difference between $\mF_t^{[n+1]}$ and $\mF_t^{[n]}$ solely depends on \eqref{eq:LtV} and is given by 
	\begin{align}\label{eq:F}
		\mF_t^{[n+1]} - \mF_t^{[n]} = A_{n+1} - A_{n} - (B_{n+1, t} - B_{n, t}).
	\end{align} In the following, we consider $A_{n+1} - A_{n} $ and $B_{n+1, t} - B_{n, t}$ separately, where we show that their limits converge to zero. 
	
	\noindent
	Since $A_{n+1} =\Var_\vtheta(  (I - \mP)' \sum_{i=0}^{n} \pi_i(b)\vu_{t-i}(\vtheta)) = (I - \mP)'\mSigma (I - \mP)\sum_{i=0}^n \pi_i^2 (b)$, and analog for $A_{n}$, one directly has
	\begin{align}
		A_{n+1} - A_{n} = (I - \mP)'\mSigma (I - \mP)\pi_n^2(b).
	\end{align}
	Note that $A_n$ is invariant w.r.t. $t$, and $\lim_{n \to \infty} \lim_{t \to \infty} (A_{n+1} - A_{n})  = \lim_{n \to \infty}  (A_{n+1} - A_{n}) =0$ since $\pi_n^2(b) = O(n^{-2-2b})$  \citep[cf.\ e.g.][lemma 5.1]{Has2018}.
	
	\noindent
	The calculation of $B_{n+1, t} - B_{n, t}$ is more involved. By writing $B_{n+1, t} = B_{n, t} + C_{n+1, t} + D_{n+1, t} + D_{n+1, t}'$ with $C_{n+1, t}=\Var_\vtheta(\E_\vtheta(  (I - \mP)' \pi_n(b) \vu_{t-n}(\vtheta)| \mathcal{F}_{t-1} ))$, $D_{n+1, t} = \Cov_\vtheta ( \E_\vtheta(  (I - \mP)' \sum_{i=0}^{n-1}\pi_i(b)\vu_{t-i}(\vtheta)| \mathcal{F}_{t-1} ), \E_\vtheta(  (I - \mP)' \pi_n(b)\vu_{t-n}(\vtheta)| \mathcal{F}_{t-1} )) $ the difference becomes $B_{n+1, t} - B_{n, t} = C_{n+1, t} + D_{n+1, t} + D_{n+1, t}'$. 
	
	\noindent
	For $D_{n+1, t}$, define $\mY_{t-1} = (y_{1}',...,y_{t-1}')'$ and $\mSigma_{Y_{t-1}} = \Var_\vtheta (\mY_{t-1})$. Then it follows from \citet[lemma 1]{DurKoo2012}
\begin{align}
	&\Cov_\vtheta \left( \E_\vtheta \left(  \sum_{i=0}^{n-1}\pi_i(b)\vu_{t-i}(\vtheta)\bigg| \mathcal{F}_{t-1} \right), \E_\vtheta \left(   \pi_n(b)\vu_{t-n}(\vtheta )| \mathcal{F}_{t-1} \right)\right) \notag \\
	&= \sum_{i=0}^{n-1}\pi_i(b) \Cov_\vtheta \left( \Cov_\vtheta (\vu_{t-i}(\vtheta), \mY_{t-1})\mSigma_{\mY_{t-1}}^{-1} \mY_{t-1},  \Cov_\vtheta (\vu_{t-n}(\vtheta), \mY_{t-1})\mSigma_{Y_{t-1}}^{-1} \mY_{t-1}\right) \pi_{n}(b) = \notag \\
	&=  \pi_{n}(b) \sum_{i=0}^{n-1}\pi_i(b)\Cov_\vtheta (\vu_{t-i}(\vtheta), \mY_{t-1}) \mSigma_{\mY_{t-1}}^{-1}\Cov_\vtheta (\vu_{t-n}(\vtheta), \mY_{t-1})' \notag \\
	&= \pi_{n}(b) \sum_{i=1}^{n-1}\pi_{i}(b)\mSigma E_{t-i} \mSigma_{\mY_{t-1}}^{-1}E_{t-n}' \mSigma, \label{eq:D}
\end{align}
where $E_{j} = \bmat 0_{p\times p} & \cdots & 0_{p\times p} & I_{p\times p} & 0_{p\times p} & \cdots & 0_{p\times p} \emat$ is a $p \times (t-1)p$ selection matrix, with an identity matrix in its $j$-th block and all other blocks zero. Hence, $ \mSigma_{\mY_{t-1}}^{-1}E_{t-n}'$ picks the columns corresponding to $\Cov_\vtheta(\mY_{t-1}, \vy_{t-n})$ from the inverse $\mSigma_{\mY_{t-1}}^{-1}$, and hence $\mSigma E_{t-i} \mSigma_{\mY_{t-1}}^{-1}E_{t-n}' \mSigma$ is finite for all $t > n$.  Since the sum $\sum_{i=1}^{n-1}\pi_{i}(b) < \infty$ for all $n$ \citep[][lemma 5.2]{Has2018}, it follows for \eqref{eq:D} that
$
\sum_{i=1}^{n-1}\pi_{i}(b)\mSigma E_{t-i} \mSigma_{\mY_{t-1}}^{-1}E_{t-n}' \mSigma
$
is finite. As noted before $\pi_n(b) = O(n^{-1-b})$, such that the limit $\lim_{n \to \infty} \lim_{t \to \infty} D_{n+1,t} = 0$.

\noindent
For $C_{n+1, t}=\pi_n^2(b)\Var_\vtheta(\E_\vtheta(  (I - \mP)'  \vu_{t-n}(\vtheta)| \mathcal{F}_{t-1} ))$ one obtains from the law of total variance that $\Var_\vtheta(\E_\vtheta(\vu_{t-n}(\vtheta)| \mathcal{F}_{t-1} )) \leq \Var_\vtheta(\vu_{t-n}(\vtheta)) = \Sigma$. Since $\pi_n^2(b)=O(n^{-2-2b})$, $\lim_{n \to \infty} \lim_{t \to \infty} C_{n+1,t}=0$. The results for $C_{n+1,t}$, $D_{n+1, t}$ imply $\lim_{n \to \infty} \lim_{t \to \infty} (B_{n+1,t} - B_{n, t})=0$. 
It then follows for \eqref{eq:F} that
\begin{align}\label{eq:Flim}
	\lim_{n \to \infty} \lim_{t \to \infty}( \mF_t^{[n+1]} - \mF_t^{[n]})  = \lim_{n \to \infty} \lim_{t \to \infty}(A_{n+1} - A_{n})  - \lim_{n \to \infty} \lim_{t \to \infty}(B_{n+1,t} - B_{n, t})  =0.
\end{align}

\noindent
Finally, to prove boundedness of $\lim_{n \to \infty} \mF^{[n]}$, it is sufficient to show that in \eqref{eq:Fn} the limit $\lim_{n \to \infty} \lim_{t \to \infty}\Var_\vtheta(  (I - \mP)' \Delta_+^b\vu_t(\vtheta)| \mathcal{F}_{t-1})<\infty$. From the law of total variance in \eqref{eq:LtV} it follows that $\Var_\vtheta(  (I - \mP)' \sum_{i=0}^{n-1}\pi_i(b)\vu_{t-i}(\vtheta)| \mathcal{F}_{t-1})\leq A_{n} $ since $B_{n, t}\geq 0$. For $A_{n}=(I - \mP)'\mSigma (I - \mP)\sum_{i=0}^{n-1} \pi_i^2 (b)$, note that $\lim_{t \to \infty} A_n = A_n$, and $
\lim_{n \to \infty} \sum_{i=0}^{n-1} \pi_{i}^2(b)< \infty,
$ 
 \citep[cf.\ e.g.][lemma 5.2]{Has2018}. Hence, $\lim_{n \to \infty} \mF^{[n]} < \infty$ and $\lim_{n \to \infty}\lim_{t \to \infty} \mF_t^{[n]}= \mF$. 
\end{proof}

 
\begin{proof}[Proof of Theorem \ref{th:consistency}]
The  prediction error $z_t(\vtheta)$ of $\Delta_+^b \frac{\vbeta' \mSigma^{-1}  }{\vbeta' \mSigma^{-1}\vbeta } y_{t}$ is the only component in $v_t(\vtheta)$ that depends on $b$. Therefore, it is the only part in $v_t(\vtheta)$ that matters for estimating $b$. In the proof of lemma \ref{lemma:z_t_order} we showed that the prediction error $z_t(\vtheta)$ is identical to the residuals in \cite{Nie2015} (compare \eqref{vbar}). 
While \cite{Nie2015} considers the CSS estimator, we consider the ML estimator based on \eqref{ll}. The latter also contains $F^{[n]}$ which  depends on the sample size $n$. By lemma \ref{L:1c} the steady state prediction error variance $F^{[n]}$ converges to $\mF$ as $n \to \infty$. Therefore, the ML estimator and the CSS estimator are asymptotically equivalent and it suffices to consider the behavior of the sum of squared residuals $\sum_{t=1}^n \vv_t(\vtheta) \vv_t(\vtheta)'$ in \eqref{ll}. By the equivalence of the prediction errors stated above this objective function is nested in the ARFIMA objective function considered in \cite{Nie2015}. Thus, his consistency results carry over to the ML estimator of $b$ if for $z_t(\vtheta)$ assumptions A -- D in \cite{Nie2015} hold. 
	
	\noindent
Since $g_t$ defined in the proof of lemma \ref{lemma:z_t_order} is univariate Gaussian white noise with positive variance and $b\in D$, assumptions A and B in \cite{Nie2015} are satisfied. Following the proof of lemma \ref{lemma:z_t_order}, $\zeta_{t}(\vtheta)=A_+(L, \vtheta) g_t $ is $I(0)$ which guarantees a well-defined inverse of the MA polynomial $A_+(L, \vtheta)$ even for $n\to\infty$. Therefore, assumptions C and D in \cite{Nie2015} hold. Under these assumptions it follows that the CSS estimator for $b$ is consistent. Since the CSS estimator has the same limit distribution as the ML estimator as argued before, it follows that $\hat{b} \pto b_0$ as $n \to \infty$. 
\end{proof}


\begin{proof}[Proof of Lemma \ref{L3a}]
	The partial derivatives of $\vv_t(\vtheta)$ w.r.t.\ $\vbeta$, $\mSigma$ have been derived for the $I(1)$ case in \citet[][lemma 3.2]{ChaMilPa2009}. We obtain similar expressions for the $I(b)$ case. 
	Note that from theorem \ref{th:1} and \eqref{vbar} 
		\begin{align*}
			v_{t}(\vtheta)' & =  \vy_t'  \left(\mI - \frac{\mSigma^{-1}\vbeta \vbeta'}{\vbeta' \mSigma^{-1}\vbeta}\right) + z_t(\vtheta) \vbeta'  =\vy_t'  \left(\mI - \frac{\mSigma^{-1}\vbeta \vbeta'}{\vbeta' \mSigma^{-1}\vbeta}\right)  +  B_+(L, \vtheta) \Delta_+^b   \vy_t'  \frac{ \mSigma^{-1}\vbeta \vbeta'}{\vbeta' \mSigma^{-1}\vbeta},
		\end{align*}
		with $B_+(L, \vtheta) = \left[1-(1+\vbeta'\mSigma^{-1}\beta)^{-1/2} + (1+\vbeta'\mSigma^{-1}\beta)^{-1/2} \Delta_+^{b}\right]^{-1}_+$ 
	following	from the proof of lemma \ref{lemma:z_t_order}. 
	The derivative w.r.t.\ $\vec \mSigma$, evaluated at $\vtheta_0$, is $I(0)$ and given by
	\begin{align*}
		&\frac{\partial \vv_t(\vtheta)'}{\partial \vec \mSigma} =\frac{\mSigma^{-1}\vbeta \otimes \mSigma^{-1}}{\vbeta'\mSigma^{-1}\vbeta}\left(I - \frac{\vbeta \vbeta' \mSigma^{-1}}{\vbeta'\mSigma^{-1}\vbeta}\right) \left(\vy_t \vbeta' - B_+(L, \vtheta) \Delta_+^b   \vy_t  \vbeta' \right) + \frac{\partial B_+(L, \vtheta)}{\partial \vec \mSigma}\Delta_+^b \vy_t' \frac{\mSigma^{-1}\vbeta \vbeta'}{\vbeta'\mSigma^{-1}\vbeta}, \\
		&\frac{\partial \vv_t(\vtheta)'}{\partial \vec \mSigma} \Bigg\rvert_{\vtheta = \vtheta_0} = \frac{\mSigma_0^{-1}\vbeta_0 \otimes \mSigma_0^{-1}}{\vbeta_0'\mSigma_0^{-1}\vbeta_0}\left(I - \frac{\vbeta_0 \vbeta_0' \mSigma_0^{-1}}{\vbeta_0'\mSigma_0^{-1}\vbeta_0}\right) \vbeta_0 x_t \vbeta_0'+ w_{t}  = a_\Sigma^0(\vu_t, \eta_t),
	\end{align*} 
	where $\partial B_+(L, \vtheta)/(\partial \vec \mSigma) = -1/2 (\mSigma^{-1} \vbeta \otimes \mSigma^{-1}\vbeta)(1+\vbeta'\mSigma^{-1}\vbeta )^{-3/2}B_+^2(L, \vtheta)(\Delta_+^b - 1)$ is a stationary filter,
	$w_t$, ${a}^0_\Sigma(\vu_t, \eta_t)$ are $I(0)$ processes that depend on $\vu_1,..., \vu_t$, $\eta_1,...,\eta_t$.
	Next, consider the derivative w.r.t.\ $\vbeta$, evaluated at $\vtheta_0$. For $x_{t|t-1}$ one has
	\begin{align*}
		\frac{\partial x_{t|t-1}}{\partial \vbeta} = \frac{\mSigma^{-1}}{\vbeta' \mSigma^{-1}\vbeta } \left( I - 2 \frac{\vbeta \vbeta' \mSigma^{-1} }{\vbeta'\mSigma^{-1}\vbeta } \right)\left(\vy_t   - B_+(L, \vtheta) \Delta_+^b   \vy_t  \right) - \frac{\partial B_+(L, \vtheta)}{\partial \vbeta} \frac{\vbeta'\mSigma^{-1}}{\vbeta'\mSigma^{-1}\vbeta}\Delta_+^b \vy_t, 
	\end{align*}
	where $\partial B_+(L, \vtheta)/ \partial \vbeta = \mSigma^{-1} \vbeta B_+^2(L, \vtheta)(\Delta_+^b - 1) (1 + \vbeta' \mSigma^{-1}\vbeta)^{-3/2}$ is a stationary filter. Thus
	\begin{align} \label{eq:partial_v_partial_beta}
		\frac{\partial \vv_t(\vtheta)'}{\partial \vbeta} &= -\frac{\vbeta' \mSigma^{-1}}{\vbeta'\mSigma^{-1}\vbeta}\left[  \vy_t  - B_+(L, \vtheta) \Delta_+^b \vy_t \right]\mI + \frac{\partial B_+(L, \vtheta)}{\partial \vbeta} \frac{\vbeta'\mSigma^{-1}}{\vbeta'\mSigma^{-1}\vbeta}\Delta_+^b \vy_t \vbeta' \nonumber \\
		&- \frac{\mSigma^{-1}}{\vbeta' \mSigma^{-1}\vbeta} \left(\mI - 2 \frac{\vbeta \vbeta' \mSigma^{-1}}{\vbeta'\mSigma^{-1}\vbeta}\right) \left[  \vy_t  - B_+(L, \vtheta) \Delta_+^b \vy_t  \right] \vbeta',\\
		\frac{\partial \vv_t(\vtheta)'}{\partial \vbeta}\Bigg\rvert_{\vtheta = \vtheta_0} &= - \left( \mI - \frac{\mSigma_0^{-1} \vbeta_0 \vbeta_0'}{\vbeta_0' \mSigma_0^{-1}\vbeta_0} \right) x_t + a_\vbeta^0(\vu_t, \eta_t), \nonumber
	\end{align}
	where again ${a}^0_\vbeta(\vu_t, \eta_t)\sim I(0)$ depends on $\vu_1,..., \vu_t$, $\eta_1,...,\eta_t$.
	
	\noindent
	For the derivative w.r.t.\ $b$, one obtains $\partial \vv_t(\vtheta)'/ \partial b = (\partial z_t(\vtheta)/\partial b) \vbeta' $. From \eqref{vbar} one has $z_t(\vtheta) = B_+(L,\vtheta) \frac{\vbeta'\mSigma^{-1}}{\vbeta'\mSigma^{-1}\vbeta} \Delta_+^{b-b_0} \Delta_+^{b_0}\vy_t,$
	\begin{align} 
		\frac{\partial z_t(\vtheta)}{\partial b} &= \frac{\vbeta' \mSigma^{-1}}{\vbeta' \mSigma^{-1}\vbeta} \left[ \frac{\partial B_+(L, \vtheta)}{\partial b}  \Delta_+^{b-b_0} \Delta_+^{b_0}\vy_t+ B_+(L, \vtheta) \frac{\partial }{\partial b} \Delta_+^{b-b_0} \Delta_+^{b_0}\vy_t\right].\label{partial:z}
	\end{align}
	To calculate the partial derivatives in \eqref{partial:z} we rearrange $ \Delta_+^{b-b_0} \Delta_+^{b_0}\vy_t = (1-L) \Delta_+^{b-b_0-1} \Delta_+^{b_0}\vy_t= (1-L)\sum_{j=0}^{t-1}\pi_j(b - b_0 - 1) \Delta_+^{b_0}\vy_{t-j}$, where $\pi_j(b-b_0-1) = \frac{\Gamma(1+b_0-b + j)}{\Gamma(j+1)\Gamma(1+b_0-b)}$, and $\Gamma(u)$ is the gamma function at $u$. Define $\Psi(u)$ as the digamma function at $u$, $\Psi(u) = \frac{\partial \Gamma(u)/\partial u}{\Gamma(u)}$. It satisfies $\Psi(u+j) - \Psi(u) = \sum_{k=0}^{j-1} (u+k)^{-1}$ for positive $u$. Due to theorem \ref{th:consistency} $|b-b_0|$ boils down to the stationary region, such that $1+b_0-b$ is positive asymptotically. 
		Then 
		\begin{align}\label{partial:b}
			\frac{\partial \pi_j(b - b_0 - 1)}{\partial b} &= - [\Psi(1+b_0-b+j) - \Psi(1+b_0-b)] \frac{\Gamma(1+b_0-b + j)}{\Gamma(j+1)\Gamma(1+b_0-b)} = \notag \\
			&=- \sum_{k=0}^{j-1}(1+b_0-b+k)^{-1}\pi_j(b-b_0-1),
		\end{align}
		and
		\begin{align}\label{partial:zeta}
			\frac{\partial}{\partial b} \Delta_+^{b-b_0} \Delta_+^{b_0}\vy_t= - (1-L)\sum_{j=1}^{t-1}\sum_{k=1}^j \frac{1}{b_0-b+k}\pi_j(b-b_0-1)\Delta_+^{b_0}\vy_{t-j}.
		\end{align}
		
		\noindent
		The first term in \eqref{partial:z} is 
		\begin{align}\label{partial:B}
			\frac{\partial B_+(L, \vtheta)}{\partial b} \Delta_+^{b-b_0} \Delta_+^{b_0} \vy_t = - B_+^2(L, \vtheta)(1+\vbeta'\mSigma^{-1}\vbeta)^{-1/2}  \Delta_+^{b}\frac{\partial}{\partial b} \Delta_+^{b-b_0} \Delta_+^{b_0} \vy_t. 
		\end{align}
		By plugging \eqref{partial:B} and \eqref{partial:zeta} into \eqref{partial:z} one obtains $\partial z_t(\vtheta)/\partial b$
		\begin{align}\label{partial:z2}
			\frac{\partial z_t(\vtheta)}{\partial b} &= \frac{\vbeta' \mSigma^{-1}}{\vbeta'\mSigma^{-1}\vbeta}B_+(L, \vtheta)\left(\frac{B_+(L, \vtheta)\Delta_+^{b}}{\sqrt{1+\vbeta'\mSigma^{-1}\vbeta}} -1\right)\Delta\sum_{j=1}^{t-1}\sum_{k=1}^j \frac{\pi_j(b-b_0-1)}{b_0-b+k}\Delta_+^{b_0}\vy_{t-j}.
		\end{align}
		For ${\vtheta = \vtheta_0}$ one has $\pi_j(-1) = 1$. The sum in \eqref{partial:z2} becomes $ (1-L)\sum_{j=1}^{t-1}\sum_{k=1}^j k^{-1}\Delta_+^{b_0}\vy_{t-j} = \sum_{j=1}^{t-1}\sum_{k=1}^j {k}^{-1}\Delta_+^{b_0}\vy_{t-j} -  \sum_{j=1}^{t-2}\sum_{k=1}^j {k}^{-1}\Delta_+^{b_0}\vy_{t-1-j} =  \sum_{j=1}^{t-1} {j}^{-1}\Delta_+^{b_0}\vy_{t-j} $,
		which is stationary and $\mathcal{F}_{t-1}$-measurable. For \eqref{partial:z2} evaluated at $\vtheta_0$ one has
		\begin{align*}
			\frac{\partial z_t(\vtheta)}{\partial b}\Big\rvert_{\vtheta = \vtheta_0} &=  \frac{\vbeta_0'\mSigma_0^{-1}}{\vbeta_0'\mSigma_0^{-1}\vbeta_0}B_+(L, \vtheta_0)\left(\frac{B_+(L, \vtheta_0)\Delta_+^{b_0}}{\sqrt{1+\vbeta_0'\mSigma_0^{-1}\vbeta_0}} -1\right)\sum_{j=1}^{t-1} {j}^{-1}\Delta_+^{b_0}\vy_{t-j},
		\end{align*}
		which is stationary since $B_+(L, \vtheta_0)$ is a stationary polynomial. Thus, $(\partial \vv_t(\vtheta)'/\partial b)\rvert_{\vtheta = \vtheta_0} = (\partial z_t(\vtheta)/ \partial b) \rvert _{\vtheta = \vtheta_0} \vbeta_0' = a_b^0(\vu_t, \veta_t)$. 
\end{proof}%
\endgroup

\noindent
The following lemmas are required for theorem \ref{th:4}

\begin{lemma}\label{L3:b}
	The process 
	\begin{align*}
		\frac{\partial \vv_t(\vtheta)'}{\partial\vtheta}\Bigg\rvert_{\vtheta = \vtheta_0} \mF_0^{{[n]}^{-1}} \vv_t(\vtheta_0),
	\end{align*}
	together with $\mathcal{F}_t$ is a martingale difference sequence. 
\end{lemma}
\begin{proof}[Proof of Lemma \ref{L3:b}]
	Note that $\frac{\partial \vv_t(\vtheta)'}{\partial\vtheta}\Big\rvert_{\vtheta = \vtheta_0} = -\frac{\partial x_{t|t-1}\vbeta'}{\partial\vtheta}\Big\rvert_{\vtheta = \vtheta_0}$ is $\mathcal{F}_{t-1}$-measurable since $x_{t|t-1}$ is $\mathcal{F}_{t-1}$-measurable. Hence,
	\begin{align*}
		\E \left[ 	\frac{\partial \vv_t(\vtheta)'}{\partial\vtheta}\Bigg\rvert_{\vtheta = \vtheta_0} \mF_0^{{[n]}^{-1}} 
		\vv_t(\vtheta_0) \Big| \mathcal{F}_{t-1} \right] 
		&= \frac{\partial \vv_t(\vtheta)'}{\partial\vtheta}\Bigg\rvert_{\vtheta = \vtheta_0} \mF_0^{{[n]}^{-1}} \E\left[ \vv_t(\vtheta_0) | \mathcal{F}_{t-1} \right] = 0, 
		\end{align*} and $
		\E \left[ 	\frac{\partial \vv_t(\vtheta)'}{\partial\vtheta}\big\rvert_{\vtheta = \vtheta_0} \mF_0^{{[n]}^{-1}} \vv_t(\vtheta_0) \right]= 0$
	by the law of iterated expectations. Since $\vy_t$ and $x_t$ are normally distributed, $\E[|\vy_t|] < \infty$ for every finite $t$, so that $\E[|\vv_t(\vtheta_0)|] <\infty$ and $\E[|x_t \vv_t(\vtheta_0)|] <\infty$ hold as well. Therefore $\E \left[ \left|	\frac{\partial \vv_t(\vtheta)'}{\partial\vtheta}\big\rvert_{\vtheta = \vtheta_0} \mF_0^{{[n]}^{-1}} \vv_t(\vtheta_0) \right| \right] < \infty$.Under these two conditions the process is a martingale difference sequence \citep[thm. 6.2.1]{Dav2000}. 
\end{proof}

\begin{lemma}\label{L3:c}
	If $b_0 < 0.5$, a CLT for the gradient in \eqref{eq:sn0} yields
	\begin{align}
		\frac{1}{\sqrt{n}} \sum_{t=1}^{n} \frac{\partial \vv_t(\vtheta)'}{\partial\vtheta}\Bigg\rvert_{\vtheta = \vtheta_0} \mF_0^{{[n]}^{-1}} \vv_t(\vtheta_0) &\dto \mG, \label{eq:L3:c:1}\\
		\frac{1}{2} \frac{\partial (\vec \mF^{[n]})'}{\partial \vtheta}\Bigg\rvert_{\vtheta = \vtheta_0} \left( \mF_0^{{[n]}^{-1}} \otimes \mF_0^{{[n]}^{-1}}  \right) \vec \left( \frac{1}{\sqrt{n}}  \sum_{t=1}^{n}\left( \vv_t(\vtheta_0) \vv_t(\vtheta_0)' - \mF_0^{[n]}  \right)  \right) &\dto \mJ \label{eq:L3:c:2},
	\end{align}
 $\mG \sim \mathrm{N}(0, \Var(\mG))$, $\mJ \sim \mathrm{N}(0, \Var(\mJ))$,
	as $n\to \infty$ where
	\begin{align*}
		\Var(\mG) &= \plim_{n \to \infty} \frac{1}{n} \sum_{t=1}^{n} \frac{\partial \vv_t(\vtheta)'}{\partial \vtheta} \Bigg \rvert_{\vtheta = \vtheta_0}  \mF_0^{-1}  \frac{\partial \vv_t(\vtheta)}{\partial \vtheta'}\Bigg\rvert_{\vtheta = \vtheta_0}, \\
		 \Var(\mJ) &= \frac{1}{2}\left[\frac{\partial (\vec \mF)'}{\partial \theta}\Bigg \rvert_{\vtheta=\vtheta_0} \left( \mF_0^{-1}  \otimes  \mF_0^{-1} \right)  \frac{\partial \vec \mF}{\partial \theta'}\Bigg \rvert_{\vtheta=\vtheta_0}\right].
	\end{align*}
\end{lemma}
\begin{proof}[Proof of Lemma \ref{L3:c}]
Due to lemma \ref{L3:b}, the l.h.s.\ of \eqref{eq:L3:c:1} together with $\mathcal{F}_t$ is a MDS. Since we show below that $\Var \left[  \frac{\partial \vv_t(\vtheta)'}{\partial \vtheta}\big \rvert_{\vtheta = \vtheta_0}  \mF_0^{{[n]}^{-1}}  \vv_t(\vtheta_0) \right] < \infty$ holds, a MDS CLT \citep[cf.][thm. 6.2.3]{Dav2000} applies and yields equation \eqref{eq:L3:c:1}. From lemma \ref{L:1c} one has $\mF_0^{[n]} \to \mF_0$ for $n\to\infty$ and $\mF_{t,0} = \Var\left(\vv_t(\vtheta_0) \rvert \mathcal{F}_{t-1}\right) = \mF_0 + o(1)$ so that
\begin{align} \nonumber
&\Var \left[  \frac{\partial \vv_t(\vtheta)'}{\partial \vtheta}\Bigg \rvert_{\vtheta = \vtheta_0} \mF_0^{{[n]}^{-1}} \vv_t(\vtheta_0) \right] 
= \E\left[ \Var \left(  \frac{\partial \vv_t(\vtheta)'}{\partial \vtheta}\Bigg \rvert_{\vtheta = \vtheta_0} \mF_0^{{[n]}^{-1}}  \vv_t(\vtheta_0) \rvert \mathcal{F}_{t-1} \right)  \right] \\
&= \E\left[ \frac{\partial \vv_t(\vtheta)'}{\partial \vtheta}\Bigg \rvert_{\vtheta = \vtheta_0} \mF_0^{-1}  \frac{\partial \vv_t(\vtheta)}{\partial \vtheta'}\Bigg \rvert_{\vtheta = \vtheta_0} \right] + o(1). \label{eq:Var_for_G} 
\end{align}
For the decisive block in \eqref{eq:Var_for_G} we have,  using lemma \ref{L3a} and the projection matrix \eqref{eq:Px}, 
\begin{align} 
\E\left( \left(-\mP_x x_t + a_\vbeta^0(\vu_t, \eta_t)\right) \mF_0^{-1} \left(-\mP_x x_t + a_\vbeta^0(\vu_t, \eta_t)\right)'\right) + o(1).\label{eq:Var_partialv_v} 
\end{align}
The leading term in \eqref{eq:Var_partialv_v} is $ \mP_x \mF_0^{-1} \mP_x' \E(x_t^2) $. It is finite for $b_0<1/2$ since $x_t$ is asymptotically stationary and so are all cross products from \eqref{eq:Var_partialv_v}. Thus, the covariance matrix is finite for $b_0<1/2$. Hence
	$
	\frac{1}{{n}} \sum_{t=1}^{n} \frac{\partial \vv_t(\vtheta)'}{\partial \vtheta}\big \rvert_{\vtheta = \vtheta_0} \mF_0^{{[n]}^{-1}}\vv_t(\vtheta_0)\vv_t(\vtheta_0)' \mF_0^{{[n]}^{-1}}  \frac{\partial \vv_t(\vtheta)}{\partial \vtheta'}\big \rvert_{\vtheta = \vtheta_0}
		 \pto \Var(\mG)
	$
	as $n \to \infty$, where $0 < \Var(\mG) < \infty$ and $\Var(\mG)$ results from \eqref{eq:Var_for_G}. The proof of \eqref{eq:L3:c:2} is identical to \citet[lemma 3.4]{ChaMilPa2009} except for the additional use of lemma \ref{L:1c}.
\end{proof}
\noindent
With these lemmas at hand, we are ready to prove theorem \ref{th:4}.
\begin{proof}[Proof of Theorem \ref{th:4}]
	 As noted in section \ref{Ch:3}, both terms in \eqref{eq:sn0} $\sum_{t=1}^n \left(\vv_t(\vtheta_0) \vv_t(\vtheta_0)' -  \mF_0^{[n]}\right)$ and $\sum_{t=1}^n \frac{\partial \vv_t(\vtheta)'}{\partial \vtheta}\Big\rvert_{\vtheta = \vtheta_0} \mF_0^{{[n]}^{-1}}\vv_t(\vtheta_0)$ are asymptotically independent. Therefore, it follows from lemmas \ref{L3:b} and \ref{L3:c} that
	$
		\frac{1}{\sqrt{n}}\vs_n(\vtheta_0) \dto \mG + \mJ \sim \mathrm{N}(0, \mathcal{J}_0)
	$
	as $n \to \infty$, with $\mathcal{J}_0 = \Var(G) + \Var(J)$ and each variance  given in lemma \ref{L3:c}. Using the results of \citet[Ch. 11.3.3]{Dav2000}, it follows for $b_0 < 0.5$ that 
	$
		\sqrt{n} (\hat{\vtheta}_n - \vtheta_0) \dto \mathrm{N}(0, \mathcal{J}_0^{-1})
$
	as $n \to \infty$.
\end{proof}

\begin{proof}[Proof of Lemma \ref{L3:d}]
First note that  from lemma \ref{L3:b} $\frac{\partial \vv_t(\vtheta)'}{\partial \vtheta}\Big\rvert_{\theta = \theta_0} \mF_0^{{[n]}^{-1}}\vv_t(\vtheta_0)$ is a martingale difference sequence adapted to the sigma-algebra $\mathcal{F}_t$. To prove weak convergence of $U_n(r)$, $W_n(r)$, $Y_n(r)$ observe that multiplication with $\mA_S'$ and $\mA_D'$ eliminates the nonstationary part, so that
\begin{align*}
	\mF_0^{{[n]}^{-1}} \vv_t(\vtheta_0), 
	\qquad A_S' \frac{\partial \vv_t(\vtheta)'}{\partial \vtheta}\Big\rvert_{\theta = \theta_0}\mF_0^{{[n]}^{-1}} \vv_t(\vtheta_0), 
	\qquad A_D' \frac{\partial \vv_t(\vtheta)'}{\partial \vtheta}\Big\rvert_{\theta = \theta_0} \mF_0^{{[n]}^{-1}}\vv_t(\vtheta_0),
\end{align*}
are (asymptotically) stationary martingale difference sequences. 
Therefore,  a functional central limit theorem for stationary martingale difference sequences  
\citep[cf.\ eg.][thm.\ 27.14]{Dav1994} implies  $(\mU_n(r), \mW_n(r), \mY_n(r)) \Rightarrow (\mU(r), \mW(r), \mY(r))$ as $n \to \infty$. 

\noindent
For the nonstationary, fractionally integrated $X_n(r)$ it follows from \eqref{gammax} that, $b_0>1/2$,
\begin{align}\label{Anv0}
	A_N' \frac{\partial \vv_t(\vtheta)'}{\partial \vtheta}\Big\rvert_{\theta = \theta_0} = - \mGamma_0' \Delta_+^{-b_0} \eta_t + \mGamma_0' a_\vbeta^0(u_t,\eta_t),
\end{align}
where $n^{-b_0+1/2} \Delta_+^{-b_0}\eta_t$ weakly converges  to fractional Brownian motion of type II \citep[cf.][eq. 6]{JohNie2010}, whereas for $b_0>1/2$ the $I(0)$ component $ \mGamma_0' a_\vbeta^0(u_t,\eta_t)$ in $X_n(r)$ converges to zero due to scaling. Hence, $X_n(r) \Rightarrow X(r)$ as $n \to \infty$.

\noindent
For $V_n$, it follows from \eqref{eq:partial_v_partial_beta} by plugging in $\vy_t$ and rearranging terms that the partial derivative $\partial \vv_t(\vtheta)'/\partial \vbeta \rvert_{\vtheta = \vtheta_0} = \mV_{x,t} + \mV_{\veta,t} + \mV_{u,t} + \mV_{B,t}$, where 
$\mV_{x,t} = -\mP_x x_t $, 
$\mV_{\veta, t} = \mP_x B_+(L, \vtheta_0) \veta_t$, 
$\mV_{B,t} = \partial B_+(L, \vtheta)/\partial \vbeta\rvert_{\vtheta=\vtheta_0} \left(\eta_t \vbeta_0' + (\vbeta_0' \mSigma_0^{-1})(\vbeta_0' \mSigma^{-1}_0\vbeta_0)^{-1} \Delta_+^{b_0} \vu_t\vbeta_0' \right)$, and
\begin{align*}\mV_{u,t} = \frac{-\vbeta_0' \mSigma_0^{-1}}{\vbeta_0' \mSigma_0^{-1} \vbeta_0} (1 - B_+(L, \vtheta_0) \Delta_+^{b_0}) \vu_t\mI  - \frac{\mSigma_0^{-1}}{\vbeta_0'\mSigma_0^{-1}\vbeta_0} \! \left(I - \frac{2\vbeta_0 \vbeta_0' \mSigma_0^{-1}}{\vbeta_0' \mSigma_0^{-1}\vbeta_0}\right) \! (1 - B_+(L, \vtheta_0)\Delta_+^{b_0}) \vu_t\vbeta_0'.
\end{align*}
Note that $\mGamma_0' \mV_{B,t }= 0$, which can be seen directly by plugging in the partial derivative of $B_+(L, \vtheta)$ as given in the proof of lemma \ref{L3a} and using \eqref{gamma:prop}. $\mV_{u,t}$ only depends on $\vu_1,...,\vu_{t-1}$, since $\mB_0 = \pi_0(b_0)=1$, which eliminates $\vu_t$ in $(1 - B_+(L, \vtheta_0)\Delta_+^{b_0})\vu_t$. Furthermore $\vv_t(\vtheta_0) = \mP_x(\vbeta_0 x_t + \vu_t) + \vbeta_0 z_t(\vtheta_0) = \mP_x \vu_t + \vbeta_0 z_t(\vtheta_0)$ only depends on contemporaneous $\vu_t, \eta_t$, since $z_t(\vtheta_0) = \veta_t + \vbeta_0'\mSigma_0^{-1}(\vbeta_0'\mSigma_0^{-1}\vbeta_0)^{-1}\vu_t$ is Gaussian white noise, as discussed in the proof of lemma \ref{lemma:z_t_order}. Finally, the relation $\mGamma_0' \mF_0^{[n]^{-1}} = \mGamma_0'\mSigma^{-1}_0$ will be helpful in proving convergence of $\mV_n$, and follows from plugging in $\mF^{[n]^{-1}}_0$ from lemma \ref{L:1} and using \eqref{gamma:prop}. For $\mV_n$ one then has
\begin{align*}
	\mV_n = \frac{1}{n^{b_0}}\sum_{t=1}^n \mGamma_0' (\mV_x + \mV_\veta + \mV_u) \mF_0^{[n]^{-1}}(\mP_x \vu_t + \vbeta_0 z_t (\vtheta_0)) =  -\frac{1}{n^{b_0}} \sum_{t=1}^n \mGamma_0' \mSigma_0^{-1} x_t u_t + o_p(1),
\end{align*}
since $\sum_{t=1}^n \mGamma_0' \mV_u \mF_0^{[n]^{-1}}(\mP_x \vu_t + \vbeta_0 z_t (\vtheta_0)) = O_p(n^{1/2})$ as $\mV_u$ is $I(0)$, depends on $\vu_1,...,\vu_{t-1}$ and $\vu_t, \eta_t$ are iid,
$\sum_{t=1}^n \mGamma_0' \mV_\eta \mF_0^{[n]^{-1}}(\mP_x \vu_t + \vbeta_0 z_t (\vtheta_0)) = \sum_{t=1}^n \mGamma_0' \mSigma_0^{-1}B_+(L, \vtheta_0) \veta_t(\mP_x \vu_t + \vbeta_0 z_t(\vtheta_0)) = \sum_{t=1}^n \mGamma_0' \mSigma_0^{-1}B_+(L, \vtheta_0) \veta_t \vu_t = O_p(n^{1/2})$, since $\mGamma_0' \mSigma_0^{-1}\vbeta_0 = 0$ and $\veta_t$, $\vu_t$ are independent. Finally, $n^{-b_0}\sum_{t=1}^n \mGamma_0' \mV_x \mF_0^{[n]^{-1}}(\mP_x \vu_t + \vbeta_0 z_t (\vtheta_0)) = 
n^{-b_0}\sum_{t=1}^n  - \mGamma_0' \mSigma_0^{-1} x_t \vu_t$. Since $\eta_t$, $\vu_t$ are independent, one can apply a central limit theorem for fractionally integrated processes \citep[cf.\ e.g.][eq.\ 7]{JohNie2010} and write $V_n \dto V= \int_{0}^{1} \mX(r) \rd \mU(r)$ as $n \to \infty$.

\end{proof}
\begin{lemma}\label{ML:1}
	For $b_0 \in (1/2, 3/2)$ and $\vnu_n^{-1}$ given in \eqref{eq:A_N_A_S_A_D_nu_n}, the score vector of the likelihood function for the fractional unobserved components model satisfies 
	\begin{align*}
		\vnu_n^{-1} \mA' \vs_n(\vtheta_0) \xrightarrow{d}\mN = \bvec - \int_{0}^{1} \mX(r)\rd \mU(r) \\ \mZ-\mW\\
		Q - Y\evec, \qquad \text{as } n \to \infty,
	\end{align*} 
	with $Z_n$, $Q_n$ given in \eqref{eq:Z}, \eqref{eq:Q}, $\mZ_n \dto \mZ \sim \mathrm{N}(0, \Var(Z))$, and $\mQ_n \dto \mQ \sim \mathrm{N}(0, \Var(Q))$, as $n \to \infty$. $\Var(Z), \Var(Q)$ are given in \eqref{VarZ}, \eqref{VarQ}.
\end{lemma}

\begin{proof}[Proof of Lemma \ref{ML:1}]
	Note that for the first block of $\vnu_n^{-1} \mA' \vs_n(\vtheta_0)$ one has for $b_0>1/2$
	\begin{align*}
		&n^{-b_0} \mA_N' \vs_n(\vtheta_0) = \frac{1}{2n^{b_0-1/2}}\mA_N' \frac{\partial (\vec \mF^{[n]})'}{\partial \vtheta}\Bigg\rvert_{\vtheta = \vtheta_0} \left( \mF_0^{{[n]}^{-1}} \otimes \mF_0^{{[n]}^{-1}} \right) \\
		&\times \vec\left[\frac{1}{\sqrt{n}} \sum_{t=1}^{n} \left(\vv_t(\vtheta_0)\vv_t(\vtheta_0)' - \mF_0^{{[n]}^{-1}} \right)\right]-\frac{1}{n^{b_0}}\mA_N'\sum_{t=1}^{n}\frac{\partial \vv_t(\vtheta)'}{\partial\vtheta}\Bigg\rvert_{\vtheta=\vtheta_0} \mF_0^{{[n]}^{-1}} \vv_t(\vtheta_0)=\\
		&=n^{-b_0+1/2}O_p(1) -\mV_n 
		\dto - \mV = - \int_{0}^{1} \mX(r)\rd \mU(r),
	\end{align*}
as $n \to \infty$ due to lemma \ref{L3:d} and \eqref{eq:wis}. Next, observe that for  $\mZ_n$ in \eqref{eq:Z},
additionally applying lemma \ref{L:1c}, one has
 $\mZ_n \dto \mZ \sim \mathrm{N}(0, \Var(\mZ))$, as $n \to \infty$, with $\Var(\mZ)$ given in \eqref{VarZ}, as \citet[lemma 3.4]{ChaMilPa2009} show. Since $\mQ_n$ in \eqref{eq:Q} only differs from $\mZ_n$ by its rotation matrix,  $\mQ_n \dto \mQ \sim \mathrm{N}(0, \Var(\mQ))$ follows analogously.
Using also the partial sums defined for lemma \ref{L3:d} one obtains for the second block 
$	\frac{1}{\sqrt{n}} \mA_S' \vs_n(\vtheta_0) 
	= \mZ_n - \mW_n   \dto \mZ - \mW,
$ as $n \to \infty$ 
and  analogously for the third block 
$\frac{1}{\sqrt{n}} \mA_D' \vs_n(\vtheta_0) = \mQ_n - \mY_n
\dto Q - Y$.
\end{proof}
\begin{lemma}\label{ML:2}
	The Hessian matrix satisfies 
	\begin{align*}
		-\vnu_n^{-1}\mA'\mH_n(\vtheta_0)\mA \vnu_n^{-1'} \xrightarrow{d} \mM > 0, \quad \text {a.s. as } n \to \infty,
	\end{align*}
	with $\mM$ given in \eqref{eq:M}
\end{lemma}
\begin{proof}[Proof of Lemma \ref{ML:2}]
By \eqref{eq:Hessian} we have $\vnu_n^{-1}\mA'\mH_n(\vtheta_0)\mA \vnu_n^{-1'} = \vnu_n^{-1}\mA' \left(\sum_{h=1}^8 \mH_{n,h}(\vtheta_0)\right) \mA \vnu_n^{-1'}$. 
	Starting with the upper-left block the decisive term stems from $H_{n,5}(\vtheta_0)$ such that
	\begin{align*}
		\frac{1}{n^{2b_0} }\mA_N' \mH_n(\vtheta_0)\mA_N &= \frac{-1}{n^{2b_0}}\mA_N' \left(\sum_{t=1}^{n}\frac{\partial \vv_t(\vtheta)'}{\partial\vtheta}\Bigg\rvert_{\vtheta = \vtheta_0} \mF_0^{{[n]}^{-1}} \frac{\partial \vv_t(\vtheta)}{\partial\vtheta'}\Bigg\rvert_{\vtheta = \vtheta_0} \right) \mA_N + o_p(1) \\
		& \dto - \int_{0}^{1} \mX(r) \mF_0^{-1} \mX(r)' \rd r ,
	\end{align*}
	as $n \to \infty$, where the nonstationary term converges due to lemma \ref{L3:d} and the continuous mapping theorem and where $o_p(1)$ accounts for the components in the Hessian matrix that converge to zero in probability.\\
	The upper-middle block is 
	$
	\frac{1}{n^{b_0+0.5}} \mA_N' \mH_n(\vtheta_0) \mA_S = O_p(n^{-1/2})
	$
	since for the components including fractionally integrated processes due to $H_{n,h}(\vtheta_0)$, $h=5,6,7,8$,
	\begin{align*}
		&\mA_N' \left(\sum_{t=1}^{n} \frac{\partial \vv_t(\vtheta)'}{\partial\vtheta}\Bigg\rvert_{\vtheta = \vtheta_0} \mF_0^{{[n]}^{-1}}\frac{\partial \vv_t(\vtheta)}{\partial\vtheta'}\Bigg\rvert_{\vtheta = \vtheta_0}\right) \mA_S = \sum_{t=1}^{n} -\mGamma_0'x_t\mF_0^{{[n]}^{-1}}\vw_t + \vw_t = O_p(n^{b_0}), \\
		&\sum_{t=1}^{n}\left( \mI \otimes \vv_t(\vtheta_0)'\mF_0^{{[n]}^{-1}}\right) \left(\frac{\partial^2}{\partial \vtheta \partial \vtheta'} \otimes \vv_t(\vtheta)\right)\Bigg\rvert_{\vtheta = \vtheta_0} = O_p(n^{b_0}), \\
		&\frac{\partial (\vec \mF^{[n]})'}{\partial \vtheta}\Bigg\rvert_{\vtheta=\vtheta_0}(\mF_0^{{[n]}^{-1}}\otimes \mF_0^{{[n]}^{-1}})\sum_{t=1}^{n}\left(\frac{\partial \vv_t(\vtheta)}{\partial \vtheta'}\Bigg\rvert_{\vtheta=\vtheta_0} \otimes \vv_t(\vtheta_0)\right)= O_p(n^{b_0}),\\
		&\sum_{t=1}^{n}\left(\frac{\partial \vv_t(\vtheta)'}{\partial \vtheta}\Bigg\rvert_{\vtheta=\vtheta_0} \otimes \vv_t(\vtheta_0)'\right)(\mF_0^{{[n]}^{-1}}\otimes \mF_0^{{[n]}^{-1}})\frac{\partial \vec \mF^{[n]}}{\partial \vtheta'}\Bigg\rvert_{\vtheta=\vtheta_0}= O_p(n^{b_0}).
	\end{align*}
	The center-middle block converges to
	$
	\frac{1}{n}\mA_S' \mH_n(\vtheta_0)\mA_S \pto - \Var(\mW) - \Var(\mZ),
	$
as shown in \citet[eq. 56--64]{ChaMilPa2009}. 
	\noindent
	For the last component $\frac{1}{n} \mA_D' \mH_n(\vtheta_0) \mA_D$, due to relevant $H_{n,h}(\vtheta_0)$, $h=3,5,6,7$, we define
\begin{align*}
		\frac{1}{n} {\mA}_D' \mH_n(\vtheta_0)\mA_D &= \mA_n^* + \mB_n^* + \mC_n^* + \mD_n^* + \mD_n^{*'} + o_p(1),
	\end{align*}
	\begin{align*}
		\mA^*_n &= -\frac{1}{2}\mA_D' \left[\frac{\partial (\vec \mF)'}{\partial \vtheta}\Bigg \rvert_{\vtheta = \vtheta_0} (\mF_0^{-1}\otimes \mF_0^{-1}) \frac{\partial \vec \mF}{\partial \vtheta'}\Bigg \rvert_{\vtheta = \vtheta_0}\right] \mA_D + o_p(1)= -\Var(\mQ) + o_p(1),  \\
		\mB^*_n &= -\frac{1}{n} \sum_{t=1}^{n} \mA_D' \left( \frac{\partial \vv_t(\vtheta)'}{\partial\vtheta}\Bigg\rvert_{\vtheta = \vtheta_0} \mF_0^{-1} \frac{\partial \vv_t(\vtheta)}{\partial\vtheta'}\Bigg\rvert_{\vtheta = \vtheta_0}  \right) \mA_D= -\Var(\mY) + o_p(1),  \\
		\mC^*_n &= -\frac{1}{n}\sum_{t=1}^{n} \mA_D' \left( \mI \otimes \vv_t(\vtheta_0)' \mF_0^{-1}\right)\left(\frac{\partial^2}{\partial \vtheta \partial \vtheta'} \otimes \vv_t(\vtheta)\right)\Bigg \rvert_{\vtheta = \vtheta_0}\mA_D=O_p(n^{-1/2}), \\
		\mD^*_n &= \mA_D' \left[\frac{\partial (\vec \mF)'}{\partial \vtheta}\Bigg \rvert_{\vtheta = \vtheta_0} (\mF_0^{-1}\otimes \mF_0^{-1}) \frac{1}{n}\sum_{t=1}^{n} \left( \frac{\partial \vv_t(\vtheta)}{\partial\vtheta'}\Bigg\rvert_{\vtheta = \vtheta_0}\otimes \vv_t(\vtheta_0) \right) \right] \mA_D=O_p(n^{-1/2}).
	\end{align*}
	
	\noindent
	The results for $A_n^*$, $B_n^*$ follow directly from lemma \ref{ML:1} and \eqref{var:Y}. The result for $\mD^{*}_n$ holds since $ \frac{\partial \vv_t(\vtheta)}{\partial b}\big\rvert_{\vtheta = \vtheta_0}$ is stationary and $\mathcal{F}_{t-1}$-measurable, as shown in the proof of lemma \ref{L3a}, such that $\frac{\partial \vv_t(\vtheta)}{\partial b}\big\rvert_{\vtheta = \vtheta_0} \otimes \vv_t(\vtheta_0)$ is a stationary MDS. For $C_n^*$ to hold we require stationarity of $\partial^2 \vv_t(\vtheta)/(\partial b \partial \vbeta')$, $\partial^2 \vv_t(\vtheta)/(\partial b \partial (\vec \mSigma)')$, and $\partial^2 \vv_t(\vtheta)/\partial b^2$ at $\vtheta=\vtheta_0$. Since $\partial \vv_t(\vtheta)/\partial b = (\partial z_t(\vtheta)/\partial b) \vbeta'$ equation \eqref{partial:z2} shows directly that the former two conditions hold, as the partial derivatives w.r.t. $\vbeta'$, $(\vec \mSigma)'$ do not change the persistence of the process. 
		
		\noindent
		For $\partial^2 \vv_t(\vtheta)/\partial b^2$ we decompose $(\partial^2 z_t(\vtheta)/\partial b^2)\rvert_{\vtheta=\vtheta_0} = (\vbeta_0 \mSigma_0^{-1}\vbeta_0)^{-1}\vbeta_0' \mSigma_0^{-1}(Z_1 + Z_2 + Z_3)$, where 
		$Z_1 = (B_+(L, \vtheta_0) (1+\vbeta_0' \mSigma^{-1}_0\vbeta_0)^{-1/2}\Delta_+^{b_0} -1)B_+(L, \vtheta_0) (1-L)\sum_{j=1}^{t-1}\sum_{k=1}^j (\frac{\partial }{\partial b} (b_0-b+k)^{-1}\pi_j(b - b_0 -1) \Delta_+^{b_0}\vy_{t-j})\rvert_{\vtheta=\vtheta_0} $,
		$Z_2 = (B_+(L, \vtheta_0)(1+\vbeta_0'\mSigma_0^{-1}\vbeta_0)^{-1/2}\Delta_+^{b_0}-1)\frac{\partial B_+(L, \vtheta)}{\partial b}  \big\rvert_{\vtheta=\vtheta_0}\sum_{j=1}^{t-1}j^{-1}\Delta_+^{b_0}\vy_{t-j} $, and
		$Z_3 = B_+(L, \vtheta_0)(1+\vbeta_0'\mSigma_0^{-1}\vbeta_0)^{-1/2}  (\frac{\partial}{\partial b}B_+(L, \vtheta)\Delta_+^b)\big\rvert_{\vtheta = \vtheta_0}\sum_{j=1}^{t-1}j^{-1} \Delta_+^{b_0}\vy_{t-j}$. The three different components are obtained by applying the product rule to the partial derivative of \eqref{partial:z2}. 
		
		\noindent
		$Z_2$ is stationary, since the stationary filter $\frac{\partial B_+(L, \vtheta)}{\partial b}\big\rvert_{\vtheta = \vtheta_0}$ applied to a stationary series yields a stationary process, see \eqref{partial:B}. $Z_3$ is stationary, since we can write $(\frac{\partial}{\partial b}B_+(L, \vtheta)\Delta_+^b)\big\rvert_{\vtheta = \vtheta_0} = (\frac{\partial}{\partial b}B_+(L, \vtheta)\Delta_+^{b-b_0})\big\rvert_{\vtheta = \vtheta_0} \Delta_+^{b_0} $ and $(\frac{\partial}{\partial b}B_+(L, \vtheta)\Delta_+^{b-b_0})\big\rvert_{\vtheta = \vtheta_0}$ is a stationary filter, as shown in the proof of lemma \ref{L3a}.
		
		\noindent
		For $Z_1$ it remains to be shown that $(1-L)\sum_{j=1}^{t-1}\sum_{k=1}^j (\frac{\partial }{\partial b} (b_0-b+k)^{-1}\pi_j(b - b_0 -1) \Delta_+^{b_0}\vy_{t-j})\rvert_{\vtheta=\vtheta_0} $ is stationary. From $((\partial/\partial b) (b_0 - b +k)^{-1} \pi_j(b-b_0-1) \Delta_+^{b_0}\vy_{t-j})\rvert_{\vtheta = \vtheta_0} = k^{-2} \pi_j(-1)\Delta_+^{b_0}\vy_{t-j} - (\partial \pi_j(b - b_0 - 1)/\partial b)\rvert_{\vtheta = \vtheta_0} k^{-1} \Delta_+^{b_0}\vy_{t-j}$ together with \eqref{partial:b} it follows 
		$
		(1-L)\sum_{j=1}^{t-1}\sum_{k=1}^j (\frac{\partial }{\partial b} (b_0-b+k)^{-1}\pi_j(b - b_0 -1) \Delta_+^{b_0}\vy_{t-j})\rvert_{\vtheta=\vtheta_0} = (1-L)\sum_{j=1}^{t-1}\sum_{k=1}^{j}{k^{-2}}\Delta_+^{b_0}\vy_{t-j} - (1-L)\sum_{j=1}^{t-1}\sum_{k=1}^{j}k^{-1} \sum_{l=1}^j l^{-1} \Delta_+^{b_0}\vy_{t-j}
		$. 
		The former term is $(1-L)\sum_{j=1}^{t-1}\sum_{k=1}^{j}{k^{-2}}\Delta_+^{b_0}\vy_{t-j}  = \sum_{j=1}^{t-1}j^{-2}\Delta_+^{b_0}\vy_{t-j}$, whereas the latter term is $(1-L)\sum_{j=1}^{t-1}\sum_{k=1}^{j}k^{-1} \sum_{l=1}^j l^{-1} \Delta_+^{b_0}\vy_{t-j} = 
		\sum_{j=1}^{t-1}j^{-2}\Delta_+^{b_0}\vy_{t-j} + 2 \sum_{j=2}^{t-1} \Delta_+^{b_0}\vy_{t-j} j^{-1} \sum_{k=1}^{j-1}k^{-1}$. Hence $(1-L)\sum_{j=1}^{t-1}\sum_{k=1}^{j}{k^{-2}}\Delta_+^{b_0}\vy_{t-j} - (1-L)\sum_{j=1}^{t-1}\sum_{k=1}^{j}k^{-1} \sum_{l=1}^j l^{-1} \Delta_+^{b_0}\vy_{t-j} = - 2 \sum_{j=2}^{t-1}\Delta_+^{b_0}\vy_{t-j} j^{-1} \sum_{k=1}^{j-1}k^{-1}$
		and therefore it is stationary. Thus, $Z_1$ is stationary, such that $(\frac{\partial^2}{\partial \vtheta \partial \vtheta'} \otimes \vv_t(\vtheta))\rvert_{\vtheta = \vtheta_0}A_D$ has finite second moments and is $\mathcal{F}_{t-1}$-measurable. Therefore, it follows directly that $A_D'( \mI \otimes \vv_t(\vtheta_0)' \mF_0^{-1})(\frac{\partial^2}{\partial \vtheta \partial \vtheta'} \otimes \vv_t(\vtheta))\rvert_{\vtheta = \vtheta_0}A_D$ is a stationary MDS, such that the result for $C_n^*$ holds.
\end{proof}

\begin{lemma}\label{ML:3}
	There exists a sequence of invertible normalization matrices $\vmu_n$ such that $\vmu_n \vnu_n^{-1} \to 0$ a.s. and \begin{align*}
		\sup_{\vtheta \in \mTheta_n} \lvert \lvert \vmu_n^{-1} \mA'\left( \mH_n(\vtheta) - \mH_n(\vtheta_0) \right)\mA \vmu_n^{-1'}\rvert \rvert \xrightarrow{p}0,
	\end{align*} 
	where $\mTheta_n=\left\{\vtheta \big| ||\vmu_n' \mA^{-1} (\vtheta - \vtheta_0) || \leq 1\right\}$ is a sequence of shrinking neighborhoods of $\vtheta_0$.
\end{lemma}
\begin{proof}[Proof of Lemma \ref{ML:3}]
First, determine all $\vtheta$'s that fulfill $\mTheta_n=\left\{\vtheta \big| ||\vmu_n' \mA^{-1} (\vtheta - \vtheta_0) || \leq 1\right\}$. 
Analogously to \citet[p.~245]{ChaMilPa2009} we let $\vmu_n = \vnu_n^{1-\gamma}$ for small $\gamma>0$. Further, 
denote the vector of rows $i$ to $j$ of a vector $\vdelta$ by $\vdelta^{(i:j)}$. All $\vtheta\in \mTheta_n$ are given by  those $\vdelta= \mu_n' \mA^{-1} (\vtheta-\vtheta_0)$ for which $|| \vdelta || \leq 1$ holds. Inverting delivers
 \begin{align}
 	\vbeta &= \vbeta_0 + n^{-b_0(1-\gamma)} \mGamma_0\, \vdelta^{(1:p-1)} +  n^{-1/2(1-\gamma)} \frac{\vbeta_0}{(\vbeta_0'\mSigma_0^{-1}\vbeta_0)^{1/2} } \vdelta^{(p)}, \label{eq:b0}\\
 	\vech\mSigma &= \vech\mSigma_0 + n^{-1/2(1-\gamma)}\vdelta^{(p+1:k-1)}, \label{eq:s0}\\
    b &= b_0 + n^{-1/2(1-\gamma)} \vdelta^{(k)}.  \label{eq:g0} 
 \end{align} 
By the properties of the projection matrix $\mP_x$, multiplication of \eqref{eq:b0} by $\mGamma_0' \mSigma_0^{-1}$ and $\frac{\vbeta_0' \mSigma_0^{-1}}{(\vbeta'_0 \mSigma_0^{-1} \vbeta_0)^{1/2}}$ delivers
 \begin{align*}
 	\mGamma_0' \mSigma_0^{-1} (\vbeta - \vbeta_0) &= O\left(n^{-b_0(1-\gamma)} \right), \qquad
 	\frac{\vbeta_0' \mSigma_0^{-1}}{(\vbeta'_0 \mSigma_0^{-1} \vbeta_0)^{1/2}} (\vbeta - \vbeta_0) = O\left(n^{-1/2(1-\gamma)}\right).
 \end{align*}
In \eqref{eq:b0} to \eqref{eq:g0} $\vbeta$, $\mSigma$, and $b$ are marginally smaller or larger than their true values depending on the sign of the elements of $\vdelta$. Note that the sign of $\vdelta^{(k)}$ matters in \eqref{eq:g0}.  Choosing $b \geq b_0$ gives $(\Delta_+^b - \Delta_+^{b_0}) \vy_t \sim I(0)$, whereas $b < b_0$ yields $(\Delta_+^b - \Delta_+^{b_0}) \vy_t \sim I(b_0 - b)$. The latter case is implied by $\vdelta^{(k)}<0$. Thus, $b = b_0 - n^{-1/2(1-\gamma)} |\vdelta^{(k)}|$ covers the more general case and is considered in the following. The results carry over to $b > b_0$ straightforwardly.

\noindent
 For lemma \ref{ML:3} to be satisfied, for the nonstationary components in \eqref{eq:Hessian} involving $H_{n,h}(\vtheta_0)$, $h=5,6,7$, we need to show that 
 \begingroup
 \allowdisplaybreaks
  	{\small
 \begin{align}
 	&\frac{1}{n^{2b_0(1-\gamma)}}\mA_N'  \sum_{t=1}^{n} \left(\frac{\partial \vv_t(\vtheta)'}{\partial \vtheta} - \frac{\partial \vv_t(\vtheta)'}{\partial \vtheta}\Bigg\rvert_{\vtheta = \vtheta_0} \right) \mF_0^{{[n]}^{-1}}\frac{\partial \vv_t(\vtheta)}{\partial \vtheta'}\Bigg\rvert_{\vtheta = \vtheta_0}   \mA_N \pto 0, \label{ML:3a} \\
 	&\frac{1}{n^{2b_0(1-\gamma)}}\mA_N'  \sum_{t=1}^{n} \left(\frac{\partial \vv_t(\vtheta)'}{\partial \vtheta} - \frac{\partial \vv_t(\vtheta)'}{\partial \vtheta}\Bigg\rvert_{\vtheta = \vtheta_0} \right) \mF_0^{{[n]}^{-1}}\left(\frac{\partial \vv_t(\vtheta)}{\partial \vtheta'} - \frac{\partial \vv_t(\vtheta)}{\partial \vtheta'}\Bigg\rvert_{\vtheta = \vtheta_0} \right)  \mA_N \pto 0, \label{ML:3b} \\
 	&\frac{1}{n^{1+\gamma}}\mA_j'  \sum_{t=1}^{n} \left(\frac{\partial \vv_t(\vtheta)'}{\partial \vtheta} - \frac{\partial \vv_t(\vtheta)'}{\partial \vtheta}\Bigg\rvert_{\vtheta = \vtheta_0} \right) \mF_0^{{[n]}^{-1}}\frac{\partial \vv_t(\vtheta)}{\partial \vtheta'}\Bigg\rvert_{\vtheta = \vtheta_0}   \mA_j \pto 0, \label{ML:3c} \\
 	&\frac{1}{n^{1+\gamma}}\sum_{t=1}^{n}\mA_j'   \left(\mI \otimes (\vv_t(\vtheta)' - \vv_t(\vtheta_0)') \mF_0^{{[n]}^{-1}}\right) \left(\frac{\partial^2 }{\partial \vtheta \partial \vtheta'} \otimes \vv_t(\vtheta)\right)\Bigg\rvert_{\vtheta = \vtheta_0}   \mA_j \pto 0, \label{ML:3d} \\
 	&\frac{1}{n^{1+\gamma}}\sum_{t=1}^{n}\mA_j'   \left(\mI \otimes \vv_t(\vtheta_0)'\mF_0^{{[n]}^{-1}}\right) \left[\left(\frac{\partial^2 }{\partial \vtheta \partial \vtheta'} \otimes \vv_t(\vtheta)\right)-\left(\frac{\partial^2 }{\partial \vtheta \partial \vtheta'} \otimes \vv_t(\vtheta)\right)\Bigg\rvert_{\vtheta = \vtheta_0} \right]  \mA_j \pto 0, \label{ML:3e} \\
 	&\mA_j'  \frac{\partial (\vec \mF^{[n]})'}{\partial \vtheta}\Bigg\rvert_{\theta = \theta_0} (\mF_0^{{[n]}^{-1}} \otimes \mF_0^{{[n]}^{-1}}) \frac{1}{n^{1+\gamma}} \sum_{t=1}^{n}\left( \frac{\partial \vv_t(\vtheta)}{\partial \vtheta'} - \frac{\partial \vv_t(\vtheta)}{\partial \vtheta'}\Bigg\rvert_{\vtheta = \vtheta_0} \right) \otimes \vv_t(\vtheta_0) \mA_j \pto 0, \label{ML:3f} \\
 	&\mA_j'  \frac{\partial (\vec \mF^{[n]})'}{\partial \vtheta}\Bigg\rvert_{\theta = \theta_0} (\mF_0^{{[n]}^{-1}} \otimes \mF_0^{{[n]}^{-1}}) \frac{1}{n^{1+\gamma}} \sum_{t=1}^{n}\frac{\partial \vv_t(\vtheta)}{\partial \vtheta'}\Bigg\rvert_{\vtheta = \vtheta_0} \otimes \left( \vv_t(\vtheta) - \vv_t(\vtheta_0)\right) \mA_j \pto 0, \label{ML:3g} \\
 	&\frac{1}{n^{1+\gamma}}\mA_j'  \sum_{t=1}^{n} \left(\frac{\partial \vv_t(\vtheta)'}{\partial \vtheta} - \frac{\partial \vv_t(\vtheta)'}{\partial \vtheta}\Bigg\rvert_{\vtheta = \vtheta_0} \right) \mF_0^{{[n]}^{-1}}\left(\frac{\partial \vv_t(\vtheta)}{\partial \vtheta'} - \frac{\partial \vv_t(\vtheta)}{\partial \vtheta'}\Bigg\rvert_{\vtheta = \vtheta_0} \right)   \mA_j \pto 0, \label{ML:3h} \\
 	&\frac{1}{n^{1+\gamma}}\sum_{t=1}^{n}\mA_j'  \left[\mI \otimes (\vv_t(\vtheta)' - \vv_t(\vtheta_0)')\mF_0^{{[n]}^{-1}}\right] \left[\left(\frac{\partial^2 }{\partial \vtheta \partial \vtheta'} \otimes \vv_t(\vtheta)\right)-\left(\frac{\partial^2 }{\partial \vtheta \partial \vtheta'} \otimes \vv_t(\vtheta)\right)\Bigg\rvert_{\vtheta = \vtheta_0} \right]   \mA_j \pto 0, \label{ML:3i} \\
 	&\mA_j'  \frac{\partial (\vec \mF^{[n]})'}{\partial \vtheta}\Bigg\rvert_{\theta = \theta_0} (\mF_0^{{[n]}^{-1}} \otimes \mF_0^{{[n]}^{-1}})\frac{1}{n^{1+\gamma}} \sum_{t=1}^{n}\left( \frac{\partial \vv_t(\vtheta)}{\partial \vtheta'} - \frac{\partial \vv_t(\vtheta)}{\partial \vtheta'}\Bigg\rvert_{\vtheta = \vtheta_0} \right) \otimes (\vv_t(\vtheta) - \vv_t(\vtheta_0)) \mA_j \pto 0, \label{ML:3j} 
 \end{align}
 }%
 \endgroup
 for $j = S, D$. Analog to \cite{ChaMilPa2009} we only prove convergence of the nonstationary components since the required conditions obviously hold for the stationary terms. Let $\Delta(n^Kx_t)$ denote terms that are of order $n^K$ times $x_t$ or of a lower order. $\tilde{w}_t$ denotes terms that converge from an $I\left(n^{-1/2(1-\gamma)} \right)$ process  to an $I(0)$ process as $n \to \infty$. For the analysis of the convergence rates of the various differences above, $\vbeta_0$ can be rewritten based on \eqref{eq:b0}  as $\vbeta_0 = \vbeta - n^{-b_0(1-\gamma)} \mGamma_0\, \vdelta^{(1:p-1)} -  n^{-1/2(1-\gamma)} \frac{\vbeta_0}{(\vbeta_0'\mSigma_0^{-1}\vbeta_0)^{1/2} } \vdelta^{(p)}$.  To obtain the required convergence rates, iterate this equation by inserting it again for the $\vbeta_0$ in the numerator in the third term.  By denoting $g=  \vdelta^{(p)} (\vbeta_0'\mSigma_0^{-1}\vbeta_0)^{-1/2} $ this leads to
 \begin{align} \label{eq:b0_higher}
 \vbeta_0 & = \vbeta \left(1- n^{-1/2(1-\gamma)} g \right)  - 
 \left(n^{-b_0(1-\gamma)} - n^{-(1/2+b_0)(1-\gamma)}  g \right) 
 \mGamma_0 \vdelta^{(1:p-1)}    + n^{-1+\gamma}g^2 \vbeta_0.  
 \end{align}
Consider the difference $\vv_t(\vtheta) - \vv_t(\vtheta_0)$ first. From theorem \ref{th:1} and \eqref{vbar} and by denoting $\mP = \mI - \frac{\mSigma^{-1} \vbeta\vbeta' }{\vbeta'\mSigma^{-1}\vbeta}$ analogously to \eqref{eq:Px}, one has $\vv_t(\vtheta) - \vv_t(\vtheta_0) = \mP' \vbeta_0 x_t +  B_+(L, \vtheta) \left( \mI - \mP'\right) \Delta_+^b \vy_t + w_t$, where $\vw_t$ denotes some $I(0)$ terms.  Note that $\Delta^b_+ y_t$ in the second term is $I(b-b_0)=I\left(n^{-1/2(1-\gamma)}\right)$ by lemma \ref{lemma:z_t_order} and \eqref{eq:g0} and therefore abbreviated by $\tilde{w}_t$. Since $\mP'\vbeta=0$ and  $b_0>1/2$, inserting \eqref{eq:b0_higher} for $\vbeta_0$ in the first term delivers 
\begin{align} \label{eq:vt0}
\vv_t(\vtheta) - \vv_t(\vtheta_0) & = \Delta\left(n^{-b_0(1-\gamma)} x_t \right) + \Delta\left( n^{-1+\gamma} x_t \right) + \tilde{w}_t + w_t.
\end{align}
For considering the differences in the partial derivatives we start from \eqref{eq:partial_v_partial_beta} derived in the proof of lemma \ref{L3a} and focus on terms driven by $x_t$ 
\begin{align} \nonumber
	\frac{\partial \vv_t(\vtheta)'}{\partial \vbeta} 
	& = - \left[  \frac{\vbeta' \mSigma^{-1} }{\vbeta' \mSigma^{-1}\vbeta } \vbeta_0 \mI + \frac{\mSigma^{-1} }{\vbeta' \mSigma^{-1}\vbeta } \left(I - 2 \frac{\vbeta \vbeta' \mSigma^{-1} }{\vbeta' \mSigma^{-1}\vbeta }  \right) \vbeta_0 \vbeta' \right] x_t + \tilde{\vw}_t + \vw_t. \nonumber 
\intertext{Next insert $\vbeta_0$ from \eqref{eq:b0_higher} and collect terms such that}
	\frac{\partial \vv_t(\vtheta)'}{\partial \vbeta}	& = -  \mP x_t  + \Delta\left( n^{-b_0(1-\gamma)} x_t \right) + \tilde{\vw}_t + \vw_t, \label{eq:partial_v_partial_beta2} \\
	\frac{\partial\vv_t(\vtheta)'}{\partial \vbeta} - \frac{\partial \vv_t(\vtheta)'}{\partial \vbeta}\Bigg \rvert_{\vtheta = \vtheta_0}&=  -\left( \mP - \mP_x \right)x_t + \Delta\left( n^{-b_0(1-\gamma)} x_t \right) +\tilde{w}_t + w_t.  \label{deriv}
\end{align}
Based on \eqref{eq:b0_higher}, one can show that $\mP - \mP_x = O\left( n^{-1/2(1-\gamma)} \right)$ so that \eqref{deriv} can be also written as $\frac{\partial\vv_t(\vtheta)'}{\partial \vbeta} - \frac{\partial \vv_t(\vtheta)'}{\partial \vbeta}\Big \rvert_{\vtheta = \vtheta_0} = \Delta\left( n^{-1/2(1-\gamma)} x_t \right)  + \tilde{\vw}_t + \vw_t$, which directly yields $\frac{\partial\vv_t(\vtheta)'}{\partial \vtheta} - \frac{\partial \vv_t(\vtheta)'}{\partial \vtheta}\Big \rvert_{\vtheta = \vtheta_0} = \Delta\left( n^{-1/2(1-\gamma)} x_t \right)  + \tilde{\vw}_t + \vw_t$. From this result, it follows for the second order partial derivatives that 
\begin{align} \label{deriv:1}
	\frac{\partial^2}{\partial \vtheta \partial \vtheta'}\otimes \vv_t(\vtheta) - \frac{\partial^2}{\partial \vtheta \partial \vtheta'}\otimes \vv_t(\vtheta)\Bigg\rvert_{\vtheta = \vtheta_0} = \Delta(n^{-1/2+\gamma}x_t) + \tilde{w}_t + w_t.
\end{align}
Now we are ready for checking  \eqref{ML:3a} to \eqref{ML:3j}. We begin with \eqref{ML:3a}. Using \eqref{deriv}, the above result on $\mP-\mP_x$, \eqref{eq:A_N_A_S_A_D_nu_n}, and lemma \ref{L3a}, the leading term in \eqref{ML:3a} can be stated 
 \begin{align*}
& n^{2\gamma b_0} \mGamma_0' O\left( n^{-1/2(1-\gamma)} \right) \left(  \frac{1}{n^{2b_0} } \sum_{t=1}^n x_t^2\right) \mF_0^{{[n]}^{-1}} \mP_x' \mGamma_0 = O\left(n^{-1/2 + \gamma(1/2 + 2b_0)} \right) O_p(1) = o_p(1)
\end{align*} 
for small $\gamma$ and where $\frac{1}{n^{2b_0}}\sum_{t=1}^{n}x_t^2 = O_p(1)$ can be shown.
Similarly, \eqref{ML:3b} can be derived.

\noindent
For the remaining equations note that $\frac{\partial \vv_t(\vtheta)'}{\partial \vec \mSigma} = \Delta(n^{-b_0+\gamma}x_t) +\tilde{w}_t+ w_t$ which can be seen directly by plugging \eqref{eq:b0} into the formula for the partial derivative as given in the proof of lemma \ref{L3a}. Furthermore $\frac{\partial \vv_t(\vtheta)'}{\partial b} = \tilde{w}_t+ w_t$. Therefore, from \eqref{eq:partial_v_partial_beta2}, \eqref{eq:A_N_A_S_A_D_nu_n} and by inserting \eqref{eq:b0_higher} for $\vbeta_0$ in the numerator and the properties of the projection matrix $\mP$ one obtains
\begin{align}\label{eq:vtbeta}
\mA_S' \frac{\partial \vv_t(\vtheta)'}{\partial \vbeta}	= \frac{\vbeta_0'}{(\vbeta_0'\mSigma_0^{-1}\vbeta_0)^{1/2}} \frac{\partial \vv_t(\vtheta)'}{\partial \vbeta}	= \Delta\left( n^{-b_0(1-\gamma)} x_t \right) + \tilde{\vw}_t + \vw_t.
\end{align}
The same holds for the second order partial derivatives.  Since $\frac{\partial \vv_t(\vtheta)}{\partial \vtheta'}\big\rvert_{\theta = \theta_0}A_S \sim I(0)$, and since $\frac{\partial \vv_t(\vtheta)}{\partial \vtheta'}\big\rvert_{\theta = \theta_0}A_D$ converges to an $I(0)$ process, as $A_D$ only picks the partial derivative w.r.t.\ $b$, equations \eqref{ML:3c}, \eqref{ML:3e}, and \eqref{ML:3h} follow directly. 
The equations  \eqref{ML:3d}, \eqref{ML:3i} can be shown by using \eqref{eq:vt0}.
 To prove \eqref{ML:3f}, \eqref{ML:3g} and \eqref{ML:3j}, note that
 \begin{align*}
 	\left(\frac{\partial \vv_t(\vtheta)}{\partial \vtheta'} \otimes \vv_t(\vtheta)\right)A_S = \left(\frac{\partial \vv_t(\vtheta)}{\partial \vtheta'} \otimes \vv_t(\vtheta)\right)(A_S \otimes 1) = \frac{\partial \vv_t(\vtheta)}{\partial \vtheta'}A_S \otimes \vv_t(\vtheta).
 \end{align*}
\eqref{ML:3g} the follows from $\frac{\partial \vv_t (\vtheta)}{\partial \vtheta'}\big\rvert_{\theta = \theta_0} A_S = w_t$ and $\frac{\partial \vv_t (\vtheta)}{\partial \vtheta'}\big\rvert_{\theta = \theta_0} A_D = \tilde{w}_t$ together with \eqref{eq:vt0}, whereas \eqref{ML:3f} follows from \eqref{eq:vtbeta} together with $\vv_t(\vtheta_0)=w_t$. Finally, \eqref{ML:3j} follows from \eqref{eq:vtbeta} together with \eqref{eq:vt0}. 
\noindent
This completes the proof for theorem \ref{ML:3}.
\end{proof}

\clearpage
\begin{spacing}{1}
\bibliographystyle{dcu}
\bibliography{literatur.bib}
\end{spacing}
\end{document}